\documentclass{lmcs} 
\pdfoutput=1

\usepackage{lastpage}
\lmcsdoi{15}{1}{6}
\lmcsheading{}{\pageref{LastPage}}{}{}%
{Feb.~22,~2018}{Jan.~29,~2019}{}


\usepackage{hyperref}
\usepackage{comment}
\usepackage{enumitem}
\usepackage{amsmath}
\usepackage{amssymb}
\usepackage{amsthm}
\usepackage{mathrsfs}
\usepackage{mathdots}
\usepackage{cmll}
\usepackage{proof}
\usepackage{tikz}
	\usetikzlibrary{arrows,shapes,positioning}
	\usetikzlibrary{calc,decorations.markings}
\usepackage{float}
	\floatstyle{boxed} 
	\restylefloat{figure}
\usepackage{booktabs}
\theoremstyle{plain} 

\newcommand{\seq}{\vec}
\newcommand{\st}{:}
\newcommand{\nat}{\mathbb{N}}
\newcommand{\con}[2]{{#1}.{#2}}
\newcommand{\concat}[2]{{#1}\star{#2}}

\newcommand{\subt}{\restr}
\newcommand{\restr}[2]{#1\hspace{-4pt}\upharpoonright_{\!#2}}

\newcommand{\Trees}[1][\infty]{\mathbb{T}^{#1}_{\textrm{\sf rec}}}

\newcommand{\emptyseq}{\varepsilon}
\newcommand{\seqof}[1]{\langle{#1}\rangle}

\newcommand{\BTle}{\le_\bot}

\newcommand{\Var}{\mathrm{Var}}
\newcommand{\Lam}{\ensuremath{\Lambda}}
\newcommand{\Lamo}{\ensuremath{\Lambda^{\!o}}}
\newcommand{\lam}{\ensuremath{\lambda}}
\newcommand{\One}{\mathtt{1}}
\newcommand{\bI}{\mathtt{I}}
\newcommand{\bP}{\mathtt{P}}
\newcommand{\bB}{\mathtt{B}}
\newcommand{\bY}{\mathtt{Y}}
\newcommand{\bK}{\mathtt{K}}
\newcommand{\bF}{\mathtt{F}}
\newcommand{\bE}{\mathtt{E}}
\newcommand{\etamax}[1][\iota]{\mathtt{\Xi}_{#1}}
\newcommand{\bJ}{\mathtt{J}}
\newcommand{\ifz}{\mathsf{ifz}}
\newcommand{\Church}[1]{\mathrm{c}_{#1}}
\newcommand{\nak}[1]{\lfloor#1\rfloor}
\newcommand{\Eq}{\mathsf{Eq}}
\newcommand{\succc}{\mathsf{succ}}
\newcommand{\pred}{\mathsf{pred}}
\newcommand{\BT}[1]{\mathrm{BT}(#1)}
\newcommand{\FV}[1]{\mathsf{FV}(#1)}
\newcommand{\subst}[2]{[#2/#1]}
\newcommand*{\redto}[1][]{\rightarrow_{#1}}
\newcommand*{\msto}[1][\beta]{\twoheadrightarrow_{#1}}
\newcommand*{\invmsto}[1][\beta]{\ {}_{#1}\!\!\twoheadleftarrow}
\newcommand{\NF}[1][]{\mathtt{NF}_{#1}}
\newcommand{\HNF}{\mathtt{HNF}}
\newcommand{\hole}[1]{[#1]}
\newcommand{\nlem}{\not\obsle[]}
\newcommand{\bU}{\mathtt{U}}

\newcommand{\BTset}{\mathbb{BT}}

\newcommand{\dom}{\mathrm{dom}}
\newcommand{\Stream}[1]{\langle\!\!\,\langle #1\rangle\!\!\,\rangle}
\newcommand{\ID}{[\bI]_{n\in\nat}}
\newcommand{\IDX}{\Stream{\bI}}
\newcommand{\ONE}{\Stream{\One}}
\newcommand{\JAY}{\Stream{\bJ}}
\newcommand{\IDo}{\Stream{\bI}^\Omega}
\newcommand{\ETAge}{\Stream{\mathbf{\One}^*}}
\newcommand{\ETA}{[\mathbf{\eta}_n]_{n\in\nat}}

\newcommand{\ETAo}{\Stream{\mathbf{\eta}}^\Omega}
\newcommand{\ETAset}[1][]{\cI^\eta_{#1}}
\newcommand{\Tuple}[1]{\langle#1\rangle}
\newcommand{\Tupler}{\bP}

\newcommand{\code}[1]{\lceil#1\rceil}

\newcommand{\num}[1]{\#(#1)}
\newcommand{\phnf}[1]{\mathrm{phnf}(#1)}
\newcommand{\fresh}[1]{#1^{\#}}
\renewcommand{\fresh}[1]{\underline #1}
\newcommand{\mes}[1]{|#1|}

\newcommand{\nf}{\mathtt{nf}}

\newcommand{\blam}{\boldsymbol{\lambda}}

\newcommand{\BTth}{\mathcal{B}}
\newcommand{\BTe}{{\mathcal{B}\eta}}
\newcommand{\BTo}{{\mathcal{B}\omega}}
\newcommand{\Hpl}{\mathcal{H^+}}
\newcommand{\Hst}{\mathcal{H^*}}

\newcommand\sqle{\sqsubseteq}
\newcommand{\leeta}[1][]{\le^\eta_{#1}}
\newcommand{\geeta}[1][]{\ge^\eta_{#1}}
\newcommand{\leetainf}{\le^{\eta}_{\omega}}
\newcommand{\geetainf}{\ge^{\eta}_{\omega}}
\newcommand{\obseq}[1][\mathtt{nf}]{\equiv^{#1}}
\newcommand{\obsle}[1][\mathtt{nf}]{\sqle^{#1}}

\newcommand{\bnf}{\,::=\,}

\newcommand{\cB}{\mathcal{B}}

\newcommand{\cD}{\mathcal{D}}

\newcommand{\cH}{\mathcal{H}}
\newcommand{\cI}{\mathcal{I}}

\newcommand{\cO}{\mathcal{O}}

\newcommand{\cT}{\mathcal{T}}

\begin{document}

\title[Degrees of extensionality in the theory of {B}\"ohm trees]{Degrees of extensionality in the theory of {B}\"ohm trees and
  {S}all\'e's conjecture.}

\author[B.~Intrigila]{Benedetto Intrigila}	
\address{Dipartimento di Ingegneria dell'Impresa, Universit\`a di Roma Tor Vergata, Rome, Italy}	
\email{intrigil@mat.uniroma2.it}  

\author[G.~Manzonetto]{Giulio Manzonetto}	
\address{LIPN, UMR 7030, Universit\'e Paris 13, Sorbonne Paris Cit\'e, F-93430, Villetaneuse, France}	
\email{giulio.manzonetto@lipn.univ-paris3.fr}  
\author[A.~Polonsky]{Andrew Polonsky}	
\address{Department of Computer Science, Appalachian State University, Boone, NC 28608, USA}	
\email{andrew.polonsky@gmail.com}  





\begin{abstract}
  \noindent 
The main observational equivalences of the untyped \lam-calculus have been characterized in terms of extensional equalities between B\"ohm~trees.
It is well known that the \lam-theory $\Hst$, arising by taking as observables the head normal forms, equates two \lam-terms whenever their B\"ohm trees are equal up to countably many possibly infinite $\eta$-expansions.
Similarly, two \lam-terms are equal in Morris's original observational theory~$\Hpl$, generated by considering as observable the $\beta$-normal forms,  whenever their B\"ohm trees are equal up to countably many finite $\eta$-expansions.

The \lam-calculus also possesses a strong notion of extensionality called \emph{the $\omega$-rule}, which has been the subject of many investigations.
It is a longstanding open problem whether the equivalence $\BTo$ obtained by closing the theory of B\"ohm trees under the $\omega$-rule is strictly included in $\Hpl$, as conjectured by Sall\'e in the seventies.
In this paper we demonstrate that the two aforementioned theories actually coincide, thus disproving Sall\'e's conjecture.

The inclusion $\BTo\subseteq \Hpl$ is a consequence of the fact that the \lam-theory $\Hpl$ satisfies the $\omega$-rule.
This result follows from a weak form of separability enjoyed by \lam-terms whose B\"ohm trees only differ because of some infinite $\eta$-expansions, a property which is proved through a refined version of the famous B\"ohm-out technique.

The inclusion $\Hpl\subseteq\BTo$ follows from the fact that whenever two \lam-terms are observationally equivalent in $\Hpl$ their B\"ohm trees have a common ``$\eta$-supremum'' that can be \lam-defined starting from a stream (infinite sequence) of $\eta$-expansions of the identity.
It turns out that in $\BTo$ such a stream is equal to the stream containing infinitely many copies of the identity, a peculiar property that actually makes the two theories collapse.

The proof technique we develop for proving the latter inclusion is general enough to provide as a byproduct a new characterization, based on bounded $\eta$-expansions, of the least extensional equality between B\"ohm~trees.
Together, these results provide a taxonomy of the different degrees of extensionality in the theory of B\"ohm trees.

\end{abstract}

\maketitle

\section*{Introduction}
The problem of determining when two programs are equivalent is central in computer science.
For instance, it is necessary to verify that the optimizations performed by a compiler actually preserve the meaning of the program. 
For \lam-calculi, it has become standard to consider two \lam-terms $M$ and $N$ as equivalent when they are \emph{contextually equivalent} with respect to some fixed set $\cO$ of observables~\cite{Morristh}. 
This means that it is possible to plug either $M$ or $N$ into any context $C[]$ without noticing any difference in the global behaviour: $C[M]$ produces a result belonging to $\cO$ exactly when $C[N]$ does.
The problem of working with this definition is that the quantification over all possible contexts is difficult to handle. Therefore, many researchers undertook a quest for characterizing observational equivalences both semantically, by defining fully abstract denotational models, and syntactically, by comparing possibly infinite trees representing possible program executions.

The most famous observational equivalence for the untyped \lam-calculus is obtained by considering as observables the head normal forms, which are \lam-terms representing stable amounts of information coming out of the computation.
Introduced by Hyland~\cite{Hyland76} and Wadsworth~\cite{Wadsworth76}, it has been ubiquitously studied in the literature~\cite{Bare,GouyTh,GianantonioFH99,RonchiP04,Manzonetto09,Breuvart14}, since it enjoys many interesting properties.
By definition, it corresponds to the extensional \lam-theory~$\Hst$ which is the greatest consistent sensible \lam-theory~\cite[Thm.~16.2.6]{Bare}.
Semantically, it arises as the \lam-theory of Scott's pioneering model $\cD_\infty$~\cite{Scott72}, a result which first appeared in~\cite{Hyland76} and~\cite{Wadsworth76}, independently.
More recently, Breuvart provided in~\cite{Breuvart14} a characterization of all $K$-models that are fully abstract for $\Hst$.
As shown in \cite[Thm.~16.2.7]{Bare}, two \lam-terms are equivalent in $\Hst$ exactly when their B\"ohm trees are equal up to countably many possibly infinite $\eta$-expansions.

However, the head normal forms are not the only reasonable choice of observables.
For instance, the original extensional contextual equivalence defined by Morris in~\cite{Morristh} arises by considering as observables the $\beta$-normal forms, which represent completely defined results.
We denote by $\Hpl$ the \lam-theory corresponding to Morris's observational equivalence\footnote{The notation $\Hpl$ has been introduced in~\cite{ManzonettoR14}, while the same theory is denoted $\mathscr{T}_{\mathrm{NF}}$ in~\cite{Bare} and $\bold{N}$ in~\cite{RonchiP04}.}.
The \lam-theory $\Hpl$ is sensible and distinct from $\Hst$, so we have $\Hpl\subsetneq\Hst$.
Despite the fact that the equality in $\Hpl$ has been the subject of fewer investigations, it has been characterized both semantically and syntactically.
In~\cite{CoppoDZ87}, Coppo \emph{et al.} proved that $\Hpl$ corresponds to the \lam-theory induced by a suitable filter model.
More recently, Manzonetto and Ruoppolo introduced a simpler model of $\Hpl$ living in the relational semantics~\cite{ManzonettoR14} and Breuvart \emph{et al.} provided necessary and sufficient conditions for a relational model to be fully abstract for~$\Hpl$~\cite{BreuvartMPR16}.
From a syntactic perspective, Hyland proved in \cite{Hyland75} that two \lam-terms are equivalent in $\Hpl$ exactly when their B\"ohm trees are equal up to countably many $\eta$-expansions of finite depth (see also~\cite[\S11.2]{RonchiP04} and \cite{Levy78}).

We have seen that both observational equivalences correspond to some extensional equalities between B\"ohm trees.
A natural question is whether $\Hpl$ can be generated just by adding $\eta$-conversion to the \lam-theory $\BTth$ equating all \lam-terms having the same B\"ohm tree.
The \lam-theory $\BTe$ so defined  has been little studied in the literature, probably because it does not arise as an observational equivalence nor is induced by some known denotational model. 
In~\cite[Lemma~16.4.3]{Bare}, Barendregt shows that one $\eta$-expansion in a \lam-term $M$ can generate infinitely many finite $\eta$-expansions on its B\"ohm tree~$\BT{M}$.
In~\cite[Lemma~16.4.4]{Bare}, he exhibits two \lam-terms that are equal in $\Hpl$ but distinct in $\BTe$, thus proving that~$\BTe\subsetneq\Hpl$.

However, the \lam-calculus also possesses another notion of extensionality, known as the \emph{$\omega$-rule}, which is strictly stronger than $\eta$-conversion.
Such a rule has been studied by many researchers in connection with several \lam-theories~\cite{IntrigilaS09,BarendregtTh,Plotkin74,BarendregtBKV78,IntrigilaS04}.
Formally, the $\omega$-rule  states that for all \lam-terms $M$ and $N$, $M=N$ whenever $MP = NP$ holds for all closed \lam-terms~$P$.
A \lam-theory $\cT$ \emph{satisfies the $\omega$-rule} whenever it is closed under such a rule.
Since this is such an impredicative rule, we can meaningfully wonder how the \lam-theory $\BTo$, obtained as the closure of $\BTth$ under the $\omega$-rule, compares with the other \lam-theories.
As shown by Barendregt in~\cite[Lemma~16.4.4]{Bare}, $\BTe$ does not satisfy the $\omega$-rule, while $\Hst$ does~\cite[Thm.~17.2.8(i)]{Bare}.

Therefore, the three possible scenarios are the following:\\[5pt]
\begin{tikzpicture}
\node (root) at (-4,0) {~};
\node (root) at (0,0) {~};
\node (BTe) at ($(root)+(-4.5,0)$) {$\qquad\BTe\subsetneq\Hpl\subseteq\BTo\subsetneq\Hst\quad$ or};
\node (BTe1) at ($(root)+(1,0)$) {$\qquad\BTe\subsetneq\BTo\subseteq\Hpl\subsetneq\Hst\quad$ or};
\node (BTe2) at ($(root)+(4.75,0)$) {$\BTe$};
\node (BTo2) at ($(BTe2)+(1,.3)$) {$\BTo$};
\node (Hpl2) at ($(BTo2)+(0,-.6)$) {$\Hpl$};
\node (Hst2) at ($(Hpl2)+(1,.3)$) {$\Hst.$};
\node[rotate=25] at (148pt,6pt) {$\subsetneq$};
\node[rotate=-25] at (149pt,-7pt) {$\subsetneq$};
\node[rotate=-25] at (178pt,6pt) {$\subsetneq$};
\node[rotate=25] at (178pt,-8pt) {$\subsetneq$};
\end{tikzpicture}\\
In the seventies, Sall\'e was working with Coppo and Dezani on type systems for studying termination properties of \lam-terms~\cite{Salle1978,CoppoDS79}.
In 1979, at the conference on \lam-calculus that took place in Swansea, he conjectured that a strict inclusion $\BTo\subsetneq\Hpl$ holds.
Such a conjecture was reported in the proof of \cite[Thm.~17.4.16]{Bare}, but for almost fourty years no progress has been made in that direction.
In this paper we demonstrate that the \lam-theories $\BTo$ and $\Hpl$ actually coincide, thus disproving Sall\'e's conjecture.
We now give an outline of the proof, discuss the results we need, the techniques we develop and the underlying ideas.
\smallskip

\noindent{$\BTo\subseteq\Hpl$.} 
The fact that the \lam-theory $\BTo$ is included in $\Hpl$ follows immediately if one can prove that $\Hpl$ satisfies the $\omega$-rule.
We notice that, on \emph{closed} \lam-terms, observational equivalences $\cT$ are equivalently defined by applicative contexts $C[]$ of shape $[]P_1\cdots P_k$, where the $P_i$'s are closed as well.
Moreover, if $\cT$ satisfies the $\omega$-rule, two closed $M,N$ are equated if they have the same observable behaviour in every \emph{non-empty} applicative context.
Whence, the key point in proving that $\cT$ is closed under the $\omega$-rule is being able to complete any applicative context $[]\vec P$ distinguishing $M$ from $N$, to ensure that it is non-empty.
For $\Hst$, this follows from Wadsworth's characterization of \lam-terms having a head normal form in terms of solvability: it is possible to find $\vec P$ such that, say, $M\vec P$ is equal to the identity~$\bI$, while $N\vec P$ is unsolvable and $\vec P$ can be chosen of any length by adding copies of $\bI$ at the end.

To prove that $\Hpl$ satisfies the $\omega$-rule, we need to show something similar, namely that when $M$ has a $\beta$-normal form while $N$ does not, we can find a non-empty context $[]\vec P$  preserving this property.
Interestingly, it is sufficient to prove this for \lam-terms $M,N$ that are equated in the \lam-theory $\Hst$; in other words we need to perform a detailed analysis of the equations in $\Hst - \Hpl$.
We show that when two closed \lam-terms $M,N$ are equal in $\Hst$, but different in $\Hpl$, their B\"ohm trees are similar but there exists a (possibly virtual) position $\sigma$ where they differ because of an infinite $\eta$-expansion of a variable $x$,  and such an $\eta$-expansion follows the structure of some computable infinite tree~$T$.
Thanks to a refined B\"ohm-out technique, we prove that it is always possible to extract such a difference by defining a suitable applicative context $[]\vec P$ that sends $M$ into the identity and $N$ into some infinite $\eta$-expansion of the identity (Theorem~\ref{thm:newsep}).
This provides a separability theorem in the spirit of~\cite{Hyland75,CoppoDR78,DezaniG01} but the notion of separability that we consider is weaker since it arises from Morris's observability. 
We then prove that applying an infinite $\eta$-expansion of the identity $\bI$ to $\bI$ itself, one still gets a (possibly different) infinite $\eta$-expansion of $\bI$.
From this closure property we obtain that also in this case the length of the discriminating context $[]\vec P$ can be chosen arbitrarily by adding copies of $\bI$ at the end.
Once this property has been established, the fact that $\Hpl$ satisfies the $\omega$-rule follows (Theorem~\ref{thm:Hplomega}).

\noindent{$\Hpl\subseteq\BTo$.} To prove this result we need to show that, whenever two \lam-terms $M$ and $N$ are equal in $\Hpl$, they are also equal in~$\BTo$.
From~\cite{Hyland75}, we know that in this case there is a B\"ohm tree~$U$ such that $\BT{M}\leeta U\geeta \BT{N}$, where $V\leeta U$ means that the B\"ohm tree $U$ can be obtained from $V$  by performing countably many finite $\eta$-expansions.
Thus, the B\"ohm trees of $M,N$ are compatible and have a common  ``$\eta$-supremum''~$U$.

Our proof can be divided into several steps:

\begin{enumerate}
\item\label{step1}
	We show that the aforementioned $\eta$-supremum $U$ is \emph{\lam-definable}: there exists a \lam-term $P$ such that $\BT{P} = U$ (Proposition~\ref{prop:Hplchar}).
\item\label{step3} We define a \lam-term $\etamax[]$ (Definition~\ref{def:etamax}) taking as arguments (the \emph{codes} $\code\cdot$ of) two \lam-terms $M_1,M_2$ and a \emph{stream} (infinite sequence)~$S$ of \lam-terms.  
	Assuming the B\"ohm tree of $M_2$ is more $\eta$-expanded than the one of $M_1$, i.e. $\BT{M_1}\leeta\BT{M_2}$, we show that:\smallskip
	\begin{enumerate}[label={(\roman*)}]
	\item\label{biribim} $\etamax[]\code{M_1}\code{M_2}$ is able to reconstruct the B\"ohm tree of $M_2$ when taking as input the stream $S = \ETA$ listing all finite $\eta$-expansions of the identity (Lemma~\ref{lem:Nwins}), \smallskip
	\item\label{biribam} $\etamax[]\code{M_1}\code{M_2}$ reconstructs the B\"ohm tree of $M_1$ from the input stream $S = \ID$ containing infinitely many copies of the identity (Lemma~\ref{lem:Mwins}).
	\end{enumerate}
\item\label{step2} By a tricky use of the $\omega$-rule, we show that the two streams $\ETA$ and $\ID$ are actually equated in $\BTo$ (Corollary~\ref{cor:IDeqETA}). From the discussion in \eqref{step3}, we realize that countably many finite $\eta$-expansions are collapsing already in the \lam-theory $\BTo$. 
	
	\item\label{step4} Summing up, if $M,N$ are equal in $\Hpl$, then by (1) there is a \lam-term $P$ such that $\BT{M}\leeta\BT{P}\geeta\BT{N}$.
	By (2), $\etamax[]\code M\code PS$ has the same B\"ohm tree as $M$ when the stream $S$ is $\ID$, and as $P$ when $S$ is $\ETA$.
	Symmetrically,  $\etamax[]\code N\code P\ID$ has the same B\"ohm tree as $N$, while $\etamax[]\code N\code P\ETA$ has the same B\"ohm tree as $P$.
	Since $\ID$ and $\ETA$ are equal in $\BTo$ by (3), we conclude $M =_{\BTo} P =_\BTo N$ (Theorem~\ref{thm:main}).\\[-1ex]
\end{enumerate}
The intuition behind $\etamax[]\code M\code N S$ is that, working on their codes, the \lam-term $\etamax[]$ computes the B\"ohm trees of $M$ and $N$, compares them, and at every position applies to the ``smaller'' (the less $\eta$-expanded one) an element extracted from the stream $S$ in the attempt of matching the structure of the ``larger'' (the more $\eta$-expanded one).
If the stream $S$ contains all possible $\eta$-expansions then each attempt succeeds, so $\etamax[]\code M\code N \ETA$ computes the $\eta$-supremum of $\BT{M}$ and $\BT{N}$.
If $S$ only contains infinitely many copies of the identity, each non-trivial attempt fails, and $\etamax[]\code M\code N \ID$ computes their $\eta$-infimum.

\subsection*{A characterization of $\BTe$}
The technique that we develop for $\eta$-expanding B\"ohm trees in a controlled way is powerful enough to open the way for a characterization of $\BTe$ as well.
More precisely, we prove that two \lam-terms $M$ and $N$ are equal in $\BTe$ exactly when their B\"ohm trees are equal up to countably many $\eta$-expansions of \emph{bounded size} (Theorem~\ref{thm:main2}).
Indeed, in this case $M$ and $N$ admit an $\eta$-supremum $U$ obtained from their B\"ohm trees by performing at every position $\sigma$ at most $n$ $\eta$-expansions, each having size bounded by $n$.
(In this context, the \emph{size} of an $\eta$-expansion is not the actual size of its tree but rather the maximum between its height and its maximal number of branchings.)
It turns out that, when exploiting our \lam-term $\etamax[]\code M \code N \ETA$ to compute the $\eta$-supremum, it only relies on a \emph{finite} portion of the input stream $\ETA$. 
Since in $\BTe$ any finite sequence $[\eta_1,\dots,\eta_k]$ is equal to the sequence $[\bI,\dots,\bI]$ and $\etamax[]\code M \code N [\bI,\dots,\bI]$ actually computes the $\eta$-infimum, we obtain once again that the $\eta$-supremum and the $\eta$-infimum collapse.
We can therefore proceed as in the proof sketched above for $\BTo$ and conclude that $M$ and $N$ are equal in $\BTe$.

\subsection*{Discussion}
We build on the characterizations of $\Hpl$ and $\Hst$ given by Hyland\linebreak and Wadsworth~\cite{Hyland75,Hyland76,Wadsworth76} and subsequently improved by L\'evy~\cite{Levy78}. 
In Section~\ref{sec:extensional_BT} we give a uniform presentation of these preliminary results using the formulation given in~\cite[\S19.2]{Bare} for $\Hst$, that exploits the notion of B\"ohm-like trees, namely labelled trees that ``look like'' B\"ohm trees but might not be \lam-definable. 
B\"ohm-like trees were introduced in~\cite{Bare} since at the time researchers were less familiar with the notion of coinduction, but they actually correspond to infinitary terms coinductively generated by the grammar of $\beta$-normal forms possibly containing the constant~$\bot$.
It is worth mentioning that such characterizations of $\Hpl$ and $\Hst$ have been recently rewritten by Severi and de~Vries using the modern approach of infinitary rewriting~\cite{Severi2002,SeveridV17}, and that we could have used their formulation instead.

A key ingredient in our proof of $\BTo\subseteq\Hpl$ is the fact that \lam-terms can be encoded as natural numbers, and therefore as Church numerals, in an effective way.
This is related to the theory of self-interpreters in \lam-calculi, which is an ongoing subject of study~\cite{Mogensen92,Given-WilsonJ11,Polonsky11,BrownP16}, and we believe that the present paper provides a nice illustration of the usefulness of such interpreters.
As a presentation choice, we decided to use the encoding described in Barendregt's book \cite[Def.~6.5.6]{Bare}, even if it works for closed \lam-terms only, because it is the most standard. 
However, our construction could be recast using any (effective) encoding, like the one proposed by Mogensen in~\cite{Mogensen92} that works more generally for open terms.

\subsection*{Disclaimer.}
The present paper is a long version of the extended abstract~\cite{IntrigilaMP17} published in the Proceedings of the Second International Conference on Formal Structures for Computation and Deduction (FSCD) 2017.
The primary goal of this article is to describe the mathematical context where Sall\'e's conjecture has arisen, and provide a self-contained treatment of its refutation.
Besides giving more detailed proofs and examples, we provide some original results only announced in~\cite{IntrigilaMP17}, like the characterization of $\BTe$ in terms of an extensional equality between B\"ohm trees up to bounded $\eta$-expansions.

Notice that a proof-sketch of the fact that $\Hpl$ satisfies the $\omega$-rule previously appeared in a conference paper written by the second and third authors in collaboration with Breuvart and Ruoppolo~\cite{BreuvartMPR16}.
Since the topic of that paper is mainly semantical we decided --- in agreement with them --- to exploit the present article to provide the missing details concerning Morris's separability theorem. 
The semantic results contained in~\cite{BreuvartMPR16} are the subject of a different article~\cite{BreuvartMR18} presenting more broadly the class of relational graph models and their properties.

\subsection*{Outline.} The structure of the present paper is the following.
	Section~\ref{sec:prelim} contains the preliminaries, mainly concerning the untyped \lam-calculus --- we present its syntax and recall several well-established properties.
	In Section~\ref{sec:extensional_BT}, we review the main notions of extensional equalities on B\"ohm trees and provide a few paradigmatic examples.
	The key results concerning the $\omega$-rule in connection with several \lam-theories are presented in Section~\ref{sec:omega}; we conclude the section by stating Sall\'e's conjecture.
	Section~\ref{sec:Boehmout} is devoted to studying the structural properties of the set of $\eta$-expansions of the identity, introduce our version of the B\"ohm-out technique and present the weak separability theorem for Morris's observability.
As a consequence, we get that $\Hpl$ satisfies the $\omega$-rule.
	In Section~\ref{sec:streams} we show how to build B\"ohm trees, and their $\eta$-supremum and $\eta$-infimum, starting from the codes of \lam-terms and streams of $\eta$-expansions of the identity.
	Section~\ref{sec:Salleiswrong} is devoted to the actual proof of the refutation of Sall\'e's conjecture.
	In Section~\ref{sec:BTeta} we provide a characterization of $\BTe$ in terms of a notion of equality on B\"ohm trees up to \emph{bounded} $\eta$-expansions.



\newpage
\section{Preliminaries}\label{sec:prelim}
We review some basic notions and introduce some notations that will be used in the rest of the paper.

\subsection{Coinduction} Throughout this paper, we often consider
possibly infinite trees as coinductive objects and perform coinductive reasoning.
Here we recall some basic facts and introduce some terminology, but we mainly assume that the reader is familiar with these concepts. 
If that is not the case, we suggest the following tutorials on the subject~\cite{JacobsR97,KozenS17}.

A \emph{coinductive structure}, also known as \emph{coinductive datatype}, is the greatest fixed point over a grammar, or equivalently the final coalgebra over the corresponding signature. 
We also consider \emph{coinductive relations}, that are the greatest relations over such coinductive structures that respect the structural constraints. 
A coinductive proof that two elements of the structures stand in
relation to one another is given by an infinite derivation tree, which is a coinductive structure itself.

Since structural coinduction has been around for decades and many efforts have been made within the community to explain why it should be used as innocently as structural induction, in our proofs we will not reassert the coinduction principle every time it is used. 
Borrowing the terminology from~\cite{KozenS17}, we say that we apply the ``coinductive hypothesis'' whenever the coinduction principle is applied.
The idea is that one can appeal to the coinductive hypothesis as long
as there has been progress in producing the nodes of a tree and there is no further analysis of the subtrees.
We believe that this style of mathematical writing greatly improves the readability of our proofs without compromising their correctness; the suspicious reader can study~\cite{KozenS17} where it is explained how this informal terminology actually corresponds to a formal application of the coinduction principle. 

\subsection{Sequences, Trees and Encodings}\label{ssec:sequences}
We let $\nat$ be the set of all natural numbers and $\nat^*$ be the set of all finite sequences over $\nat$.
Given a sequence $\sigma = \seqof{n_1,\dots,n_k}$ and $n\in\nat$ we write $\con{\sigma}{n}$ for the sequence $\seqof{n_1,\dots,n_k,n}$.
Given two sequences $\sigma,\tau\in\nat^*$ we write $\concat{\sigma}{\tau}$ for their concatenation.
We will denote the empty sequence by $\emptyseq$.

We consider fixed an effective (bijective) encoding of all finite sequences of natural numbers $\# : \nat^*\to\nat$. In particular, we assume that for all $\sigma\in\nat^*$ and $n\in\nat$ the code $\#(\sigma.n)$ is computable from $\#\sigma$ and~$n$.

\begin{defi}\label{def:tree} 
An \emph{(unlabelled) tree} is a partial function $T : \nat^*\to\nat$ such that $\dom(T)$ is closed under prefixes and, for all $\sigma\in\dom(T),n\in\nat$ we have $\con{\sigma}n \in\dom(T) \iff n < T(\sigma)$. 
	The \emph{subtree of $T$ at $\sigma$} is the tree $\subt T\sigma$ defined by setting $\subt T\sigma\!\!(\tau) = T(\concat\sigma\tau)$ for all $\tau\in\nat^*$.
\end{defi}

The elements of $\dom(T)$ are called \emph{positions}.
Notice that, in our definition of a tree $T$, $T(\sigma)$ provides the number of children of the node at position $\sigma$; therefore we have $T(\sigma) = 0$ whenever the position $\sigma$ corresponds to a leaf.
\begin{defi}
A tree $T$ is called: \emph{recursive} if the function $T$ is partial recursive (after coding $\nat^*$ using $\#(-)$); \emph{finite} if $\dom(T)$ is finite;  \emph{infinite} if it is not finite.
\end{defi}
We denote by $\Trees[]$ (resp.\ $\Trees$) the set of all (infinite) recursive trees.

\subsection{The Lambda Calculus}
We generally use the notation of Barendregt's first book~\cite{Bare} for the untyped \lam-calculus.
Let us fix a denumerable set $\Var$ of variables.

The set $\Lambda$ of \emph{\lam-terms} (over $\Var$) is defined inductively by the following grammar:
$$
\Lambda:\qquad\quad	M,N,P,Q\bnf\ x\ \mid\ \lam x.M\ \mid\ MN \qquad\textrm{ where  $x\in\Var$}.
$$
We suppose that application associates to the left and has higher precedence than \lam-abstraction. 
For instance, we write $\lam x.\lam y.\lam z.xyz$ for the \lam-term $\lam x.(\lam y.(\lam z.((xy)z)))$.
We write $\lam\seq x.M$ as an abbreviation for $\lam x_1.\dots\lam x_n.M$,  
$M\seq N$ for $MN_1\cdots N_n$,
$MN^{\sim n}$ for $MN\cdots N$ 
and finally $M^n(N)$ for $M(M(\cdots (MN)))$ ($n$ times).

The set $\FV{M}$ of \emph{free variables of $M$} and {$\alpha$-conversion} are defined as in~\cite[Ch.~1\S2]{Bare}.
Given $M,N\in\Lam$ and $x\in\Var$ we denote by $M\subst{x}{N}$ the capture-free substitution of $N$ for all free occurrences of $x$ in $M$. 
From now on, \lam-terms are considered up to \emph{$\alpha$-conversion}.

\begin{defi}
We say that a \lam-term $M$ is \emph{closed} whenever $\FV{M}= \emptyset$. 
In this case, $M$ is also called \emph{a combinator}.
We denote by~$\Lamo$ the set of all combinators.
\end{defi}

The \lam-calculus is a higher order term rewriting system and several \emph{notions of reduction} can be considered. 
Let us consider an arbitrary notion of reduction $\to_\mathtt{R}$.
The \emph{multistep $\mathtt{R}$-reduction $\msto[\mathtt{R}]$} is obtained by taking its reflexive and transitive closure.
The \emph{$\mathtt{R}$-conversion $=_\mathtt{R}$} is defined as the reflexive, transitive and symmetric closure of $\to_\mathtt{R}$.
We denote by $\nf_{\mathtt R}(M)$ the $\mathtt{R}$-normal form of $M$, if it exists,
and by $\NF[\mathtt{R}]$ the set of all \lam-terms in $\mathtt{R}$-normal form.

In this paper we are mainly concerned with the study of $\beta$- and $\eta$- reductions.

\begin{defi}\
\begin{itemize}
\item
	The \emph{$\beta$-reduction} $\to_\beta$ is the contextual closure of the rule: 
	$$
	(\beta)\qquad	(\lam x.M)N \redto M\subst{x}{N}.\hspace{50pt}
	$$ 
\item 
	The \emph{$\eta$-reduction} $\redto[\eta]$ is the contextual closure of the rule: 
	$$
	(\eta)\qquad	\lam x.Mx \redto M\textrm{ where } x\notin\FV{M}.
	$$
\end{itemize}
\end{defi}
We denote by $\redto[\beta\eta]$ the notion of reduction obtained from the union $\to_\beta\cup\to_\eta$, and by~$\msto[\beta\eta]$ (resp.\ $=_{\beta\eta}$) the corresponding multistep reduction (resp.\ conversion).
A \lam-term $M$ in $\beta$-normal form has the shape $M = \lam\seq x.yM_1\cdots M_k$ where each $M_i$ is in $\beta$-normal form.

Concerning specific combinators, we use the following notations (for $n\in\nat$):\label{pageref:one}
$$  
	\begin{array}{c}
	\bI = \One^0 =  \lam x.x,
	\hspace{40pt}
	\One^{n+1}= \lam xz.x(\One^nz),
	\hspace{40pt}
	\bB = \lam fgx.f(gx),\\
	\hspace{19pt}	
	\bK = \lam xy.x,	\hspace{41pt}
	\hspace{35pt}
	\bF = \lam xy.y,	\hspace{64pt}
	\Omega = (\lam x.xx)(\lam x.xx),
	\\
	\hspace{12pt}
	\bY = \lam f.(\lam x.f(xx))(\lam x.f(xx)),
	\hspace{103pt}
	\bJ = \bY(\lam jxz.x(jz)),\\
	\end{array}
$$
where $=$ denotes syntactic equality (up to $\alpha$-conversion).
The \lam-term $\bI$ represents the identity, $\One^n$ is a $\beta\eta$-expansion of the identity, $\bB$ is the composition combinator $M\circ N = \bB MN$, $\bK$ and $\bF$ are the first and second projections, $\Omega$ is the paradigmatic looping \lam-term, $\bY$ is Curry's fixed point combinator and $\bJ$ is the combinator defined by Wadsworth in~\cite{Wadsworth76}.

We denote by $\Church{n} = \lam fz.f^n(z)$ the $n$-th \emph{Church numeral}~\cite[Def.~6.4.4]{Bare}, by $\succc$ and $\pred$ the successor and predecessor, and by $\ifz(\Church{n},M,N)$  the {\lam-term} $\beta$-convertible to $M$ if $n = 0$ and $\beta$-convertible to $N$ otherwise.
For $\sigma\in\nat^*$, we denote by $\code{\sigma}$ the numeral~$\Church{\#\sigma}$.

We encode the \emph{$n$-tuple} $(M_1,\dots,M_n)\in\Lam^n$ in \lam-calculus by setting $[M_1,\dots,M_n] = \lam y.yM_1\cdots M_n$ with $y\notin\FV{M_i}$. 
Thus, $[M_1,M_2]$ represents the \emph{pair} $(M_1,M_2)$~\cite[Def.~6.2.4]{Bare}.
\begin{defi}
An enumeration of closed \lam-terms $e = (M_0,M_1,M_2,\dots)$ is called \emph{effective} (or \emph{uniform} in~\cite[\S8.2]{Bare}) if there exists $F\in\Lamo$ such that $F\Church n =_\beta M_n$.
\end{defi}
Given an effective enumeration $e$ as above, it is possible\footnote{One needs to use the fixed point combinator $\bY$, as shown in~\cite[Def.~8.2.3]{Bare}.} to define the \emph{stream} $[M_n]_{n\in\nat}$ as a single \lam-term satisfying $[M_n]_{n\in\nat} =_\beta [M_0,[M_{n+1}]_{n\in\nat}]$. We often use the notation:
$$
	[M_n]_{n\in\nat} = [M_0,[M_1,[M_2,\dots]]]. 
$$
The \emph{$i$-th projection} associated with a stream is $\pi_i = \lam y.y\bF^{\sim i}\bK$ since $\pi_i [M_n]_{n\in\nat} =_\beta M_i$.\label{proj_stream}

\subsection{Solvability}


The \lam-terms are classified into solvable and unsolvable,~dep\-ending on their capability of interaction with the environment.

\begin{defi} A closed \lam-term $N$ is  {\em solvable} if $NP_1\cdots P_n =_\beta \bI$ for some $P_1,\dots,P_n\in\Lam$.
A \lam-term $M$ is \emph{solvable} if its closure $\lam\seq x.M$ is solvable.
Otherwise $M$ is called \emph{unsolvable}. 
\end{defi}

It was first noticed by Wadsworth in~\cite{Wadsworth76} that solvable terms can be characterised in terms of ``head normalizability''.
The \emph{head reduction} $\redto[h]$ is the notion of reduction obtained by contracting the \emph{head redex} $(\lam y.P)Q$ in a \lam-term having shape $\lam\seq x.(\lam y.P)QM_1\cdots M_k$.
A \lam-term $M$ \emph{is in head normal form} (\emph{hnf}, for short) if it has the form $\lambda x_{1}\ldots x_{n}.yM_{1}\cdots M_{k}$.
The \emph{principal hnf of $M$} is the hnf (if it exists) obtained from $M$ by head reduction, i.e.\ $M\msto[h] \phnf{M}$; it is unique since $\redto[h]$ is deterministic.
A \lam-term $M$ has a head normal form if and only if  $\phnf{M}$ exists.
We denote by $\HNF$ the set of all head normal forms.

\begin{thm} A \lam-term $M$ is solvable if and only if $M$ has a head normal form.
\end{thm}

The typical example of an unsolvable is $\Omega$.
The following result appears as Lemma~17.4.4 in~\cite{Bare} and shows that any $M\in\Lamo$ can be turned into an unsolvable by applying enough~$\Omega$'s.

\begin{lem}\label{lemma:Omegak}
 For all $M\in\Lamo$ there exists $k\in\nat$ such that $M\Omega^{\sim k}$ is unsolvable.
\end{lem}
\vspace{-.3cm}

\subsection{B\"ohm(-like) trees}
The B\"ohm trees were introduced by Barendregt~\cite{Barendregt77}, who named them after B\"ohm since their structure arises from the proof of the homonymous theorem~\cite{bohm68}.\linebreak
They are coinductively defined labelled trees representing the execution of a \lam-term~\cite{Lassen99}.

\begin{defi}
The \emph{B\"{o}hm tree of a \lam-term $M$} is the coinductive
structure defined by: 
\begin{figure}[t!]
\begin{tikzpicture}
\node (myspot) at (-5.9,0) {~};
\node (BTxyo) at (-6.5,3.5) {$\BT{\lam x.y\Omega}$};
\node[below of = BTxyo, node distance=8pt] (eq) {$\shortparallel$};
\node[below of = eq, node distance=10pt] (l1) {$\lam x.y$};
\node[below of = l1, node distance=18pt,xshift=6pt] (l2) {$\bot$};
\draw ($(l1.south)+(6pt,0.05)$) -- ($(l2.north)-(0,0.05)$);
\node (BTJ) at (-4.5,3.5) {$\BT{\One^3}$};
\node[below of = BTJ, node distance=8pt] (eq) {$\shortparallel$};
\node[below of = eq, node distance=10pt] (l1) {$\lam xz_0.x~~$};
\node[below of = l1, node distance=18pt] (l2) {$\lam z_1.z_0$};
\node[below of = l2, node distance=18pt] (l3) {$\lam z_2.z_1$};
\node[below of = l3, node distance=18pt] (l4) {$\phantom{\lam z_3.}z_2$};
\draw ($(l1.south)+(8pt,0.05)$) -- ($(l2.north)+(8pt,-0.15)$);
\draw ($(l2.south)+(8pt,0.05)$) -- ($(l3.north)+(8pt,-0.15)$);
\draw ($(l3.south)+(8pt,0.05)$) -- ($(l4.north)+(8pt,-0.15)$);
\node (BTJ) at (-2.75,3.5) {$\BT{\bJ}$};
\node[below of = BTJ, node distance=8pt] (eq) {$\shortparallel$};
\node[below of = eq, node distance=10pt] (l1) {$\lam xz_0.x~~$};
\node[below of = l1, node distance=18pt] (l2) {$\lam z_1.z_0$};
\node[below of = l2, node distance=18pt] (l3) {$\lam z_2.z_1$};
\node[below of = l3, node distance=18pt] (l4) {$\lam z_3.z_2$};
\node[below of = l4, xshift=8pt, node distance=10pt] (l5) {$\vdots$};
\draw ($(l1.south)+(8pt,0.05)$) -- ($(l2.north)+(8pt,-0.15)$);
\draw ($(l2.south)+(8pt,0.05)$) -- ($(l3.north)+(8pt,-0.15)$);
\draw ($(l3.south)+(8pt,0.05)$) -- ($(l4.north)+(8pt,-0.15)$);
\node (BTOf) at (-1,3.5) {$\BT{\bY}$};
\node[below of = BTOf, node distance=8pt] (eq) {$\shortparallel$};
\node[below of = eq, node distance=10pt] (l1) {$\lam f.f$};
\node[below of = l1, xshift=6pt, node distance=18pt] (l2) {$f$};
\node[below of = l2, node distance=18pt] (l3) {$f$};
\node[below of = l3, node distance=18pt] (l4) {$f$};
\node[below of = l4, node distance=10pt] (l5) {$\vdots$};
\draw ($(l1.south)+(6pt,0.05)$) -- ($(l2.north)+(0pt,-0.05)$);
\draw ($(l2.south)+(0,0.05)$) -- ($(l3.north)+(0pt,-0.05)$);
\draw ($(l3.south)+(0,0.05)$) -- ($(l4.north)+(0pt,-0.05)$);
\node (BTP) at (1.6,3.5) {$\BT{[M]_{n\in\nat}}$};
\node[below of = BTP, node distance=10pt] (eq) {$\shortparallel$};
\node (root) at ($(BTP)+(0,-20pt)$) {};
\node (BTM) at (root) {$\lam y.y$};
\node (M0) at ($(root)+(-10pt,-22pt)$) {$\BT{M_0}$};
\node (x1) at ($(root)+(50pt,-22pt)$) {$\lam y.y$};
\node (M1) at ($(x1)+(-13pt,-22pt)$) {$\BT{M_1}$};
\node (x2) at ($(x1)+(50pt,-22pt)$) {$\lam y.y$};
\node (M2) at ($(x2)+(-13pt,-22pt)$) {$\BT{M_2}$};
\draw ($(BTM.south east)+(-4pt,3pt)$) -- ($(x1.north)+(0pt,-2pt)$);
\draw ($(BTM.south east)+(-10pt,3pt)$) -- (M0);
\draw ($(x1.south east)+(-2pt,2pt)$) -- ($(x2.north)+(0pt,-2pt)$);
\draw ($(x1.south east)+(-10pt,2pt)$) -- (M1);
\draw ($(x2.south east)+(-10pt,2pt)$) -- (M2);
\draw[densely dotted] ($(x2.south east)+(-2pt,2pt)$) -- ($(x2.south east)+(30pt,-7pt)$);
\node (BTM) at (6,3.5) {$\BT{\Omega} = \bot$};

\end{tikzpicture}
\caption{Some examples of B\"ohm trees.}\label{fig:Boehm}
\end{figure}
\begin{itemize}
\item	
	if $M$ is has a hnf and $\phnf M = \lambda x_{1}\ldots x_{n}.yM_{1}\cdots M_{k}$ then: 
\begin{center}
\begin{tikzpicture}
\node (BT) at (0,0) {$\BT{M}=$};
\node[right of = BT, node distance =2cm] (head) {$\lam x_{1} \ldots x_{n}.y$};
\node (BTN1) at ($(head)+(-.25,-.75)$) {$\BT{M_1}$};
\draw ($(head.south east)+(-0.5,0.1)$) -- ($(BTN1.north)-(0.,0)$);
\node (BTNk) at ($(head)+(2.,-.75)$) {$\BT{M_k}$};
\draw ($(head.south east)+(-0.1,0.1)$) -- ($(BTNk.north)-(0.,0)$);
\node at ($(head)+(.9,-.75)$) {$\cdots$};
\end{tikzpicture}
\end{center}
\item
	otherwise $M$ is unsolvable and $\BT{M} = \bot$.
\end{itemize}
\end{defi}

In Figure~\ref{fig:Boehm}, we provide some examples of B\"ohm trees of notable \lam-terms.
Comparing the B\"ohm trees of $\One^3$ and $\bJ$ we note that both look like $\eta$-expansions of the identity, but the former actually is, while the latter gives rise to infinite computations.
Terms like $\bJ$ are called ``infinite $\eta$-expansions'' of the identity.
A simple inspection of the B\"ohm tree of $[M_n]_{n\in\nat}$ should
convince the reader that the following lemma holds.

\begin{lem}[{cf.~\cite[10.1.5(v)]{Bare}}]\label{lem:eqimpBTeq}
Let $(M_i)_{i\in\nat}$ and $(N_i)_{i\in\nat}$ be two effective enumerations.
$$
	\forall i\in\nat\ .\ \BT{M_i} = \BT{N_i} \iff \BT{[M_n]_{n\in\nat}} = \BT{[N_n]_{n\in\nat}}.
$$
\end{lem}

Many results from the literature are expressed by exploiting, directly as in~\cite[Ch.~10]{Bare} or indirectly as in~\cite[Ch.~11]{RonchiP04}, the more general notion of ``B\"ohm-like'' trees.

\begin{defi}
The set $\BTset$ of \emph{B\"ohm-like trees} is coinductively generated by:
$$
	\BTset: \qquad U,V\ ::=_{\textrm{coind}}\ \bot\mid \lam x_1\dots x_n.yU_1\cdots U_k\qquad\textrm{ for }k,n\ge 0.
$$
\end{defi}
In~\cite[Def.~10.1.12]{Bare}, B\"ohm-like trees are labelled trees defined as partial functions mapping positions $\sigma\in\nat^*$ to pairs $(\lam \seq x.y,k)$ where $k$ is the number of children of the node labelled $\lam\seq x.y$ at $\sigma$. 
In a personal communication, Barendregt told us that the reason is twofold: firstly coinduction was not as well understood at the time he wrote that book as it is now; secondly speaking of functions makes it easier to define ``partial computable'' trees.
Indeed, not all B\"ohm-like trees $U$ arise as a B\"ohm tree of a \lam-term $M$: by~\cite[Thm.~10.1.23]{Bare}, this is the case precisely when $U$ is partial computable and $\FV{U}$ is finite.

\begin{nota}
Given $U\in\BTset$ we denote by $\nak{U}$ its \emph{underlying naked tree}, namely the (unlabelled) tree $T$ having the same structure as $U$.
\end{nota}


\subsection{The Lattice of Lambda Theories}\label{subsec:lamtheories}

 Inequational theories and \lam-theories become the main object of
study when one considers the computational equivalence between
\lam-terms as being more important than the process of computation itself.

A relation $\mathtt{R}$ between $\lam$-terms is \emph{compatible} if it is compatible with lambda abstraction and application.
We say that $\mathtt{R}$ is a \emph{congruence} if it is a compatible equivalence relation.

\begin{defi} 
An \emph{inequational theory} is any compatible preorder on $\Lam$ containing the $\beta$-conversion. 
A {\em \lam-theory} is any congruence on $\Lam$ containing the $\beta$-conversion. 
\end{defi}
Given an inequational theory $\cT$, we write $\cT\vdash M \sqle N$ or $M\sqle_\cT N$ when $M$ is smaller than or equal to $N$ in~$\cT$.
For a \lam-theory~$\cT$, $\cT\vdash M = N$ and $M =_\cT N$ stand for $(M,N)\in\cT$.

The set of all \lam-theories, ordered by set-theoretical inclusion, forms a complete lattice having a quite rich mathematical structure, as shown by Lusin and Salibra in~\cite{SalibraL04}. 

\begin{defi}
A \lam-theory (resp. inequational theory) $\cT$ is called:
\begin{itemize}
\item
{\em consistent} if it does not equate all \lam-terms, \emph{inconsistent} otherwise;
\item
{\em extensional} if it contains $\eta$-conversion; 
\item
{\em sensible} if it equates all unsolvable \lam-terms.
\end{itemize}
\end{defi}

\begin{nota}
We denote by 
\begin{itemize}
\item
	$\blam$ the least \lam-theory, 
\item	
	 $\blam\eta$ the least extensional \lam-theory,
\item 
	 $\cH$ the least sensible \lam-theory,  
\item 
	 $\BTth$ the \lam-theory equating all \lam-terms having the same B\"ohm tree.
\end{itemize}
\end{nota}

The \lam-theory $\BTth = \{(M,N)\st\BT M=\BT N\}$ is sensible, but not extensional.

\begin{defi} We let $\cT\eta$ be the least extensional \lam-theory containing the \lam-theory~$\cT$.
\end{defi}

\begin{rem}\label{rem:closed_terms}
It is well known that two \lam-theories $\cT,\cT'$ that coincide on closed terms must be equal. Hence we often focus on equalities between closed \lam-terms.
\end{rem}

Several interesting \lam-theories are obtained through observational preorders that are defined with respect to some set $\cO$ of \emph{observables}.

\begin{defi}
A \emph{context} $C[]$ is a \lam-term containing some occurrences of a ``\emph{hole}'', namely an algebraic variable denoted by~$[]$.
A context is called a \emph{head context} if it has the shape $(\lam x_1\dots x_n.[]) P_1\cdots P_k$ for some $k,n\ge 0$ and $\vec P\in\Lam$.
A head context is called \emph{applicative} if it has the shape $[]P_1\cdots P_k$ and \emph{closed} if all the $P_i$'s are closed.
\end{defi}

Given a context $C[]$ and a \lam-term $M$, we write $C\hole M$ for the \lam-term obtained from $C[]$ by replacing without renaming $M$ for the hole, possibly with capture of free variables in $M$.

\begin{defi}\label{def:obseq}
Let $\cO\subseteq\Lam$.
\begin{itemize}
\item We write $M\in_{\beta}\cO$ if there exists $M'\in\Lam$ such that $M =_\beta M'\in\cO$.
\item 
	The \emph{$\cO$-observational preorder} $\obsle[\cO]$ is defined as follows:
$$
	M\obsle[\cO] N \iff \forall C[]\ .\ \big(C\hole M \in_\beta\cO\ \Rightarrow\ C\hole N\in_\beta\cO\big).
$$
\item
	The \emph{$\cO$-observational equivalence} $\obseq[\cO]$ is defined by setting $M\obseq[\cO]N$ if and only if both $M\obsle[\cO]N$ and $N\obsle[\cO]M$ hold.
\end{itemize}
\end{defi}

It is easy to check that $M\obseq[\cO] N$ if and only if $\forall C[]\ .\ (C\hole M \in_\beta\cO\iff C\hole N\in_\beta\cO)$.

In the present paper we focus on the following observational preorders and equivalences.

\begin{defi}\ 
\begin{itemize}
\item $\Hst:$ Hyland/Wadsworth's observational preorder $\sqle_\Hst$ and equivalence $=_\Hst$ are generated by taking as $\cO$ the set $\HNF$ of head normal forms \cite{Hyland76,Wadsworth76}.
\item $\Hpl:$ Morris's observational preorder $\sqle_\Hpl$ and equivalence $=_\Hpl$ are generated by taking as $\cO$ the set $\NF[\beta]$ of $\beta$-normal forms \cite{Morristh}.
\end{itemize}
\end{defi}

These observational preorders and equivalences are easily proved to be inequational and \lam-theories, respectively.
(For the general case, one needs some hypotheses on the set $\cO$ of observables, as discussed in~\cite{Paolini08}.)
The \lam-theory $\Hst$ can be characterized as the (unique) maximal consistent sensible \lam-theory~\cite[Thm. 16.2.6]{Bare}.
The proof extends easily to $\sqle_\Hst$ which is maximal consistent among sensible inequational theories~\cite[Lemma~2.3]{BreuvartMR18}.
It follows that $\Hst$ is extensional and $\Hpl\subseteq\Hst$. 
As $\eta$-reduction is strongly normalizing, a \lam-term has a $\beta$-normal form if and only if it has a $\beta\eta$-normal form, therefore $\Hpl$ is extensional as well.

For both $\Hst$ and $\Hpl$ appropriate Context Lemmas describe the kind of contexts it is sufficient to consider for separating two \lam-terms
(cf. \cite[Lemma~6.1]{Wadsworth76} and \cite[Lemma~1.2]{DezaniG01}).

\begin{lem}[Context Lemma]\label{lemma:ctx}
For $\cO =\NF[\beta]$ or $\cO=\HNF$, the following are equivalent:
\begin{enumerate}
\item $M\not\obsle[\cO] N$,
\item there exists a closed head context $C[]$ such that $C[M]\in_\beta\cO$, while $C[N]\notin_\beta\cO$.
\end{enumerate}
In particular, if $M$ and $N$ are closed then $C[]$ can be chosen closed and applicative.
\end{lem}

The quantification over all possible (applicative) contexts makes these definitions difficult to handle in practice.
That is the reason why researchers have looked for characterizations of $\Hst$ and $\Hpl$ in terms of extensional equalities between B\"ohm trees.


\section{B\"ohm Trees and Extensionality}\label{sec:extensional_BT}
We review three different notions of extensional equality between B\"ohm trees corresponding to the equality in $\BTe,\Hpl$ and $\Hst$.
\begin{figure*}[t!]
\begin{tikzpicture}
\node (root) at (0,0) {~};
\node (BTM) at (root) {$\BT{\IDX}$};
\node at ($(BTM)+(45pt,0)$) {$=_{\BTe}$};
\node at ($(BTM)+(0,-10pt)$) {$\shortparallel$};
\node (M) at ($(BTM)+(0,-20pt)$) {$\lam xy.y$};
\node (M1) at ($(M)+(-2pt,-21.5pt)$) {$x$};
\node (M2) at ($(M)+(20pt,-20pt)$) {$\lam y.y$};
\node (M21) at ($(M2)+(-2pt,-21pt)$) {$x$};
\node (M22) at ($(M2)+(18pt,-20pt)$) {$\lam y.y$};
\node (M221) at ($(M22)+(-2pt,-21pt)$) {$x$};
\node (M222) at ($(M22)+(18pt,-20pt)$) {$\lam y.y$};
\node (M2221) at ($(M222)+(-2pt,-20pt)$) {~};
\node (M2222) at ($(M222)+(20pt,-20pt)$) {~};
\draw[-] ($(M.south)+(7pt,2pt)$) -- ($(M1.north)+(0pt,-0pt)$);
\draw[-] ($(M.south east)+(-3pt,2pt)$) -- ($(M2.north)+(2pt,-3pt)$);
\draw[-] ($(M2.south)+(5pt,2pt)$) -- ($(M21.north)+(0pt,-0pt)$);
\draw[-] ($(M2.south east)+(-2pt,2pt)$) -- ($(M22.north)+(2pt,-3pt)$);
\draw[-] ($(M22.south)+(5pt,2pt)$) -- ($(M221.north)+(0pt,-0pt)$);
\draw[-] ($(M22.south east)+(-2pt,2pt)$) -- ($(M222.north)+(2pt,-3pt)$);
\draw[-,densely dotted] ($(M222.south)+(5pt,2pt)$) -- (M2221.south west);
\draw[-,densely dotted] ($(M222.south east)+(-2pt,2pt)$) -- (M2222.south);
\node (BTM) at ($(root)+(90pt,0)$) {$\BT{\ONE}$};
\node at ($(BTM)+(50pt,0)$) {$=_{\Hpl}$};
\node at ($(BTM)+(0,-10pt)$) {$\shortparallel$};
\node (M) at ($(BTM)+(0,-20pt)$) {$\lam xy.y$};
\node (M1) at ($(M)+(-12pt,-20pt)$) {$\lam z.x$};
\node (M2) at ($(M)+(30pt,-20pt)$) {$\lam y.y$};
\node (w) at ($(M1)+(7pt,-21.5pt)$) {$z$};
\node (M21) at ($(M2)+(-14pt,-20pt)$) {$\lam z.x$};
\node (M22) at ($(M2)+(26pt,-20pt)$) {$\lam y.y$};
\node (w1) at ($(M21)+(7pt,-21.5pt)$) {$z$};
\node (M221) at ($(M22)+(-10pt,-20pt)$) {$\lam z.x$};
\node (M222) at ($(M22)+(24pt,-20pt)$) {$\lam y.y$};
\node (w2) at ($(M221)+(7pt,-21.5pt)$) {$z$};
\node (M2221) at ($(M222)+(-2pt,-20pt)$) {~};
\node (M2222) at ($(M222)+(20pt,-20pt)$) {~};
\draw[-] ($(M.south)+(7pt,2pt)$) -- ($(M1.north)+(4pt,-3pt)$);
\draw[-] ($(M.south east)+(-3pt,2pt)$) -- ($(M2.north)+(2pt,-3pt)$);
\draw[-] ($(M1.south)+(8pt,2pt)$) -- ($(w.north)+(1pt,-0pt)$);
\draw[-] ($(M2.south)+(4pt,2pt)$) -- ($(M21.north)+(4pt,-3pt)$);
\draw[-] ($(M2.south east)+(-3pt,2pt)$) -- ($(M22.north)+(2pt,-3pt)$);
\draw[-] ($(M21.south)+(8pt,2pt)$) -- ($(w1.north)+(1pt,-0pt)$);
\draw[-] ($(M22.south)+(5pt,2pt)$) -- ($(M221.north)+(4pt,-3pt)$);
\draw[-] ($(M22.south east)+(-3pt,2pt)$) -- ($(M222.north)+(2pt,-3pt)$);
\draw[-] ($(M221.south)+(8pt,2pt)$) -- ($(w2.north)+(1pt,-0pt)$);
\draw[-,densely dotted] ($(M222.south)+(4pt,2pt)$) -- (M2221.south west);
\draw[-,densely dotted] ($(M222.south east)+(-3pt,2pt)$) -- (M2222.south);
\node (BTM) at ($(root)+(190pt,0)$) {$\BT{\ETAge}$};
\node at ($(BTM)+(57pt,0)$) {$=_{\Hst}$};
\node at ($(BTM)+(0,-10pt)$) {$\shortparallel$};
\node (M) at ($(BTM)+(0,-20pt)$) {$\lam xy.y$};
\node (M1) at ($(M)+(-6pt,-20pt)$) {$\eta^1(x)$};
\node (M2) at ($(M)+(23pt,-20pt)$) {$\lam y.y$};
\node (M21) at ($(M2)+(-8pt,-21pt)$) {$\eta^2(x)$};
\node (M22) at ($(M2)+(22pt,-20pt)$) {$\lam y.y$};
\node (M221) at ($(M22)+(-8pt,-21pt)$) {$\eta^3(x)$};
\node (M222) at ($(M22)+(22pt,-20pt)$) {$\lam y.y$};
\node (M2221) at ($(M222)+(-2pt,-20pt)$) {~};
\node (M2222) at ($(M222)+(20pt,-20pt)$) {~};
\draw[-] ($(M.south)+(7pt,2pt)$) -- ($(M1.north)+(5pt,-4pt)$);
\draw[-] ($(M.south east)+(-3pt,2pt)$) -- ($(M2.north)+(2pt,-3pt)$);
\draw[-] ($(M2.south)+(5pt,2pt)$) -- ($(M21.north)+(5pt,-4pt)$);
\draw[-] ($(M2.south east)+(-2pt,2pt)$) -- ($(M22.north)+(2pt,-3pt)$);
\draw[-] ($(M22.south)+(5pt,2pt)$) -- ($(M221.north)+(5pt,-4pt)$);
\draw[-] ($(M22.south east)+(-2pt,2pt)$) -- ($(M222.north)+(2pt,-3pt)$);
\draw[-,densely dotted] ($(M222.south)+(5pt,2pt)$) -- (M2221.south west);
\draw[-,densely dotted] ($(M222.south east)+(-3pt,2pt)$) -- (M2222.south);
\node (BTM) at ($(root)+(305pt,0)$) {$\BT{\JAY}$};
\node at ($(BTM)+(0,-10pt)$) {$\shortparallel$};
\node (M) at ($(BTM)+(0,-20pt)$) {$\lam xy.y$};
\node (M1) at ($(M)+(-16pt,-20pt)$) {$\lam z_0.x$};
\node (M2) at ($(M)+(34pt,-20pt)$) {$\lam y.y$};
\node (w1) at ($(M1)+(3pt,-20pt)$) {$\lam z_1.z_0$};
\node (M21) at ($(M2)+(-8pt,-21pt)$) {$\eta^{\infty}(x)~$};
\node (M22) at ($(M2)+(22pt,-20pt)$) {$\lam y.y$};
\node (w2) at ($(w1)+(0pt,-20pt)$) {$\lam z_2.z_1$};
\node (M221) at ($(M22)+(-8pt,-21pt)$) {$\eta^\infty(x)~$};
\node (M222) at ($(M22)+(22pt,-20pt)$) {$\lam y.y$};
\node (w3) at ($(w2)+(0pt,-20pt)$) {$~$};
\node (M2221) at ($(M222)+(-2pt,-20pt)$) {~};
\node (M2222) at ($(M222)+(20pt,-20pt)$) {~};
\draw[-] ($(M.south)+(7pt,2pt)$) -- ($(M1.north)+(5pt,-4pt)$);
\draw[-] ($(M.south east)+(-3pt,2pt)$) -- ($(M2.north)+(2pt,-3pt)$);
\draw[-] ($(M1.south)+(11pt,2pt)$) -- ($(w1.north)+(8pt,-4pt)$);
\draw[-] ($(M2.south)+(5pt,2pt)$) -- ($(M21.north)+(5pt,-4pt)$);
\draw[-] ($(M2.south east)+(-2pt,2pt)$) -- ($(M22.north)+(2pt,-3pt)$);
\draw[-] ($(w1.south)+(8pt,2pt)$) -- ($(w2.north)+(8pt,-4pt)$);
\draw[-] ($(M22.south)+(5pt,2pt)$) -- ($(M221.north)+(5pt,-4pt)$);
\draw[-] ($(M22.south east)+(-2pt,2pt)$) -- ($(M222.north)+(2pt,-3pt)$);
\draw[-,densely dotted] ($(w2.south)+(8pt,2pt)$) -- ($(w3.north)+(8pt,-4pt)$);
\draw[-,densely dotted] ($(M222.south)+(5pt,2pt)$) -- (M2221.south west);
\draw[-,densely dotted] ($(M222.south east)+(-3pt,2pt)$) -- (M2222.south);
\end{tikzpicture}
\caption{The B\"ohm trees of $\IDX,\ONE,\ETAge$ and $\JAY$.}\label{fig:BTseq}
\end{figure*}

The following streams will be used as running examples in the rest of the section:
$$
	\begin{array}{lcl}
	\hspace{5pt}\IDX  x= [x,[x,[x,\dots]]],
	&\qquad&
	\ONE x= [\One x,[\One x,[\One x,\dots]]],\\[1ex]
	\ETAge x= [\One^1x,[\One^2x,[\One^3x,\dots]]],
	&&
	\JAY x= [\bJ x,[\bJ x,[\bJ x,\dots]]],\\
	\end{array}
$$
where the combinators $\One$, $\One^n$ and $\bJ$ are given on Page~\pageref{pageref:one}. 
For example, one can define $\ETAge = \bY(\lam mnx.[Fnx,m(\succc\, n)x])\Church{0}$ where $F = [\lam zxy.x(zy),\One]$.
The B\"ohm trees of these streams are depicted in Figure~\ref{fig:BTseq}, using the notations $\eta^n(x) = \BT{\One^nx}$ and $\eta^\infty(x) = \BT{\bJ x}$.

\subsection{$\BTe$: Countably Many $\eta$-Expansions of Bounded Size}\label{subsec:BTeta}

Recall that $\BTe$ is the least extensional \lam-theory including~$\BTth$.
Except for some lemmas in \cite[\S 16.4]{Bare}, the \lam-theory $\BTe$ has been mostly neglected in the literature, probably because it does not arise as the theory of any known denotational model.
Perhaps, one might be led to think that $M =_\BTe N$ entails that the B\"ohm trees of $M$ and $N$ differ because of finitely many $\eta$-expansions. 
In reality, one $\eta$-expansion of $M$ can generate countably many $\eta$-expansions in its B\"ohm tree.

A typical example of this situation is the following:
$$
\ONE =_\BTth \bY(\lam mx.[\lam z.xz,mx]) \redto[\eta]  \bY(\lam mx.[x,mx])  =_\BTth  \IDX  
$$
thus $\IDX $ and $\ONE$ are equated in $\BTe$ despite the fact that their B\"ohm trees differ by infinitely many $\eta$-expansions.
More precisely, $M\redto[\eta] N$ entails that $\BT{M}$ can be obtained from $\BT{N}$ by performing at most \emph{one} $\eta$-expansion at every position. 
\begin{lemC}[{\cite[Lemma~16.4.3]{Bare}}]\label{lemma:Bareta}
If $M\redto[\eta] N$, then $\BT{M}$ is obtained from $\BT{N}$ by replacing in the latter some 
\begin{center}
\begin{tikzpicture}
\node (M1) {$\lam x_1\dots x_n.y$};
\node (T) at ($(M1)+(24pt,-20pt)$) {$U_1\quad\cdots\quad U_k$};
\draw ($(M1.south east)+(-3pt,3pt)$) -- ($(T.east)+(-12pt,6pt)$);
\draw ($(M1.south east)+(-13pt,3pt)$) -- ($(T.north)+(-26pt,-2pt)$);
\node at ($(M1)+(100pt,-10pt)$) (by) {by};
\node (N1) at ($(M1)+(170pt,0)$) {$\lam x_1\dots x_nz.y$};
\node (T1) at ($(N1)+(26pt,-20pt)$) {$~\quad U_1\quad\cdots\quad U_k\ z$};
\draw ($(N1.south east)+(-3pt,3pt)$) -- ($(T1.east)+(-6pt,4pt)$);
\draw ($(N1.south east)+(-3pt,3pt)$) -- ($(T1.east)+(-21pt,6pt)$);
\draw ($(N1.south east)+(-13pt,3pt)$) -- ($(T1.north)+(-25pt,-2pt)$);
\draw (M1);
\end{tikzpicture}
\end{center}
possibly infinitely often (but simultaneously, thus without $\eta$-expanding the new variable~$z$).
\end{lemC}

In particular, no finite amount of $\eta$-expansions in $\IDX $ can turn its B\"ohm tree into $\BT{\ETAge}$, which has infinitely many $\eta$-expansions of increasing size.

\begin{cor}
$\BTe\vdash \IDX  = \ONE$, while $\BTe\vdash \IDX \neq \ETAge$.
\end{cor}

\subsection{$\Hpl$: Countably Many Finite $\eta$-Expansions}

By definition, $\Hpl$ is the \lam-theory corresponding to Morris's original observational equivalence where the observables are the $\beta$-normal forms~\cite{Morristh}. 
The \lam-theory $\Hpl$ and its inequational version have been studied both from a syntactic and from a semantic point of view.
We refer to~\cite{Hyland75,Levy78,CoppoDZ87} for some standard literature and to~\cite{ManzonettoR14,BreuvartMPR16} for more recent work.
The properties that we present here can be found in \cite[\S11.2]{RonchiP04}.
Two \lam-terms having the same B\"ohm tree cannot be distinguished by any context $C[]$, so we have $\BTth\subseteq\Hpl$.
Since $\Hpl$ is an extensional \lam-theory, we get $\BTe\subseteq\Hpl$.

The question naturally arising is whether there are \lam-terms different in $\BTe$ that become equal in $\Hpl$.
It turns out that $\Hpl\vdash M = N$ holds exactly when $\BT{M}$ and $\BT{N}$ are equal up to countably many $\eta$-expansions of finite size.
A typical example of this situation is given by the streams $\IDX $ and $\ETAge$ since the B\"ohm tree of the latter can be obtained from the B\"ohm tree of the former by performing infinitely many finite $\eta$-expansions.

\begin{defi}\label{def:leeta}
Given two B\"ohm-like trees $U$ and $V$, we define coinductively the relation  $U\leeta V$ expressing the fact that $V$ is a \emph{finitary $\eta$-expansion} of $U$.
We let $\leeta$ be the greatest relation between B\"ohm-like trees such that $U\leeta V$ entails that
\begin{itemize} 
\item either $U = V=\bot$, 
\item or, for some $k,m,n\ge 0$:
$$
U = \lam x_1\dots x_n.yU_1\cdots U_k\ \textrm{ and }\
V =\lam x_1\dots x_nz_1\dots z_m.yV_1\cdots V_k Q_1\cdots Q_{m}
$$
where $\seq z\notin\FV{yU_1\cdots U_kV_1\cdots V_k}$, $U_j \leeta V_j$ for all $j\le k$ and $Q_i \msto[\eta]z_i$ for all $i\le m$.
\end{itemize}
\end{defi}

It is easy to check that $\BT\IDX \leeta  \BT\ETAge$ holds.

\begin{nota}For $M,N\in\Lam$, we write $M\leeta N$ if and only if $\BT{M}\leeta \BT{N}$.
\end{nota}

Note that $M\leeta  N$ and $N\leeta  M$ entail $\BT{M} = \BT{N}$, therefore the equivalence corresponding to $\leeta $ and capturing $=_\Hpl$ needs to be defined in a more subtle way.

\begin{thm}[Hyland~\cite{Hyland75}, see also \cite{Levy78}]\label{thm:CharacterizationHpl}
For all $M,N\in\Lam$, we have 
	$\Hpl\vdash M = N$ if and only if there exists a B\"ohm-like tree $U\in\BTset$ such that $\BT{M}\leeta  U \geeta \BT{N}$.
\end{thm}

This means that in general, when $M =_\Hpl N$, one may need to perform countably many $\eta$-expansions of finite size both in $\BT{M}$ and in $\BT{N}$ to find the common ``$\eta$-supremum''.

\begin{cor} $\Hpl\vdash\IDX  = \ETAge$, while $\Hpl\vdash\IDX\neq\JAY$.
\end{cor}

\subsection{$\Hst$:  Countably Many Infinite $\eta$-Expansions}\label{subsec:Hstar}

The \lam-theory $\Hst$ and its inequational version are, by far, the most well studied theories of the untyped \lam-calculus~\cite{Bare,GouyTh,GianantonioFH99,RonchiP04,Manzonetto09,Breuvart14}.
Recall that $\Hst$ corresponds to the observational equivalence where the observables are the head normal forms. 
Two \lam-terms $M,N$ are equated in $\Hst$ if their B\"ohm trees are equal up to countably many $\eta$-expansions of possibly infinite depth.
The typical example is $\bI =_\Hst \bJ$.

\begin{defi}\label{def:leetainf}
Given two B\"ohm-like trees $U$ and $V$, we define coinductively the relation $U\leetainf V$ expressing the fact that $V$ is a \emph{possibly infinite $\eta$-expansion} of $U$.
We let $\leetainf$ be the greatest relation between B\"ohm-like trees such that $U\leetainf V$ entails that
\begin{itemize} 
\item either $U = V=\bot$, 
\item or (for some $k,m,n\ge 0$):
$$
U = \lam x_1\dots x_n.yU_1\cdots U_k\ \textrm{ and }\
V =\lam x_1\dots x_nz_1\dots z_m.yV_1\cdots V_k V'_1\cdots V'_{m}
$$
where $\seq z\notin\FV{yU_1\cdots U_kV_1\cdots V_k}$, $U_j \leetainf V_j$ for all $j\le k$ and $z_i\leetainf V'_i$ for all $i\le m$.
\end{itemize}
\end{defi}
For instance, we have $\BT{\IDX }\leetainf\BT{\JAY}$.
\begin{nota}For $M,N\in\Lam$, we write $M\leeta[\omega] N$ if and only if $\BT{M}\leeta[\omega] \BT{N}$.
\end{nota}
\begin{defi}
Given two B\"ohm-like trees $U$ and $V$ we write $U\BTle V$ whenever
$U$ results from $V$ by replacing some subtrees by $\bot$ (possibly infinitely many). 
\end{defi}

\begin{thm}[Hyland~\cite{Hyland76}/Wadsworth~\cite{Wadsworth76}]
\label{thm:char_Hst}
For all $M,N\in\Lam$, we have that:
\begin{enumerate}[label={(\roman*)}]
\item \label{thm:char_Hst1}
	$\Hst\vdash M \sqle N$ if and only if there exist two B\"ohm-like trees $U,V\in\BTset$ such that $\BT{M}\leetainf U\BTle V \geetainf \BT{N}$.
\item \label{thm:char_Hst2}
	$\Hst\vdash M = N$  if and only if there exists a B\"ohm-like tree $U\in\BTset$ such that $\BT{M}\leetainf U \geetainf \BT{N}$.
\end{enumerate}
\end{thm}


Exercise~10.6.7 in Barendregt's book \cite{Bare} consists in showing
that the $\eta$-supremum $U$ in item~\ref{thm:char_Hst2} above can
always be chosen to be the B\"ohm tree of some \lam-term. 
As we will prove in Section~\ref{sec:Salleiswrong}, this property also holds for the B\"ohm-like tree $U$ of Theorem~\ref{thm:CharacterizationHpl}.

\begin{cor} $\IDX ,\ONE,\ETAge,\JAY$ are all equal~in~$\Hst$. 
However,
$\BTe\vdash \IDX  \neq \ETAge$ and $\Hpl\vdash \ETAge \neq  \JAY$, so we have  $\BTe\subsetneq\Hpl\subsetneq\Hst$.
\end{cor}

\section{The Omega Rule and Sall\'e's Conjecture}\label{sec:omega}
The \lam-calculus possesses a strong form of extensionality which is known as the \emph{$\omega$-rule} \cite[Def.~4.1.10]{Bare}.
The $\omega$-rule has been extensively investigated in the literature by many authors~\cite{IntrigilaS09,BarendregtTh,Plotkin74,BarendregtBKV78,IntrigilaS04}.
Intuitively the $\omega$-rule mimics the definition of functional equality, namely it states that two \lam-terms $M$ and $N$ are equal whenever they coincide on every closed argument $P$.
\begin{figure}[t!]
\begin{center}
\begin{tikzpicture}
\node at (0,4) {~};
\node (BTP) at (0,3.5) {$\BT{P}$};
\node[below of = BTP, node distance=10pt] (eq) {$\shortparallel$};
\node[below of = eq, node distance=10pt] (l1) {$\lam yx.x$};
\node at ($(l1)+(.9,-.7)$) (l2y) {$y$};
\draw ($(l1.south east)+(-0.2,0.1)$) -- ($(l2y.north west)-(-0.05,0.05)$);
\node at ($(l2y)+(.9,-.7)$) (l3x) {$x$};
\draw ($(l2y.south east)+(-0,0.1)$) -- ($(l3x.north west)-(-0.05,0.05)$);
\node[below of = l1, node distance=39pt] (l3b) {$\bot$};
\draw ($(l2y.south west)+(0.,0.1)$) -- ($(l3b.north east)-(0.05,0.1)$);
\node at ($(l3x)+(.9,-.7)$) (l4y) {$y$};
\draw ($(l3x.south east)+(-0,0.1)$) -- ($(l4y.north west)-(-0.05,0.05)$);
\node at ($(l4y)+(.9,-.7)$) (l5x) {$x$};
\draw ($(l4y.south east)+(-0,0.1)$) -- ($(l5x.north west)-(-0.05,0.05)$);
\node[below of = l3x, node distance=39pt] (l5b) {$\bot$};
\draw ($(l4y.south west)+(0.,0.1)$) -- ($(l5b.north east)-(0.05,0.1)$);
\node[below of = l4y, node distance=19pt] (l5bb) {$\bot$};
\draw ($(l4y.south)+(0.,0.)$) -- ($(l5bb.north)-(0.,0.05)$);
\node at ($(l5x)+(.3,-.15)$) {$\ddots$};  	  
\node (BTQ) at (6,3.5) {$\BT{Q}$};
\node[below of = BTQ, node distance=10pt] (eq) {$\shortparallel$};
\node[below of = eq, node distance=10pt] (l1) {$\lam yx.x$};
\node at ($(l1)+(.9,-.7)$) (l2y) {$y$};
\draw ($(l1.south east)+(-0.2,0.1)$) -- ($(l2y.north west)-(-0.05,0.05)$);
\node at ($(l2y)+(.95,-.7)$) (l3x) {$\eta^1(x)$};
\draw ($(l2y.south east)+(-0,0.1)$) -- ($(l3x.north west)-(-0.2,0.1)$);
\node[below of = l1, node distance=39pt] (l3b) {$\bot$};
\draw ($(l2y.south west)+(0.,0.1)$) -- ($(l3b.north east)-(0.05,0.1)$);
\node at ($(l3x)+(.8,-.7)$) (l4y) {$y$};
\draw ($(l3x.south)+(0.25,0.1)$) -- ($(l4y.north west)-(-0.05,0.1)$);
\node at ($(l4y)+(.9,-.7)$) (l5x) {$\eta^2(x)$};
\draw ($(l4y.south east)+(-0.1,0.1)$) -- ($(l5x.north west)-(-0.15,0.05)$);
\node[below of = l3x, node distance=39pt] (l5b) {$\bot$};
\draw ($(l4y.south west)+(0.,0.1)$) -- ($(l5b.north east)-(0.05,0.1)$);
\node[below of = l4y, node distance=19pt] (l5bb) {$\bot$};
\draw ($(l4y.south)+(0.,0.)$) -- ($(l5bb.north)-(0.,0.05)$);
\node at ($(l5x)+(.5,-.3)$) {$\ddots$};
\end{tikzpicture}
\end{center}
\caption{The B\"ohm trees of the \lam-terms $P,Q$ from~\cite[Lemma~16.4.4]{Bare}.}\label{fig:PQ}
\end{figure}  

\begin{defi} The \emph{$\omega$-rule} is given by:
$$
	(\omega)\qquad\left(\forall P\in\Lambda^o \ .\ MP = NP\right) \textrm{ entails } M = N.
$$
We write $(\omega^0)$ for the $\omega$-rule restricted to $M,N\in\Lamo$.
\end{defi}

\begin{defi}
We say that a \lam-theory $\cT$ \emph{satisfies the $\omega$-rule}, written $\cT\vdash\omega$, if for all $P\in\Lamo$ $\cT\vdash MP = NP$ entails $\cT\vdash M=N$.
\end{defi}

Given a \lam-theory $\cT$ we denote by $\cT\omega$ the closure of $\cT$ under the rule~$(\omega)$. 
By definition, $\cT\vdash \omega$ entails $\cT\omega =\cT$.
The lemma below collects some basic results from~\cite[\S4.1]{Bare}.
\begin{lem}\label{lemma:omega-props} 
For all \lam-theories $\cT$, we have:
\begin{enumerate}[label={(\roman*)}]
\item\label{lemma:omega-props1}  
	$\cT\eta\subseteq\cT\omega$,
\item\label{lemma:omega-props2}  
	$\cT\vdash\omega$ if and only if $\cT\vdash\omega^0$,
\item\label{lemma:omega-props3}  
	$\cT\subseteq\cT'$ entails $\cT\eta\subseteq\cT'\eta$ and $\cT\omega\subseteq\cT'\omega$.
\end{enumerate}
\end{lem}
In general, because of the quantification over all $P\in\Lamo$, it can be difficult to understand which \lam-terms different in $\cT$ become equal in $\cT\omega$, especially when $\cT$ is extensional.

Theorem~17.4.16 from \cite{Bare}, that we reproduce below, contains one of the main diagrams of the book and shows all the relationships among the different \lam-theories under consideration. 

\begin{thmC}[{\cite[Thm.~17.4.16]{Bare}}] The following diagram indicates all possible inclusion relations of the \lam-theories involved (if $\cT_1$ is above $\cT_2$, then $\cT_1\subsetneq \cT_2$):
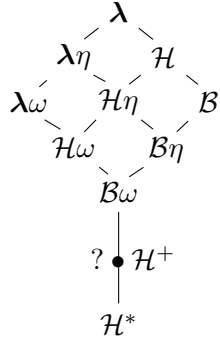
\begin{center}
  \begin{tikzpicture}
  	\node (root) at (0,0) {};
	\node (blam) at (root) {$\blam$};
	\node (lameta) at ($(blam)+(-17pt,-17pt)$) {$\blam\eta$};			
	\node (H) at ($(blam)+(17pt,-17pt)$) {$\cH$};			
	\node (Heta) at ($(blam)+(0,-34pt)$) {$\cH\eta$};	
	\node (Homega) at ($(Heta)+(-17pt,-17pt)$) {$\cH\omega$};		
	\node (BTeta) at ($(Heta)+(18pt,-18pt)$) {$\BTth\eta$};			
	\node (lamomega) at ($(blam)+(-34pt,-34pt)$) {$\blam\omega$};			
	\node (BT) at ($(blam)+(34pt,-34pt)$) {$\BTth$};		
	\node (BTomega) at ($(Heta)+(0,-34pt)$) {$\cB\omega$};		
	\node (Hpl) at ($(BTomega)+(5.4pt,-25pt)$) {$?\bullet\Hpl$};
	\node (Hst) at ($(BTomega)+(0,-50pt)$) {$\Hst$};	
	\draw (blam) -- ($(lameta.north)+(4pt,-2pt)$);
	\draw (blam) -- ($(H.north)-(4pt,1pt)$);	
	\draw (lameta) -- ($(lamomega.north)+(4pt,-2pt)$);
	\draw ($(lameta.south east)+(-3pt,2pt)$) -- ($(Heta.north)-(5pt,1pt)$);	
	\draw (H) -- ($(Heta.north)+(4pt,-2pt)$);
	\draw (H) -- ($(BT.north)-(4pt,1pt)$);		
	\draw ($(lamomega.south)+(6pt,1pt)$) -- ($(Homega.north)-(5pt,1pt)$);
	\draw (BT) -- ($(BTeta.north)-(-4pt,2pt)$);	
	\draw ($(Heta.south west)+(3pt,2pt)$) -- ($(Homega.north)+(4pt,-2pt)$);
	\draw ($(Heta.south east)+(-3pt,2pt)$) -- ($(BTeta.north)-(5pt,1pt)$);		
	\draw (Homega) -- ($(BTomega.north)-(4pt,1pt)$);
	\draw ($(BTeta.south west)+(2pt,2pt)$) -- ($(BTomega.north)-(-4pt,2pt)$);	
	\draw (BTomega) -- (Hst);
  \end{tikzpicture}
\end{center}
\end{thmC}
The picture above is known as ``Barendregt's kite'' because of its kite shaped structure.
Since $\blam\subsetneq\BTth\subsetneq\cH\subsetneq \Hst$ and $\Hst$ is maximal sensible, we have both $\blam\eta\subseteq\cH\eta\subseteq\BTth\eta\subseteq\Hst$ and 
$\blam\omega\subseteq\cH\omega\subseteq\BTth\omega\subseteq\Hst$ by Lemma~\ref{lemma:omega-props}\ref{lemma:omega-props3}.
All these inclusions turn out to be strict.

The counterexample showing that $\blam\eta\nvdash\omega$ is based on complicated universal generators known as \emph{Plotkin terms} \cite[Def.~17.3.26]{Bare}.
However, since these terms are unsolvable, they become useless when considering sensible \lam-theories.
We refer to \cite[\S17.4]{Bare} for the proof of $\cH\eta\nvdash\omega$ and rather discuss the validity of the $\omega$-rule for \lam-theories containing $\BTth$.

Let us consider two \lam-terms $P$ and $Q$ whose B\"ohm trees are depicted in Figure~\ref{fig:PQ}.
The B\"ohm trees of $P$ and $Q$ differ because of countably many finite $\eta$-expansions of increasing size, therefore they are different in $\BTe$ but equal in $\Hpl$.
This situation is analogous to what happens with $\IDX$ and $\ETAge$; indeed, $P\leeta  Q$ holds.
Perhaps surprisingly, $P$ and $Q$ can also be used to prove that $\BTe\subsetneq\BTo$ since $\BTo\vdash P = Q$ holds.
Indeed, by Lemma~\ref{lemma:Omegak}, for every $M\in\Lamo$, there exists $k$ such that $M\Omega^{\sim k}$ becomes unsolvable.
By inspecting Figure~\ref{fig:PQ}, we notice that in $\BT{P}$ the variable $y$ 
is applied to an increasing number of $\Omega$'s (represented in the tree by~$\bot$). 
So, when substituting some $M\in\Lamo$ for $y$ in $\BT{Py}$, there 
will be a level $k$ of the tree where $M\Omega\cdots\Omega$ becomes $\bot$,
thus cutting $\BT{PM}$ at level~$k$.
The same reasoning can be done for $\BT{QM}$.
Therefore $\BT{PM}$ and $\BT{QM}$ only differ because of finitely many $\eta$-expansions.
Since $\BTth\eta\subseteq\BTth\omega$, we conclude that $PM =_{\BTth\omega}QM$ and therefore by the $
\omega$-rule $P =_{\BTth\omega}Q$.
This argument is due to Barendregt~\cite[Lemma~16.4.4]{Bare}.

The fact that $\Hst\vdash\omega$ is an easy consequence of its maximality.
However, there are several direct proofs: see \cite[\S17.2]{Bare} for a syntactic 
demonstration and \cite{Wadsworth76} for a semantic~one.

A natural question, raised by Barendregt in~\cite[Thm.~17.4.16]{Bare}, concerns the position of $\Hpl$ in the kite. In the proof of that theorem, it is mentioned that Sall\'e formulated the following conjecture (represented in the diagram with a question mark).

\begin{conj}[Sall\'e's Conjecture]\label{conj:Salle}
$\BTo\subsetneq \Hpl$.
\end{conj}


\section{$\Hpl$ Satisfies the $\omega$-Rule}\label{sec:Boehmout}
We start by proving that $\Hpl$ satisfies the $\omega$-rule (Theorem~\ref{thm:Hplomega}), a result from which it follows that $\BTo\subseteq \Hpl$. 

\subsection{$\eta$-Expansions of the Identity}\label{subsec:ETAid}

We need to analyze more thoroughly the possibly infinite $\eta$-expansions of the identity that will play a key role in the rest of the paper.
The structural properties of the finite $\eta$-expansions of the identity have been analyzed in~\cite{IntrigilaN03}.

\begin{defi}
We say that \emph{$Q\in\Lam$ is an $\eta$-expansion of the identity} whenever $\bI\leeta[\omega] Q$. 
We let $\ETAset[\omega]$ be the set of all $\eta$-expansions of the identity.
\end{defi}

As a matter of terminology, we are slightly abusing notation when
saying that $Q\in\ETAset[\omega]$ is an $\eta$-expansion of the
identity, since $Q\msto[\eta] \bI$ does not hold in general. 
We use this terminology because we are silently considering $\cB$-equivalence classes of \lam-terms.

\begin{defi}
We say that $Q\in\ETAset[\omega]$ is a \emph{finite} (resp.\ \emph{infinite}) $\eta$-expansion of the identity whenever $\BT{Q}$ is a finite (infinite) tree.
We let $\ETAset$ (resp.\ $\ETAset[\infty]$) be the set of finite (infinite) $\eta$-expansions of the identity.
\end{defi}

Clearly, $\ETAset[\omega] = \ETAset \cup \ETAset[\infty]$ and this decomposition actually gives a partition.
Concerning finite $\eta$-expansions of the identity, we have that $\One^n\in\ETAset$ for all $n\in\nat$, every $Q\in\ETAset$ is $\beta$-normalizing, $\NF[\beta](Q)$ is a closed \lam-term and $\BT{Q}$ is not only finite but also $\bot$-free.

\begin{lem}\label{lemma:eta-equiv-defs}
For $Q\in\Lam$, the following are equivalent:
\begin{enumerate}
\item[(i)]
	$Q\in\ETAset$,
\item[(ii)]
	$Q =_\beta \lam yz_1\dots z_m.yQ_1\cdots Q_m$ such that $\lam z_i.Q_i\in\ETAset$ for all $i\le m$,	
\item[(iii)]
	$Q =_\beta \lam y.Q'$ such that $Q'\msto[\beta\eta] y$,
\item[(iv)]	
	$Q\msto[\beta\eta]\bI$.
\end{enumerate}
\end{lem}

In~\cite{IntrigilaN03}, it is shown that there exists a bijection between $\beta$-equivalence classes of \lam-terms in $\ETAset$ and finite (unlabelled) trees.
It is also proved that $(\ETAset,\circ,\bI)$ is an idempotent commutative monoid which is moreover closed under \lam-calculus application.
We generalize some of these properties to encompass the infinite $\eta$-expansions of the identity, like $\bJ$.

\begin{thm}\label{thm:one-one}
There is a one-to-one correspondence between $\Trees[]$ and $\ETAset[\omega]/_\BTth$.
\end{thm}

\begin{proof}
For the first direction, consider any $Q\in\ETAset[\omega]$.
By~\cite[Thm.~10.1.23]{Bare}, $\BT{Q}$ is computable and so is its underlying naked tree $\nak{\BT{Q}}$. 
Since $\bI\leeta[\infty] \BT{Q}$ entails that $\BT{Q}$ cannot have any occurrences of $\bot$, we conclude that $\dom(\nak{\BT{Q}})$ is recursive.
By construction, for all $Q,Q'\in\ETAset[\omega]$, we have that $\BTth\vdash Q\neq Q'$ entails $\nak{\BT{Q}}\neq\nak{\BT{Q'}}$.

For the second direction, consider $T\in\Trees[]$.
Recall from Section~\ref{ssec:sequences} that $\#\sigma$ denotes the code of $\sigma\in\nat^*$ and $\code{\sigma}$ the corresponding Church numeral $\Church{\#\sigma}$.
Using the combinator $\bY$, define a \lam-term $X\in\Lamo$ satisfying the following recursive equation (for all $\sigma\in\dom(T)$):
\begin{equation}\label{eq:receq}
	X \code{\sigma} =_\beta \lam xz_1\dots z_m.x(X\code{\sigma.0}z_1)\cdots(X\code{\sigma.m-1}z_m)\textrm{ where }m = T(\sigma).
\end{equation}
The existence of such a \lam-term follows from the fact that $T$ is recursive, the effectiveness of the encoding $\#$ and Church's Thesis.
We prove by coinduction that for all $\sigma\in\dom(T)$, $X\code{\sigma}$ is an infinite $\eta$-expansion of the identity having $\restr{T}{\sigma}$ as underlying naked tree. 
Indeed, $X \code{\sigma}$ is $\beta$-convertible to the \lam-term of Equation~\ref{eq:receq}. 
By coinductive hypothesis we get for all $i < T(\sigma)$ that $\bI \leeta[\omega] \BT{X\code{\sigma.i}}$ and $\nak{\BT{X\code{\sigma.i}}} = \restr{T}{\sigma.i}$.
From this, we conclude that $\bI\leeta[\infty] \BT{X\code{\sigma}}$ and $\nak{\BT{X\code{\sigma}}}= \restr{T}{\sigma}$.
So, the \lam-term associated with $T$ is $\bJ_T=X\code{\emptyseq}$.
By construction, for all $T,T'\in\Trees$, we have that $T\neq T'$ entails $\BTth\vdash \bJ_T\neq\bJ_{T'}$.
\end{proof}


\begin{nota} Given $T\in\Trees[]$, we denote by $\bJ_T$ the corresponding \lam-term in $\ETAset[\omega]$ whose B\"ohm tree is shown in Figure~\ref{fig:BTJT}.
We say that $\bJ_T$ is \emph{an $\eta$-expansion of the identity following~$T$}.
\begin{figure}[t!]
\begin{tikzpicture}
\node (myspot) at (-5.9,0) {~};
\node (lev1) at (0.5, .1) {\!\!\!\!\!\!\!\!\!\!\!\!\!\!\!\!\!\!\!\!\!\!\!\!\!\!\!\!\!\!\!\!\!$\lam xy_1\dots y_{T\emptyseq}.x$};
\node (lev2) at (0,-1) {$
	\lam z_1\dots z_{T\langle0\rangle}.y_1\qquad\qquad\ \, \cdots \quad \qquad
	\lam z_1\dots z_{T\langle T\emptyseq\textrm{-}1\rangle}.y_{T\emptyseq}$};
\node (lev3) at (.8,-2) {
	$\lam\vec w_{T\langle0,0\rangle}.z_1\quad \cdots \quad 
	  \lam\vec w_{T\langle0,T\langle0\rangle\textrm{-}1\rangle}.z_{T\langle0\rangle}\qquad
	  \lam \vec w_{T\langle T\emptyseq\textrm{-}1,0\rangle}.z_1\quad\cdots\quad 
	  \lam \vec w_{T\langle T\emptyseq\textrm{-}1,T\langle T\emptyseq\textrm{-}1\rangle\textrm{-}1\rangle}.z_{T\langle T\emptyseq\textrm{-}1\rangle}$};
	  \draw[-] ($(lev1.south east)+(-8pt,2pt)$) -- (-2.3,-.8);
	  \draw[-] ($(lev1.south east)+(-2pt,2pt)$) -- (3.8,-.8);	  	  
	  \draw[-] (-2.2,-1.2) -- (-1.2,-1.7);
	  \draw[-] (-2.4,-1.2) -- (-4.75,-1.7);	  
	  \draw[-] (4,-1.2) -- (2.7,-1.7);	  	  	  	  
	  \draw[-] (4.2,-1.2) -- (7.2,-1.7);	
	  
	  \draw[-] (-4.9,-2.15) -- (-3.75,-3);	 
	  \draw[-] (-5.1,-2.15) -- (-6.45,-3);	  	  	     	  
	  \node at (-5.1,-2.75) {$\cdots$}; 	  	  
	 
	  \draw[-] (-1.15,-2.15) -- (-2.45,-3);	  	  	  
  	  \draw[-] (-.95,-2.15) -- (.25,-3);	  	  	  	                      	  
	  \node at (-1.1,-2.75) {$\cdots$};	  
	
	  \draw[-] (2.35,-2.15) -- (3.6,-3);	  	  	                                           	  
	  \draw[-] (2.15,-2.15) -- (.9,-3);	  
	  \node at (2.3,-2.75) {$\cdots$};
	  
	    \draw[-] (7.2,-2.15) -- (8.5,-3);	  	  
	  \draw[-] (7,-2.15) -- (5.8,-3);	  	  	  	  	  
	  \node at (7.2,-2.75) {$\cdots$};	  	  	  
\end{tikzpicture}
\caption{The B\"ohm-like tree of an infinite $\eta$-expansion of $\bI$ following $T\in\Trees[]$. 
To lighten the notations we write $T\sigma$ for $T(\sigma)$
and $\lam \seq w_n$ for $\lam w_1\dots w_n$.}
\label{fig:BTJT}
\end{figure}
\end{nota}

\begin{lem}\label{lemma:Texas}\
\begin{enumerate}[label={(\roman*)}]
\item\label{lemma:Texas1} 
	The set $\ETAset[\omega]$ is closed under composition.\\
Moreover, $Q_1\circ Q_2\in\ETAset[\infty]$ whenever one of the $Q_i$'s belongs to $\ETAset[\infty]$.
\item\label{lemma:Texas2} 
	If $Q\in\ETAset[\infty]$ then $Q\,\bI\in\ETAset[\infty]$.
\end{enumerate}
\end{lem}
\begin{proof}
\ref{lemma:Texas1} Let $Q,Q'\in\ETAset[\omega]$, we prove by coinduction that $x\leeta[\omega] (Q\circ Q')x$ which is equivalent to $\bI\leeta[\omega] Q\circ Q'$. 
By definition, we have $Q=_\beta\lam xz_1\dots z_n.x(Q_1z_1)\cdots (Q_nz_n)$ with $z_i\leeta[\omega] Q_iz_i$ for $i\le n$, and $Q'=_\beta\lam xz_1\dots z_m.x(Q'_1z_1)\cdots (Q'_mz_m)$ with $z_j\leeta[\omega] Q'_jz_j$ for $j\le m$.

Consider $n\le m$, the other case being analogous. 
Easy calculations give:
$$
	\begin{array}{lcl}
	Q\circ Q' &=_\beta& \lam xz_1\dots z_n.(Q'x)(Q_1z_1)\cdots (Q_nz_n)\\
	&=_\beta&\lam x z_1\dots z_m.x(Q'_1(Q_1z_1))\cdots (Q'_n(Q_nz_n))(Q'_{n+1}z_{n+1})\cdots(Q'_mz_{m})
	\end{array}
$$
Since $Q'_i(Q_iz_i) =_\beta (Q'_i\circ Q_i)z_i$ for every $i\le n$ we get from the coinductive hypothesis that $z_i\leeta[\omega](Q'_i\circ Q_i)z_i$.
Together with the hypotheses $z_j\leeta[\omega] Q'_jz_j$ for $j\le m$, this allows to conclude $x\leeta[\omega] (Q\circ Q')x$.
Note that when $Q$ or $Q'$ belongs to $\ETAset[\infty]$, then at least one among $Q_1,\dots,Q_n,Q'_1,\dots Q'_m$ must have an infinite B\"ohm tree. If $\BT{Q'_j}$ is infinite for some $j> n$ then it is immediate that $Q\circ Q'\in\ETAset[\infty]$, otherwise it follows from the coinductive hypothesis.

\ref{lemma:Texas2}
We prove that $\bI\leeta[\omega] Q\bI$ holds and in the meanwhile we check that $\BT{Q\bI }$  is infinite.
By definition, $Q =_\beta\lam xz_0\dots z_n.x(Q_0z_0)\cdots (Q_nz_n)$ with $z_i\leeta[\omega] Q_iz_i$, for all $i\le n$, and $\BT{Q_j}$ is infinite for some index $j$.
Easy calculations give $Q\bI =_\beta\lam z_0\dots z_n.(Q_0z_0)\cdots (Q_nz_n)$.
Now, if $Q_0 =_\beta\bI$ then $j>0$ and the result follows immediately.
Otherwise $Q_0 =_\beta \lam z_0 y_1\dots y_m.z_0(Q'_1y_1)\cdots (Q'_my_m)$ with, say, $m\ge n$ and $y_j\leeta[\omega]Q'_jy_j$ for all $j\le m$. 

Therefore, in this case we have:
$$
	\begin{array}{lcl}
	Q\bI &=_\beta&\lam z_0\dots z_n.(\lam y_1\dots y_m.z_0(Q'_1y_1)\cdots (Q'_my_m))(Q_1z_1)\cdots (Q_nz_n)\\
	&=_\beta&\lam z_0\dots z_n.z_0(Q'_1(Q_1z_1))\cdots (Q'_n(Q_nz_n))(Q'_{n+1}z_{n+1})\cdots (Q'_mz_m)\\
	\end{array}
$$
Notice that $Q'_i(Q_iz_i) =_\beta (Q'_i\circ Q_i)z_i$ for every $i\le n$ so by applying \ref{lemma:Texas1} we get $Q'_i\circ Q_i\in\ETAset[\omega]$ and even $Q'_i\circ Q_i\in\ETAset[\infty]$ if one among these $Q'_i,Q_i$ has an infinite B\"ohm tree. Otherwise $\BT{Q'_j}$ must be infinite for some $j>n$. In both cases we conclude $Q\, \bI\in\ETAset[\infty]$.
\end{proof}

We notice that the properties above generalize to the following original result, that however we do not prove since it is beyond the scope of the present paper.

\begin{prop} 
The triple $(\ETAset[\omega],\circ,\bI)$ is an idempotent commutative monoid closed under application.
Moreover, $\ETAset[\infty]$ is a two-sided ideal in this monoid.
\end{prop}

\subsection{A Morris-Style Separation Theorem}

The \emph{B\"ohm-out} technique~\cite{bohm68,Bare,RonchiP04} aims to build a context which extracts (an instance of) the subterm of a \lam-term $M$ at position $\sigma$. 
It is used for separating two \lam-terms $M,N$ provided that their structure is sufficiently different, depending on the notion of \emph{separation} under consideration. 
When $\Hst\vdash M = N$ the \lam-terms $M,N$ are not \emph{semi-separable}\footnote{Using the terminology of \cite{Bare}.} and no B\"ohm-out technique has been developed in that context. 
We show that also in this case, under the hypothesis $\Hpl\vdash M \not\sqle N$, this difference can be B\"ohmed-out via an appropriate head context thus providing a weak separation result.

\begin{defi}
Given $M\in\Lam$ and $\sigma\in\nat^*$, we define the \emph{subterm of $M$ at $\sigma$ (relative to its B\"ohm tree)}, written $M_\sigma$, as follows: 
\begin{enumerate}
\item $M_\emptyseq = M$,
\item\label{subterm2} $M_{\concat{\Tuple{i}\,}{\,\sigma}} = (M_{i+1})_\sigma$ whenever $\phnf{M} = \lambda x_1\dots x_n.yM_1\cdots M_k$ and $i< k$,
\item $M_\sigma$ is undefined, otherwise.
\end{enumerate}
\end{defi}

Notice that $M_{\Tuple{0}}$, when defined, corresponds to the first argument of the phnf of~$M$. 
This explains the apparent mismatch in \eqref{subterm2} between the indices of $M_{\concat{\Tuple{i}\,}{\,\sigma}}$ and $(M_{i+1})_\sigma$.
The following definition will always be used under the hypothesis that $\Hst\vdash M = N$.
\begin{figure}[t!]
\begin{tikzpicture}
\node (root) at (0,0) {~};
\node  at ($(root)+(3pt,20pt)$) {$\BT{M}$};
\node (BTM) at($(root)+(3pt,10pt)$) {$\shortparallel$};
\node (M) at  (root) {$\lam x.x$};
\node (M1) at ($(M)+(-10pt,-30pt)$) {$\lam z_0.y$};
\node (M2) at ($(M)+(30pt,-30pt)$) {$\lam z_0z'_0.x$};
\node (M11) at ($(M1)+(-3pt,-30pt)$) {$x$};
\node (M12) at ($(M1)+(20pt,-30pt)$) {$z_0$};
\node (M21) at ($(M2)+(5pt,-30pt)$) {$z_0$};
\node (M22) at ($(M2)+(27pt,-29pt)$) {$z'_0$};
\draw ($(M.south east)+(-4pt,2pt)$) -- (M2);
\draw ($(M.south east)+(-9pt,2pt)$) -- (M1);
\draw ($(M1.south east)+(-10pt,1pt)$) -- (M11.north);
\draw ($(M1.south east)+(-4pt,1pt)$) -- ($(M12)+(0pt,5pt)$);
\draw ($(M2.south east)+(-10pt,3pt)$) -- (M21.north);
\draw ($(M2.south east)+(-3pt,3pt)$) -- ($(M22)+(-2pt,5pt)$);
\node at ($(root)+(130pt,20pt)$) {$\BT{N}$};
\node (BTN) at ($(root)+(130pt,10pt)$) {$\shortparallel$};
\node (N) at ($(BTN)+(0pt,-10pt)$) {$\lam xw.x$};
\node (N1) at ($(N)+(-25pt,-30pt)$) {$y$};
\node (N2) at ($(N)+(12pt,-30pt)$) {$\lam z_0.x$};
\node (N3) at ($(N)+(110pt,-30pt)$) {$\lam z_0z'_0.w~$};
\node (N11) at ($(N1)+(0pt,-30pt)$) {$x$};
\node (N21) at ($(N2)+(0pt,-30pt)$) {$\lam z_1.z_0$};
\node (N31) at ($(N3)+(-40pt,-30pt)$) {$\lam z_1z'_1.z_0$};
\node (N32) at ($(N3)+(70pt,-30pt)$) {$\lam z_1z'_1.z'_0$};
\node (N211) at ($(N21)+(0pt,-30pt)$) {$\lam z_2.z_1$};
\node (N2111) at ($(N211)+(7pt,-8pt)$) {$\vdots$};
\node (N3111) at ($(N31)+(-20pt,-30pt)$) {$\lam z_2z'_2.z_1$};
\node (N3112) at ($(N31)+(40pt,-30pt)$) {$\lam z_2z'_2.z'_1$};
\node (N3121) at ($(N32)+(-20pt,-30pt)$) {$\lam z_2z'_2.z_1$};
\node (N3122) at ($(N32)+(40pt,-30pt)$) {$\lam z_2z'_2.z'_1$};
\node at ($(N3111)+(20pt,-8pt)$) {$\ddots$};
\node at ($(N3111)+(5pt,-8pt)$) {$\iddots$};
\node at ($(N3112)+(20pt,-8pt)$) {$\ddots$};
\node at ($(N3112)+(5pt,-8pt)$) {$\iddots$};
\node at ($(N3121)+(20pt,-8pt)$) {$\ddots$};
\node at ($(N3121)+(5pt,-8pt)$) {$\iddots$};
\node at ($(N3122)+(20pt,-8pt)$) {$\ddots$};
\node at ($(N3122)+(5pt,-8pt)$) {$\iddots$};
\draw ($(N.south east)+(-10pt,2pt)$) -- (N1);
\draw ($(N.south east)+(-7pt,2pt)$) -- (N2);
\draw ($(N.south east)+(-4pt,2pt)$) -- (N3.north);
\draw ($(N1.south)+(0,1pt)$) -- (N11.north);
\draw ($(N3.south east)+(-12pt,3pt)$) -- ($(N31.north east)+(-10pt,-2pt)$);
\draw ($(N3.south east)+(-8pt,3pt)$) -- ($(N32.north west)+(15pt,-2pt)$);
\draw ($(N2.south east)+(-7pt,3pt)$) -- ($(N21.north east)+(-9pt,-4pt)$);
\draw ($(N21.south east)+(-9pt,3pt)$) -- ($(N211.north east)+(-9pt,-5pt)$);
\draw ($(N31.south east)+(-7pt,3pt)$) -- ($(N3112.north west)+(20pt,-4pt)$);
\draw ($(N31.south east)+(-12pt,3pt)$) -- ($(N3111.north east)+(-12pt,-4pt)$);
\draw ($(N32.south east)+(-7pt,3pt)$) -- ($(N3122.north west)+(20pt,-4pt)$);
\draw ($(N32.south east)+(-12pt,3pt)$) -- ($(N3121.north east)+(-12pt,-4pt)$);
\node at (0,-4) {~};
\end{tikzpicture}
\caption{Two \lam-terms ${M},{N}$ such that ${M}\in\NF[\beta]$, $\Hst\vdash {M} = {N}$, but $\Hpl\vdash{M}\not\sqle {N}$.}\label{fig:trees}
\end{figure}

\begin{defi}\label{def:spine}
We say that $\sigma\in\nat^*$ is a \emph{Morris separator for $M,N$}, written $\sigma : M\nlem N$, if there exists $i>0$ such that, for some $p\ge i$, we have:
$$
	M_{\sigma} =_\beta \lam x_1\dots x_{n}.yM_1\cdots M_{m}\quad\textrm{ and } \quad
	N_{\sigma} =_\beta \lam x_1\dots x_{n+p}.yN_1\cdots N_{m+p}	
$$
where $\BTth\vdash N_{m+i} = \bJ_T x_{n+i}$ for some $T\in\Trees$. 
\end{defi}

\begin{rem}
$\sigma : M\nlem N$ and $\sigma = \Tuple{k}\star\tau$ entail $\tau : M_{k+1}\nlem N_{k+1}$.
\end{rem}

Intuitively, a Morris separator for two \lam-terms $M,N$ that are equal in $\Hst$ is a common position belonging to their B\"ohm trees and witnessing the fact that $\Hpl\vdash M\not\sqle N$.
We will start by considering the case where $M$ has a $\beta$-normal form, which entails that $\BT{M}$ is finite and $\bot$-free.
Since $\Hst\vdash M = N$, by Theorem~\ref{thm:char_Hst}(ii) also $\BT{N}$ is $\bot$-free; moreover, at every common position $\sigma$, $M_\sigma$ and $N_\sigma$ have \emph{similar hnf's}, which means that the number of lambda abstractions and applications can be matched via suitable $\eta$-expansions.
Note that $\BT{M}$ might have $\eta$-expansions that are not present in $\BT{N}$.
As $\Hpl\vdash M\not\sqle N$, the B\"ohm tree of $N$ must have infinite subtrees of the form $\BT{\bJ_T x}$ for some $x\in\Var,T\in\Trees$.

\begin{exa}
Consider the \lam-terms $M,N$ whose B\"ohm trees are depicted in Figure~\ref{fig:trees}.
This example admits two Morris separators: 
\begin{itemize}
\item
	The empty sequence $\emptyseq$ is a separator since $\BTth\vdash N_{\langle 2 \rangle} = \bJ_{T_2} w$ where $T_2$ is the complete binary tree.
\item 
The sequence $\langle 1,0\rangle$ is a separator because $\BTth\vdash N_{\langle1,0,0\rangle} = \bJ_{T_1} z_1$ where $T_1$ is the complete unary tree (i.e., $\BTth\vdash \bJ_{T_1} = \bJ$).
\end{itemize}
\end{exa}

\begin{prop}\label{prop:Morrissep}  
Let $M,N\in\Lam$ be such that $M\in\NF[\beta]$, $N$ does not have a $\beta$-normal form and $\Hst\vdash M = N$.
Then there exists a position $\sigma\in\nat^*$ such that  $\sigma : M\nlem N$.
\end{prop}

\begin{proof} 
We proceed by structural induction on the $\beta$-normal form of $M$.

Base case: $\nf_\beta(M) = \lam x_1\dots x_n.y$ for $n\ge 0$. 
By Theorem~\ref{thm:char_Hst}\ref{thm:char_Hst2} there is $U\in\BTset$ such that $\BT{M}\leeta[\omega]U\geeta[\omega] \BT{N}$.
As one can easily check by induction, since $\nf_\beta(M)$ is $\eta$-normal, this entails $\nf_\beta(M)\leeta[\omega] N$.
Then $N=_\beta \lam x_1\dots x_{n+p}.yQ_1\cdots Q_p$ with $x_{n+j} \leeta[\omega] Q_j$ for all $j\le p$.
As $N$ does not have a $\beta$-normal form, we have that $p>0$ and there is an index $0<i\le p$ such that $\lam x_{n+i}.Q_i\in\ETAset[\infty]$.
By Theorem~\ref{thm:one-one}, there is $T\in\Trees$ such that $\BTth\vdash \lam x_{n+i}.Q_i = \bJ_T$, hence $\BTth\vdash Q_i = \bJ_Tx_{n+i}$.
Therefore $\emptyseq$ is a Morris separator for $M,N$.

Induction: $\nf_\beta(M) = \lam x_1\dots x_n.yM_1\cdots M_m$ for $m>0$, where each $M_i$ is $\beta$-normal. 
By Theorem~\ref{thm:char_Hst}(ii), there is a B\"ohm-like tree $U$ such that $\nf_\beta(M)\leeta[\omega] U \geeta[\omega] \BT{N}$, say:
$$
	U = \lam x_1\dots x_{n+p} z_1\dots z_k.yU_1\cdots U_{m+p}V_1\cdots V_k
$$
and $N =_\beta \lam x_1\dots x_{n+p}.yN_1\cdots N_{m+p}$ for some integer $p$ and $k\ge 0$.
Assume $p\ge 0$, the other case being analogous.
We also have $M_j\leeta[\omega] U_j\geeta[\omega] \BT{N_j}$ for $j \le m$, 
$x_{m+j}\leeta[\omega] U_{m+j}\geeta[\omega] \BT{N_{m+j}}$ for $j\le p$ and finally
$z_{\ell}\leeta[\omega] V_{\ell}$ for $\ell\le k$.
By Theorem~\ref{thm:char_Hst}(ii), this entails that $\Hst\vdash M_j = N_j$ for $j \le m$ and $\Hst\vdash x_{m+j} = N_{m+j}$ for $j\le p$.
By hypothesis, there is an index $i\le m+p$ such that $N_i$ does not have a $\beta$-normal form. There are two subcases:
\begin{enumerate}
\item $0 <i\le m$. By the inductive hypothesis there exists $\sigma : M_i\nlem N_i$, so the Morris separator we are looking for is the position $\concat{\Tuple{i-1}}\sigma$.
\item $i > m$. By Theorem~\ref{thm:one-one}, there exists $T\in\Trees$ such that $\BTth\vdash \lam x_{i}.N_i = \bJ_T$, hence $\BTth\vdash N_i = \bJ_Tx_{i}$ and the Morris separator is $\emptyseq$. \qedhere
\end{enumerate}

\end{proof}

\begin{defi}
The following two combinators constitute the main ingredients to build our B\"ohm-out context:
$$
        \bU^n_k = \lambda x_1\dots x_n. x_k,
        \qquad\qquad
        \Tupler_n =  \lambda x_1 \dots x_n.[x_1, \dots, x_n].
$$
where we recall that $[M_1,\dots,M_n] = \lam z.zM_1\cdots M_n$ for some fresh variable $z$.
\end{defi}
The combinator $\bU^n_k$ is called \emph{the projector} and the combinator $\Tupler_n$ \emph{the tupler} since they enjoy the following properties. 
\begin{lem}\label{lemma:bcprop} Let $k\ge n\ge 0$ and $X_1,\dots,X_n,Y_1,\dots,Y_{k-n}\in\Lamo$.
\begin{enumerate}
\item\label{lemma:bcprop1} 
	$(\Tupler_kX_1\cdots X_n)Y_1\cdots Y_{k-n} =_\beta [X_1,\dots, X_n,Y_1,\dots, Y_{k-n}]$,
\item\label{lemma:bcprop2} 
	$[X_1,\dots, X_n]\,\bU^n_i =_\beta X_i$.
\end{enumerate}
\end{lem}

When $\bU^n_k$ is substituted for $y$ in $yM_1\cdots M_{n}$, 
it extracts an instance of the subterm $M_{k}$. 
Let us consider the \lam-term $N$ whose B\"ohm tree is depicted in Figure~\ref{fig:trees}.
The applicative context $[]\bU^3_1$ extracts from $N$ the subterm $yx$ where $x$ is replaced by $\bU^3_1$.
The idea of the B\"ohm-out technique is to replace every variable along the path $\sigma$ with the correct projector. 

An issue arises when the same variable occurs several times in $\sigma$
and we must select different children in these occurrences. 
For example, to extract $N_{\langle1,0\rangle}$ from the $N$ in Figure~\ref{fig:trees}, the first occurrence of $x$ should be replaced by $\bU^3_2$, the second by $\bU^1_1$ (which is actually $\bI$).

The problem was originally solved by B\"ohm in~\cite{bohm68} by first replacing the occurrences of the same variables along 
the path by different variables using the tupler, and then replacing each variable by the suitable projector. 
In the example under consideration, the applicative context 
$[]\Tupler_3\Omega\bU^3_2\bU^1_1\Omega\Omega\bU^3_1$ extracts from $N$ the instance of $N_{\langle1,0\rangle}$ where $z_0$ is replaced by $\bI$. 

Obviously, finite $\eta$-differences can be destroyed during the process of B\"ohming out. 
In contrast, we show that infinite $\eta$-differences can always be preserved.
In the following lemma we simply write $\subst{\seq y\,}{\Tupler_k}$ for the sequence of substitutions $\subst{y_1}{\Tupler_k}\cdots\subst{y_n}{\Tupler_k}$.

\begin{lem}[B\"ohm-out]\label{lemma:Bohmout}
Let $M,N\in\Lam$ be such that $\Hst\vdash M = N$, let $\seq y = \FV{MN}$ and $\sigma : M\nlem N$.
Then for all $k\in\nat$ large enough, there are combinators $\vec X\in\Lamo$ such that 
$$
	M\subst{\seq y\,}{\Tupler_k}\vec X =_\beta \bI\textrm{ and }N\subst{\seq y\,}{\Tupler_k}\vec X =_{\BTth} \bJ_T
$$ for some $T\in\Trees$.
\end{lem}

\begin{proof} We proceed by induction on $\sigma$.
\smallskip

\noindent {Base case} $\sigma = \emptyseq$. 
Then there exists $i>0$ such that, for some $p\ge i$, we have:
$$
	M =_\beta \lam x_1\dots x_{n}.yM_1\cdots M_{m}	\textrm{ and }
    N =_\beta \lam x_1\dots x_{n+p}.yN_1\cdots N_{m+p}	
$$
where $\BTth\vdash N_{m+i} = \bJ_T x_{n+i}$.
For any $k\ge n+m+p$ let us set 
$$
	\seq X = \Tupler_k^{\sim n}\bI^{\sim p}\Omega^{\sim k-m-p}\, \bU^k_{m+i},
$$ where we recall that $M^{\sim n}$ denotes the sequence of \lam-terms containing $n$ copies of $M$. 

We split into cases depending on whether $y$ is free or $y = x_j$ for some $j\le n$. We consider the former case, as the latter is analogous.
On the one side we have:
$$
	\begin{array}{ll}
	(\lam x_1\dots x_{n}.yM_1\cdots M_{m})\subst{\seq y\,}{\Tupler_k}\seq X =\\
	(\lam x_1\dots x_{n}.\Tupler_kM'_1\cdots M'_{m})\seq X =_\beta 
	&\textrm{ where } M'_\ell = M_\ell\,\subst{\seq y\,}{\Tupler_k},\\	
	(\Tupler_kM''_1\cdots M''_{m})\bI^{\sim p}\Omega^{\sim k-m-p}\,\bU^k_{m+i} =_\beta 
	&\textrm{ where } M''_\ell = M'_\ell\,\subst{\seq x\,}{\Tupler_k},\\		
	(\lam z.zM''_1\cdots M''_{m}\bI^{\sim p}\Omega^{\sim k-m-p})\bU^k_{m+i} =_\beta \bI
	&\textrm{ by Lemmas~\ref{lemma:bcprop}\eqref{lemma:bcprop1} and \ref{lemma:bcprop}\eqref{lemma:bcprop2}}.\\
	\end{array}
$$
On the other side, we have:
$$
	\begin{array}{ll}
	(\lam x_1\dots x_{n+p}.yN_1\cdots N_{m+p})\subst{\seq y\,}{\Tupler_k}\seq X =\\
	(\lam x_1\dots x_{n+p}.\Tupler_kN'_1\cdots N'_{m+p})\seq X =_\beta
	&\textrm{ for }N'_\ell = N_\ell\,\subst{\seq y\,}{\Tupler_k},\\
	(\lam x_{n+1}\dots x_{n+p}.\Tupler_kN''_1\cdots N''_{m+p})\bI^{\sim p}\Omega^{\sim k-m-p}\, \bU^k_{m+i} =_\beta 
	&\textrm{ for }N''_\ell = N'_\ell\,\subst{x_1,\dots,x_n}{\Tupler_k},\\
	(\Tupler_kN'''_1\cdots N'''_{m+p})\Omega^{\sim k-m-p}\, \bU^k_{m+i}=_\beta
	&\textrm{ for }N'''_\ell = N''_\ell\,\subst{x_{n+1},\dots,x_{n+p}}{\bI},\\
	(\lam z.zN'''_1\cdots N'''_{m+p}\Omega^{\sim k-m-p})\, \bU^k_{m+i}=_\beta
	&\textrm{ by Lemma~\ref{lemma:bcprop}\eqref{lemma:bcprop1},}\\		
	N'''_{m+i} = (\bJ_T x_{n+i})\subst{x_{n+i}}{\bI} = \bJ_T\bI
	&\textrm{ by Lemma~\ref{lemma:bcprop}\eqref{lemma:bcprop2},}\\
 =_\BTth \bJ_{T'}\textrm{ for some $T'\in\Trees$}	&\textrm{ by Thm.~\ref{thm:one-one} and Lemma~\ref{lemma:Texas}\ref{lemma:Texas2}}.\\
	\end{array}
$$
\smallskip

\noindent{Induction case} $\sigma = \Tuple{i}\star{\sigma'}$. \\
Since $\Hst\vdash M = N$, we can apply Theorem~\ref{thm:char_Hst}\ref{thm:char_Hst2} as in the proof of Proposition~\ref{prop:Morrissep} to show that $M,N$ have similar hnf's, namely, for $n-m = n' - m'$ and $i+1\le\min\{m,m'\}$ we have:
$$
	M =_\beta \lam x_1\dots x_{n}.yM_1\cdots M_{m}
	\textrm{ and }
	N =_\beta \lam x_1\dots x_{n'}.yN_1\cdots N_{m'}
$$
where $\Hst \vdash M_j = N_j$ for all $j \le \min\{m,m'\}$ and either $y$ is free or $y = x_j$ for $j\le \min\{n,n'\}$. 
Suppose that, say, $n \le n'$. 
Then there is $p\ge 0$ such that $n'= n+p$ and $m'=m+p$.
Since $\Hst\vdash M_{i+1} = N_{i+1}$ and $\sigma' : M_{i+1} \nlem N_{i+1}$ we
apply the induction hypothesis and get, for any $k$ large enough but in particular for $k > n+m+p$, a sequence
$\seq Y\in\Lamo$ such that 
$$M_{i+1}\subst{\seq y,\seq x\,}{\Tupler_{k}}\seq Y =_\beta \bI\textrm{ and }
N_{i+1}\subst{\seq y,\seq x\,}{\Tupler_{k}}=_\BTth \bJ_T\textrm{ for some }T\in\Trees.
$$

For this index $k$, let us set 
$$
	\seq X = \Tupler_k^{\sim n+p}\Omega^{\sim k-m-p}\,\bU^k_{i+1}\seq Y.
$$

\noindent We suppose that $y$ is free, the other case being analogous.
On the one side we have:
$$
	\begin{array}{ll}
	(\lam x_1\dots x_{n}.yM_1\cdots M_{m})\subst{\seq y\,}{\Tupler_k}\,\seq X =\\
	(\lam x_1\dots x_{n}.\Tupler_kM'_1\cdots M'_{m})\Tupler_k^{\sim n+p}\Omega^{\sim k-m-p}\,\bU^k_{i+1}\seq Y =_\beta
	&\textrm{ where }M'_\ell = M_\ell\,\subst{\seq y\,}{\Tupler_k},\\	
	(\Tupler_kM''_1\cdots M''_{m})\Tupler_k^{\sim p}\Omega^{\sim k-m-p}\,\bU^k_{i+1}\seq Y =_\beta	
	&\textrm{ where }M''_\ell = M'_\ell\,\subst{\seq x\,}{\Tupler_k},\\
	(\lam z.zM''_1\cdots M''_{m}\Tupler_k^{\sim p}\Omega^{\sim k-m-p})\,\bU^k_{i+1}\seq Y =_\beta
	&\textrm{ by Lemma~\ref{lemma:bcprop}\eqref{lemma:bcprop1}},\\
	M''_{i+1}\seq Y 	= M_{i+1}\subst{\seq y,\seq x\,}{\Tupler_k}\seq Y = _\beta&\textrm{ by Lemma~\ref{lemma:bcprop}\eqref{lemma:bcprop2}},\\				
	\bI&\textrm{ by the induction hypothesis.}\\				
	\end{array}
$$
On the other side, we have:
\begin{equation*}
	\begin{array}[b]{ll}
	(\lam x_1\dots x_{n+p}.yN_1\cdots N_{m+p})\subst{\seq y\,}{\Tupler_k}\seq X =\qquad\qquad\qquad\!\\
	(\lam x_1\dots x_{n+p}.\Tupler_kN'_1\cdots N'_{m+p})\seq X=_\beta
	&\textrm{where }N'_\ell = N_\ell\,\subst{\seq y\,}{\Tupler_k},\\	
	(\Tupler_kN''_1\cdots N''_{m+p})\Omega^{\sim k-m-p}\,\bU^k_{i+1}\seq Y=_\beta
	&\textrm{where }N''_\ell = N'_\ell\,\subst{\seq x\,}{\Tupler_k},\\		
	(\lam z.zN''_1\cdots N''_{m+p}\Omega^{\sim k-m-p})\,\bU^k_{i+1}\seq Y=_\beta
	&\textrm{by Lemma~\ref{lemma:bcprop}\eqref{lemma:bcprop1}},\\
	N''_{i+1}\seq Y = N_{i+1}\subst{\seq y,\seq x\,}{\Tupler_k}\seq Y =_\BTth 
		&\textrm{by Lemma~\ref{lemma:bcprop}\eqref{lemma:bcprop2}},\\
	\bJ_T&\textrm{by the induction hypothesis.}\\		
	\end{array}
  \tag*{\!\!\!\qedhere}
\end{equation*}
\end{proof}

Thanks to the Context Lemma we obtain the Morris-separation result in its full generality.
The following theorem first appeared in~\cite{BreuvartMPR16}.


\begin{thm}[Morris Separation]\label{thm:newsep}~\\
Let $M,N\in\Lambda$ be such that $\Hst\vdash M = N$ while $\Hpl\vdash M \not\sqle N$.
There is a closed head context $C[]$ such that 
$$
	C\hole{M} =_{\beta\eta} \bI\qquad\textrm{ and }\qquad C\hole{N}=_{\BTth} \bJ_T\textrm{ for some $T\in\Trees$.}
$$ 
When $M,N\in\Lamo$ the context $C[]$ can be chosen applicative.
\end{thm}

\begin{proof} 
By Lemma~\ref{lemma:ctx} for $\Hpl$, we get that $\Hpl\vdash M\not\sqle N$ entails the existence of a closed head context $C_2[]$ such 
that $C_2\hole M\in_\beta\NF[\beta]$, while $C_2\hole N$ does not have a $\beta$-normal form. 
From $\Hst\vdash M = N$ we obtain $\Hst\vdash C_2\hole{M} = C_2\hole{N}$.
Therefore, by Proposition~\ref{prop:Morrissep} there exists a Morris separator $\sigma : C_2\hole{M} \nlem C_2\hole{N}$.
We can now apply Lemma~\ref{lemma:Bohmout}, and get a head context $C_1[] = (\lam\seq y.[])\seq\Tupler_k\seq X$, where $\seq y = \FV{C_2\hole{M}C_2\hole{N}}$, and such that $C_1\hole{C_2\hole{M}} =_{\beta} \bI$ and $C_1\hole{C_2\hole{N}} =_\BTth \bJ_T$ for some $T\in\Trees$.
Hence the closed head context $C[]$ we are looking for is actually $C_1\hole{C_2[]}$.
When $M,N\in\Lamo$, all the contexts can be chosen applicative.
\end{proof}

\subsection{The \lam-Theory $\BTth\omega$ is Included in $\Hpl$}

The fact that $\Hpl$ satisfies the $\omega$-rule is an easy consequence of the Morris Separation Theorem.
Without loss of generality we can focus on $(\omega^0)$, namely the restriction of $\omega$ to closed \lam-terms.


\begin{lem}\label{lemma:key2} 
Let $M, N \in\Lamo$ be such that $M\in_\beta\NF[\beta]$, while $N$ does not have a $\beta$-normal form.
Then, there exist $k\ge 1$ and combinators $Z_1,\dots,Z_k\in\Lambda^o$ such that 
\begin{itemize}
\item $MZ_1\cdots Z_k\in_\beta\NF[\beta]$ while 
\item $NZ_1\cdots Z_k$ does not have a $\beta$-normal form.
\end{itemize}
\end{lem}
\begin{proof}
By hypothesis $\Hpl\vdash M\not\sqle N$. There are two possible cases.

Case $\Hst\vdash M\sqle N$. 
Since $M\in\NF[\beta]$, Theorem~\ref{thm:char_Hst}\ref{thm:char_Hst1} actually entails $\Hst\vdash M= N$. 
Therefore, by Theorem~\ref{thm:newsep} there are $Z_1,\dots,Z_k\in\Lambda^o$ and $T\in\Trees$
such that $MZ_1\cdots Z_k =_{\beta} \bI$ and $NZ_1\cdots Z_k =_{\BTth} \bJ_T$.
If $k = 0$ just take $Z_1 = \bI$ and conclude since $\bJ_T\bI  =_\BTth \bJ_{T'}$ for some $T'\in\Trees$ by Theorem~\ref{thm:one-one} and Lemma~\ref{lemma:Texas}\ref{lemma:Texas2}.

Case $\Hst\vdash M\not\sqle N$. 
By the usual semi-separability theorem, there are $Z_1,\dots,Z_k\in\Lambda^o$ such that $M\vec Z =_\beta \bI$ and $N\vec Z =_\cH \Omega$. 
When $k = 0$ we can take again $Z_1 = \bI$ since $\Omega\bI =_\cH \Omega$. 
\end{proof}

\begin{lem}\label{lemma:omega0}
 Let $M, N \in\Lambda^o$. 
If $\forall Z\in\Lambda^o, MZ =_{\Hpl} NZ$, then 
$$
	\forall \seq P\in\Lambda^o ( M\seq P 	\in_\beta\NF[\beta] \iff N\seq P \in_\beta\NF[\beta]).
$$
\end{lem}

\begin{proof} Suppose that for all $Z,\seq Q\in\Lambda^o$, $MZ\seq Q	\in_\beta\NF[\beta]$ if and only if $NZ\seq Q	\in_\beta\NF[\beta]$.
We show by induction on the length $k$ of $\seq P\in\Lambda^o$ that 
$M\seq P	\in_\beta\NF[\beta]$ if and only if $N\seq P\in_\beta\NF[\beta]$.

{Base: $k=0$}. Since the contrapositive holds by Lemma~\ref{lemma:key2}.

{Induction: $k>0$}. It follows trivially from the induction hypothesis.
\end{proof}

As a consequence, we get the main result of this section.

\begin{thm}\label{thm:Hplomega} 
$\Hpl$ satisfies the $\omega$-rule, therefore $\BTo\subseteq\Hpl$.
\end{thm}

\begin{proof}
Lemma~\ref{lemma:omega0} shows that $\Hpl\vdash \omega^0$, which is equivalent to $\Hpl\vdash \omega$ by Lemma~\ref{lemma:omega-props}\ref{lemma:omega-props2}.
By Lemma~\ref{lemma:omega-props}\ref{lemma:omega-props3}, $\BTth\subseteq \Hpl$ entails $\BTo\subseteq \Hpl\omega$ and from $\Hpl\vdash \omega$ we get $\Hpl\omega = \Hpl$.
\end{proof}


\section{Building B\"ohm Trees by Codes and Streams}\label{sec:streams}

Now that we have shown $\BTo\subseteq\Hpl$, we focus our efforts to prove the converse inclusion.
The key step in our proof is to show that when $\Hpl\vdash M = N$ holds the B\"ohm-like tree $U$ of Theorem~\ref{thm:CharacterizationHpl} giving the ``$\eta$-supremum'' of $\BT{M}$ and $\BT{N}$ can always be chosen to be the B\"ohm tree of a \lam-term $P$ (Proposition~\ref{prop:Hplchar}).
Intuitively, the \lam-term $P$ that we construct is going to inspect the structure of $M$ and $N$ looking at their codes and retrieve the correct $\eta$-expansion to apply from a suitable stream.

\subsection{Building B\"ohm Trees by Codes}
We start by showing that the B\"ohm tree of a \lam-term can be reconstructed from its code.

\begin{defi}\ 
\begin{itemize}
\item Let us fix an effective one-to-one encoding  $\# : \Lam\to \nat$ associating with every \lam-term $M$ its \emph{code} $\#M$ (the G\"odel number of $M$).
\item
	The \emph{quote} $\code{M}$ of $M$ is the corresponding numeral $ \Church{\# M}$. 
\end{itemize}
\end{defi}
We recall some well established facts from \cite[\S8.1]{Bare}.
By \cite[Thm.~8.1.6]{Bare}, there is a \lam-term $\bE\in\Lamo$ such that $\bE\code{M} = M$ for all $M\in\Lamo$. This is false in general for open \lam-terms $M$.

\begin{rem}\label{rem:effectiveness} The following operations are effective:
\begin{itemize}
\item from $\# M$ compute $\#M'$ where $M\redto[h] M'$ (since the head-reduction is an effective reduction strategy),
\item from $\num{\lam\seq x.yM_1\cdots M_k}$ compute $\#M_i$ for $i\le k$,
\item from $\# M$ compute $\num{\lam x_1\dots x_n.M}$ for $x_1,\dots,x_n\in\Var$.
\end{itemize}
\end{rem}
\noindent From Remark~\ref{rem:effectiveness}, it follows that the combinator~$\Phi$ below exists and can be defined using the fixed point combinator~$\bY$.

\begin{defi}\label{def:Phi}
Let $\Phi\in\Lamo$ be such that for all $M\in\Lamo$:
\begin{itemize}
\item if $M\msto[h] \lam  x_1\dots x_n.x_jM_1\cdots M_k$, then
$$\Phi \code{M} =_\beta \lam \fresh x_1\dots\fresh x_n.\fresh{x_j}(\Phi\code{\lam\seq x.M_1}\fresh x_1\cdots\fresh x_n)\cdots (\Phi\code{\lam\seq x.M_k}\fresh x_1\cdots\fresh x_n)
$$
where the $\fresh{x}_i$ are underlined to stress the fact that they are fresh variables.
\item $\Phi \code{M}$ is unsolvable whenever $M$ is unsolvable.
\end{itemize}
\end{defi}
The term $\Phi$ builds the B\"ohm tree of $M$ from its code $\#M$.
Notice that the closure $\lam\seq x.M_i$ on the recursive calls is needed to obtain a closed term (since $M\in\Lamo$ entails $\FV{M_i}\subseteq \seq x$).
In the definition above we use the fact that $\BT{\lam\seq x.M} = \lam\seq x.\BT{M}$, under the hypotheses that $\lam x.\bot = \bot$ and $\bot M =\bot$, thus the free variables $x_1,\dots,x_n$ can be reapplied externally. 

\begin{lem} For all $M\in\Lamo$, we have $\BTth\vdash \Phi\code{M} = M$.
\end{lem}
\begin{proof} We prove by coinduction the generalization of the statement to open \lam-terms, namely: for all $M\in\Lam$ and $\{y_1,\dots,y_m\}\supseteq\FV{M}$, we have $\BT{\Phi\code{\lam\seq y.M}y_1\cdots y_m} = \BT{M}$.

If $M$ is unsolvable then $\BT{\Phi\code{\lam\seq y.M}y_1\cdots y_m} = \BT{M} = \bot$.

Otherwise $M$ is solvable, so we have $M\msto[h] \lam x_1\dots x_n.yM_1\cdots M_k$ and $\Phi\code{\lam\seq y.M}\seq y =_\beta \lam x_1\dots x_n.y(\Phi\code{\lam\seq y\seq x.M_1}y_1\cdots y_mx_1\cdots x_n)\cdots (\Phi\code{\lam\seq y\seq x.M_k}y_1\cdots y_mx_1\cdots x_n)$.
We  conclude\linebreak since by coinductive hypothesis we have $\BT{\Phi\code{\lam\seq y\seq x.M_i}\seq y\seq x} = \BT{M_i}$ for every $i\le k$.
\end{proof}
The generalization to open terms above is formally needed to apply the coinductive hypothesis, but complicates the statements and the corresponding demonstrations.
Hereafter, it will be omitted and our proofs will be slightly more informal for the sake of readability.
%
%

\subsection{$\eta$-Expanding B\"ohm Trees from Streams}\label{subsec:building_etamax}

The existence of $\Phi$ should not be particularly surprising, Barendregt's enumerator $\bE$ being an example of such a \lam-term.
The recursive equation of Definition~\ref{def:Phi} has the advantage of exposing the structure of the B\"ohm tree and, doing so, opens the way for altering the tree.
For instance, it is possible to modify Definition~\ref{def:Phi} in order to obtain a \lam-term $\Psi$ which builds a finitary $\eta$-expansion of $\BT M$ starting from the code of $M$ and a stream of finite $\eta$-expansions of the identity.

\begin{defi}\label{def:streamETA}\ 
\begin{enumerate}
\item 
	Let $\ID$ be the stream containing infinitely many copies of $\bI$:
	$$
		\ID = [\bI,[\bI,[\bI,\dots]]]
	$$
\item
	Let $\seq\eta =(\eta_0,\eta_1,\eta_2,\dots)$ be an effective enumeration of all (closed) finite $\eta$-expansions of the identity, i.e.\ $\eta_i\in\ETAset\cap\Lamo$ for all $i\in\nat$. Define the corresponding stream 
	$$
		\ETA = [\eta_0,[\eta_1,[\eta_2,\dots]]].
	$$
\end{enumerate}
\end{defi}

Notice that, from now on, we consider the enumeration $\vec \eta$ and the stream $\ETA$ fixed.

In order to decide what $\eta$-expansion is applied at a certain position $\sigma$ in $\BT{M}$, we use a function $f(\sigma) = k$ and extract the $\eta$-expansion of index $k$ from $\ETA$ using the projection~$\pi_k$.
Since $f$ needs to be \lam-definable we consider $f$ computable ``after coding''. 
We recall that an effective encoding $\# : \nat^*\to \nat$ of all finite sequences of natural numbers has been fixed in Section~\ref{ssec:sequences} and that we denote by $\code{\sigma}$ the corresponding numeral $\Church{\#\sigma}$.

In the following definition $s$ is an arbitrary variable, but in practice we will always apply $\Psi_f \code{M}\code{\sigma}$ to some stream of combinators.

\begin{defi}\label{def:Psi}
Let $f: \nat^*\to\nat$ be a computable function, and $\Psi_f\in\Lamo$ be such that for all $M\in\Lamo$ and $\sigma\in\nat^*$:
\begin{itemize}
\item if $M\msto[h] \lam x_1\dots x_n.x_jM_1\cdots M_k$ then
	$$
		\Psi_f \code{M}\code{\sigma} s =_\beta \lam\seq{\fresh{x}}.\pi_{f(\sigma)}s\,\big(\fresh{x_j}(L_1\,\fresh{x}_1\cdots \fresh{x}_n)\cdots (L_k\,\fresh{x}_1\cdots \fresh{x}_n)\big)
	$$ where $L_i = \Psi_f\code{\lam\seq x.M_i}\code{\sigma.i}\,s$ for all $i\le k$ and $\pi_n$ denotes the $n$-th projection for streams,
\item 
	$\Psi_f \code{M} $ is unsolvable whenever $M$ is unsolvable.
\end{itemize}
\end{defi}

The actual existence of such a $\Psi_f$ follows from Remark~\ref{rem:effectiveness}, the effectiveness of the encodings and the fact that $f$ is computable.
We now verify that the \lam-term $\Psi_f \code{M}\code{\sigma}$ when applied to the stream $\ETA$ actually computes a finitary $\eta$-expansion of $\BT{M}$ in the sense of Definition~\ref{def:leeta}.

\begin{lem}\label{lemma:etaexpo}
 Let $f:\nat^*\to\nat$ be a computable function. 
For all $M\in\Lamo$ and $\sigma\in\nat^*$:
$$
	M\leeta  \Psi_f\code{M}\code{\sigma}\ETA.
$$
\end{lem}

\begin{proof} 
We prove by coinduction that $\BT{M}\leeta\BT{\Psi_f\code{M}\code{\sigma}\ETA}$.

If $M$ is unsolvable, then $\BT{M}=\BT{\Psi_f\code{M}\code{\sigma}\ETA}=\bot$. 

Otherwise $M\msto[h] \lam\seq x.x_jM_1\cdots M_k$.
Thus, for $f(\sigma) = q$ and $\eta_q =_\beta \lam yz_1\cdots z_m.yQ_1\cdots Q_m$ where $\lam z_i.Q_i\in\ETAset$ for all $i\le m$ we have:
$$
	\begin{array}{rl}
	\Psi_f\code{M}\code{\sigma}\ETA =_\beta& \lam\seq x.\pi_q\ETA(x_j(L_1\seq x)\cdots (L_k\seq x))\\
	=_\beta&\lam\seq x.\eta_q(x_j(L_1\seq x)\cdots (L_k\seq x))\\
	=_\beta&\lam\seq x z_1\dots z_m.x_j(L_1\seq x)\cdots (L_k\seq x)Q_1\cdots Q_m\\
	\end{array}
$$ 
where $L_i = \Psi_f\code{\lam\seq x.M_i}\code{\sigma.i}\ETA$ for all $i\le k$.
We  conclude because by coinductive hypothesis we get for each $i\le k$ that $\BT{\lam \seq x.M_i}\leeta  \BT{L_i}$ holds.
\end{proof}

Since $\Psi_f$ picks the $\eta$-expansion to apply from the input stream, we can retrieve the behaviour of $\Phi$ by applying $\ID$.

\begin{lem}\label{lemma:barbatrucco}
Let $f:\nat^*\to\nat$ be computable. 
For all $M\in\Lamo$ and $\sigma\in\nat^*$, we have 
$$
	\BTth\vdash M = \Psi_f\code{M}\code{\sigma}\ID.
$$
\end{lem}

\begin{proof}
We prove by coinduction that $\BT{M}=\BT{\Psi_f\code{M}\code{\sigma}\ID}$.

If $M$ is unsolvable, then $\BT{M}=\BT{\Psi_f\code{M}\code{\sigma}\ID}=\bot$. 

Otherwise $M\msto[h] \lam\seq x.x_jM_1\cdots M_k$.
Since $\pi_q\ID =_\beta \bI$ for all $q\in\nat$ we have:
$$
	\begin{array}{rl}
	\Psi_f\code{M}\code{\sigma}\ID =_\beta& 
	\lam\seq x.\pi_{f(\sigma)}\ID(x_j(L_1\seq x)\cdots (L_k\seq x))\\
	=_\beta&\lam\seq x.\bI(x_j(L_1\seq x)\cdots (L_k\seq x))\\
	=_\beta&\lam\seq x.x_j(L_1\seq x)\cdots (L_k\seq x)\\
	\end{array}
$$
where $L_i = \Psi_f\code{\lam\seq x.M_i}\code{\sigma.i}\ID$ for all $i\le k$.
We  conclude because by coinductive hypothesis we get for each $i\le k$ that $\BT{\lam \seq x.M_i} = \BT{L_i}$ holds.
\end{proof}

We would like to draw the attention to the fact that the \lam-term $\Psi_f$ does not inspect the structure of the input stream to select one of its elements, it blindly relies on the function~$f$.
This is the key property that opens the way to modify the behaviour of $\Psi_f$ by merely changing the input stream as in Lemma~\ref{lemma:barbatrucco}.
We mention this fact since the \lam-term $\etamax[]$ that we define in the next section satisfies a similar property and we exploit it in a crucial way.

\subsection{Building the $\eta$-Supremum}\label{subsec:building_etasup}

Using similar techniques, we now define a \lam-term $\etamax[]$ that builds from the codes of $M,N$ and the stream $\ETA$ the (smallest) $\eta$-supremum satisfying 
$$
	M\leeta \etamax[]\code{M}\code{N}\ETA\geeta  N
$$
when such an $\eta$-supremum exists, that is whenever $M$ and $N$ are compatible (Proposition~\ref{prop:Hplchar}).
Intuitively, at every common position $\sigma$, the \lam-term $\etamax[]$ needs to compare the structure of the subterms of $M,N$ at $\sigma$ and apply the correct expansion $\eta_i$ taken from~$\ETA$.

Rather than proving that there exists a computable function $f : \nat^*\to\nat$ associating to every $\sigma$ the corresponding $\eta_i$ (which can be tedious) we use the following property of $\vec\eta = (\eta_0,\eta_1,\dots,\eta_i,\dots)$: since every closed $\eta$-expansion $Q\in \ETAset$ is $\beta$-normalizable and the enumeration $\vec \eta$ is effective, it is possible to decide starting from the code $\#Q$ the index $i$ of $Q$ in $\vec \eta$.
Moreover, it is possible to choose such an index $i$ minimal.

\begin{lem}\label{lemma:iotacomp}
There exists a computable function $\iota : \nat\to\nat$ such that, for all $M\in\Lamo$, 
if $M =_{\beta} \eta_i$ and $M\neq_\beta \eta_k$ for all $k < i$ then $\iota(\# M) = i$ .
\end{lem}

\begin{proof} Let $\delta(m,n)$ be the partial computable map satisfying for all normalizing $M,N\in\Lamo$: 
$$
\delta(\#M,\#N) = 
	\begin{cases}
	0&\textrm{if $M$ and $N$ have the same $\beta$-normal form,}\\
	1&\textrm{otherwise.}\\
	\end{cases}
$$
Then $\iota$ can be defined as $\iota(n) := \mu k.\delta\big(\#(\pi_k \circ \vec\eta\,),n\big) = 0$. 
I.e., $\iota(n) = k$ if $k$ is the smallest index such that $\delta(\#\eta_k,n) = 0$, equivalently $\nf(\eta_k) = \nf(\bE\, \Church n)$; otherwise $\iota(n) =\  \uparrow$.
\end{proof}

From now on we consider fixed such a function $\iota$, which depends on the enumeration $\vec\eta$ generating the stream $\ETA$.

\begin{defi}\label{def:le_h} For $M,N\in\Lam$, we define:
\begin{enumerate}[label={(\roman*)}]
\item\label{def:le_h1}
	$M\le_h N$ whenever $M\msto[h] \lam\seq x.y M_1\cdots M_k$ and $N\msto[h]\lam\seq xz_1\dots z_m.y N_1\cdots N_kQ_1\cdots Q_m$ with $\lam z_i.Q_i\in\ETAset$ for all $i\le m$,
\item\label{def:le_h2} 
	$M\sim_h N$ if both $M\le_h N$ and $N\le_h M$ hold,
\item\label{def:le_h3} 
	$M<_h N$ if $M\le_h N$ holds but $M\not\sim_h N$.
\end{enumerate}
\end{defi}

Whenever $M \le_h N$ holds, we say that $N$ \emph{looks like an $\eta$-expansion of $M$}.
This does not necessarily mean that it actually is, as shown by the following examples. 

\begin{exas}
We have:
\begin{itemize}
\item
	 $\lam z.x\bF z \le_h \lam z.x\bK z$ since we do not require that $\bF\le_h \bK$ holds;
\item $z \le_h \lam z. z z$ since in Definition~\ref{def:le_h}\ref{def:le_h1} we do not check that $z_i\notin\FV{\BT{y\seq M \seq N}}$.
\end{itemize}
\end{exas}
Therefore, compared with the relation $\leeta $ of Definition~\ref{def:leeta}, the relation $\le_h$ is weaker since it lacks the coinductive calls and the occurrence check on the $z_i$'s.
This is necessary to ensure the following semi-decidability property.

\begin{rem}\label{rem:effectiveness2}
The property $M\le_h N$ can be semi-decided by the following procedure: 
\begin{itemize}
\item first head-reduce in parallel $M,N$ until they reach a head normal form,
\item if both reductions terminate, compare the two hnf's and check whether they have the shape of Definition~\ref{def:le_h}\ref{def:le_h1},
\item then semi-decide whether $Q_i\msto[\beta\eta] z_i$ for all $i\le m$.
\end{itemize}
This procedure might fail to terminate when $M\not\le_h N$.
\end{rem}

By Remarks~\ref{rem:effectiveness} and \ref{rem:effectiveness2} and the fact that $\iota$ is computable (Lemma~\ref{lemma:iotacomp}), the \lam-term $\etamax$ below exists.
\begin{defi}\label{def:etamax}
Let $\iota : \nat\to\nat$ be the computable function from Lemma~\ref{lemma:iotacomp}. 
We define a combinator $\etamax\in\Lamo$ such that for all $M,N\in\Lamo$:
\begin{enumerate}
\item 
	if $M\le_h N$ with 
	$$
		M \msto[h] \lam  x_1\dots x_n. x_j M_1 \cdots M_k\textrm{ and }N\msto[h]  \lam x_1\dots x_n  z_1\dots z_m. x_j N_1 \cdots N_k Q_1\cdots Q_m
	$$ 
	then 
	$$
		\etamax \code{N} \code{M} s = \lam \fresh{x_1}\dots \fresh{x_n}. \pi_q s \big(\fresh{x_j} (\Upsilon_1 \fresh{x_1}\cdots\fresh{x_n}) \cdots (\Upsilon_k \fresh{x_1}\cdots\fresh{x_n})\big)
	$$
	where $q =\iota\big(\#\NF[\beta](\lam \fresh{y} z_1 \dots z_m.\fresh{y} Q_1 \cdots Q_m)\big)$ and $\Upsilon_i = \etamax \code{\lam x_1\dots x_n.M_i} \code{\lam x_1\dots x_n.N_i} s$.\smallskip
\item
	if $N <_h M$ then $\etamax \code{N} \code{M} s = \etamax \code{M} \code{N} s$,\smallskip
\item
	otherwise, $\etamax \code{N} \code{M} s$ is unsolvable.
\end{enumerate}
\end{defi}

In the definition of~$\etamax$ there are some subtleties that deserve to be discussed. 
The fact that $Q_i\msto[\beta\eta] z_i$ for all $i \le m$, although not explicitly written, is a consequence of $M\le_h N$.
\emph{A priori} $\lam z_i.Q_i\in\ETAset$ might be open (consider for instance $\lam z_i.\bK z_i y\redto[\beta]\bI$) but its $\beta$-normal form is always a closed \lam-term. 
This is the reason why we compute $\NF[\beta](\lam y\seq z.y\seq Q)$ before applying $\iota$ to its code.
In particular, $\iota$ is defined on all the codes $\num{\NF[\beta](\lam y\seq z.y\seq Q)}$.

The following commutativity property follows from the second condition of Definition~\ref{def:etamax} and should be natural considering the fact that $\etamax\code{M}\code{N}\ETA$ is supposed to compute the $\eta$-join of $\BT{M}$ and $\BT{N}$, which is a commutative operation.

\begin{lem}\label{lemma:etamax_comm}
For all $M, N\in \Lambda^o$, we have:
$$
	\BTth\vdash\etamax \code{M} \code{N} = \etamax \code{N} \code{M}.
$$
\end{lem}

\begin{proof} We prove by coinduction that $\BT{\etamax \code{M} \code{N}} = \BT{\etamax \code{N} \code{M}}$.

If $M,N$ are unsolvable or neither $M \le_h N$ nor ${N \le_h M}$ holds, then 
$\BT{\etamax \code{M} \code{N}}= \BT{\etamax \code{N} \code{M}} = \bot$.

The cases $M<_h N$ and $N <_h M$ follow directly from the definition.

If $M \sim_h N$, then we have $M\msto[h]\lam x_1\dots x_n.x_jM_1\cdots M_k$ and $N\msto[h]\lam x_1\dots x_n.x_jN_1\cdots N_k$.
Since $\pi_q s \msto s\,\bF^{\sim q}\,\bK$ we have
$$
	\begin{array}{l}
\etamax \code{M} \code{N}s =_\beta
	\lam\seq x.s\,\bF^{\sim q}\,\bK(x_j(\Upsilon_1\seq x) \cdots (\Upsilon_k \seq x))\\
\etamax \code{N} \code{M}s =_\beta
	\lam\seq x.s\,\bF^{\sim q}\,\bK(x_j(\Upsilon'_1 \seq x) \cdots (\Upsilon'_k \seq x))\\
	\end{array}
$$
where, for all $i\le k$, we have  $\Upsilon_i = \etamax\code{\lam\seq  x.M_i}\code{\lam\seq  x.N_i}s$ and $\Upsilon'_i = \etamax\code{\lam\seq  x.N_i}\code{\lam\seq x.M_i}s$.
We conclude since, by coinductive hypothesis, we get $\BT{\Upsilon_i\seq x} = \BT{\Upsilon'_i\seq x}$ for all $i\le k$.
\end{proof}

Another property that we expect is that whenever $M\leeta  N$ the \lam-term $\etamax\code{M}\code{N}\ETA$ computes the B\"ohm tree of $N$.

\begin{lem}\label{lem:Nwins}
For all $M,N\in\Lamo$, if $M \leeta  N$ then 
$$
	\BTth\vdash \etamax \code{M}\code{N}\ETA = N.
$$
\end{lem}

\begin{proof}  We prove by coinduction that $\BT{\etamax \code{M}\code{N}\ETA} = \BT{N}$.

If $M,N$ are both unsolvable, then also $\etamax \code{M}\code{N}$ must be and their B\"ohm trees are $\bot$.

Otherwise, we have $M \msto[h] \lam\seq x.x_jM_1\cdots M_k$ and $N \msto[h]\lam\seq xz_1\dots z_m.x_jN_1\cdots N_kQ_1\cdots Q_m$ where each $z_\ell\notin\FV{\BT{x_j\seq M\seq N}}$, $\lam z_\ell.Q_\ell\in\ETAset$ for all $\ell\le m$ and $M_i\leeta  N_i$ for all $i\le k$.
In particular $M\le_h N$ holds, so the first condition of Definition~\ref{def:etamax} applies.

From $\lam z_\ell.Q_\ell\in\ETAset$ it follows that $\lam y\seq z.y\seq Q\in\ETAset$, therefore $\iota(\#\NF[\beta](\lam y\seq z.y\seq Q)) = q$ for some index $q$. 
Setting $\Upsilon_i = \etamax\code{\lam\seq x.M}\code{\lam\seq x.N}\ETA$, easy calculations give:
$$
	\begin{array}{lcl}
	\etamax\code{M}\code{N} \ETA &=_\beta&
	\lam\seq x.\pi_{q}\ETA(x_j(\Upsilon_1\seq x)\cdots (\Upsilon_k\seq x))\\
	&=_\beta&\lam\seq x.(\lam y\seq z.y\seq Q)(x_j(\Upsilon_1\seq x)\cdots (\Upsilon_k\seq x))\\	
	&=_\beta&\lam\seq x\seq z.x_j(\Upsilon_1\seq x)\cdots (\Upsilon_k\seq x)Q_1\cdots Q_m
	\end{array}
$$
We conclude since, by coinductive hypothesis, we have $\BT{\Upsilon_i} = \BT{\lam\seq x.N_i}$ for all $i\le k$.
\end{proof}

Under the assumption $M \leeta  N$ we can also use $\etamax$ to retrieve the B\"ohm tree of $M$ by applying the stream $\ID$.

\begin{lem}\label{lem:Mwins}
For all $M,N\in\Lamo$, if $M \leeta  N$ then 
$$
	\BTth\vdash \etamax \code{M}\code{N}\ID = M.
$$
\end{lem}

\begin{proof}  
Analogous to the proof of Lemma~\ref{lemma:barbatrucco}.
Since $\iota$ has been defined in terms of the enumeration~$\vec\eta$, $\iota(\#Q)$ still provides an index $q$ such that $\pi_q\ETA =_\beta Q$ but when applied to $\ID$ it necessarily gives $\pi_q\ID = \bI$.
%
\end{proof}


\section{$\BTth\omega = \Hpl$ and Sall\'e's Conjecture is False}\label{sec:Salleiswrong}

This section is devoted to proving that $\BTo = \Hpl$  holds (Theorem~\ref{thm:main}). 
As mentioned earlier, the first step is to show that the term $\etamax$ defined in the previous section, when applied to $\Hpl$-equivalent \lam-terms, computes a common $\eta$-upper bound.

\begin{prop}\label{prop:Hplchar}
For all $M,N\in\Lamo$, 
$$
 	\Hpl\vdash M=N \iff M \leeta  \etamax\code{M}\code{N}\ETA\geeta  N
$$
\end{prop}

\begin{proof}
$(\Leftarrow)$ It follows directly from Theorem~\ref{thm:CharacterizationHpl}.

$(\Rightarrow)$ By Theorem~\ref{thm:CharacterizationHpl}, we known that there exists a B\"ohm-like tree $U\in\BTset$ such that $\BT{M}\leeta  U \geeta  \BT{N}$. 
As usual, we proceed by coinduction on the B\"ohm(-like) trees.

If $M$ or $N$ is unsolvable then $\BT{M} = \BT{N}  = U = \bot$.

Otherwise, from $\BT M \leeta  U\geeta \BT  N$ we have, say:
$$
	\begin{array}{c}
	M \msto[h] \lam\seq x.x_jM_1\cdots M_k,\qquad\qquad
	N \msto[h] \lam\seq xz_1\dots z_m.x_jN_1\cdots N_kQ_1\cdots Q_m,\\
	U =  \lam\seq xz_1\dots z_m\dots z_{m'}.x_jU_1\cdots U_kQ'_1\cdots Q'_m\cdots Q'_{m'},\\
	\end{array}
$$
such that 
\begin{itemize}
\item
	$z_1,\dots,z_{m'}\notin\FV{x_j\BT{M_1}\cdots\BT{M_k}\,U_1\cdots U_k}$, 
\item 
	$z_{m+1},\dots,z_{m'}\notin\FV{x_j\BT{N_1}\cdots\BT{N_k}}$, 
\item 
	$\BT{M_i}\leeta  U_i$ and $\BT{N_i}\leeta  U_i$ for all $i\le k$, 
\item
	$Q_\ell\leeta  Q'_\ell$ for all $\ell\le m$, 
\item
	and $\lam z_{\ell'}.Q'_{\ell'}\in\ETAset$ for all $\ell' > m$. 
\end{itemize}
By Lemma~\ref{lemma:eta-equiv-defs}, we have $ Q'_\ell\msto[\beta\eta]
z_\ell$ so $Q_\ell\leeta Q'_\ell$ entails $\lam
z_\ell.Q_\ell\in\ETAset$, hence $N$ looks like an $\eta$-expansion of
$M$. (That is, $M\le_h N$ holds.)

Put $q = \iota\big(\num{\nf(\lam yz_1\dots z_m.yQ_1\cdots Q_m)}\big)$,
$\Upsilon_i = \etamax\code{\lam\seq x.M}\code{\lam\seq x.N}\ETA$.  We have
$$
	\begin{array}{rl}
	\etamax\code{M}\code{N}\ETA =_\beta&
	\lam\seq x.\pi_q\ETA(x_j(\Upsilon_1\seq x)\cdots (\Upsilon_k\seq x))\\
	=_\beta&
	\lam\seq x.(\lam y\seq z.yQ_1\cdots Q_m)(x_j(\Upsilon_1\seq x)\cdots (\Upsilon_k\seq x))\\
	=_\beta&	
	\lam\seq x\seq z.x_j(\Upsilon_1\seq x)\cdots (\Upsilon_k\seq x)Q_1\cdots Q_m. \\
	\end{array}
$$
This case follows by coinductive hypotheses since, for all ${i\le k}$,  $\lam\seq x.M_i\leeta  \Upsilon_i\geeta  \lam\seq x.N_i$.

The symmetric case $N <_h M$ is treated analogously, using Lemma~\ref{lemma:etamax_comm} to apply the coinductive hypotheses.
\end{proof}

In the proof above, it is easy to check that $T=\BT{\etamax\code{M}\code{N}\ETA}$ is minimally $\eta$-expanded so that it satisfies $M\leeta T \geeta  N$, thus $T$ is actually the $\eta$-supremum of $\BT{M}$ and $\BT N$.
The second step towards the proof of Theorem~\ref{thm:main} is to show that the streams $\ID$ and $\ETA$ are equated in $\BTo$.
To prove this result, we are going to use the following auxiliary streams.
\begin{defi}\ Using the fixed point operator $\bY$, define the streams:
\begin{itemize}
\item $\IDo\, yx  = [yx,[y\Omega x,[y\Omega^{\sim 2}x,[y\Omega^{\sim 3}x,\dots]]]]$,
\item $\ETAo\, yx =  [y(\eta_0x),[y\Omega(\eta_1x),[y\Omega^{\sim 2}(\eta_2x),[y\Omega^{\sim 3}(\eta_3x),\dots]]]]$.
\end{itemize}
\end{defi}

The streams $\IDo$ and $\ETAo$ are equal in $\BTo$, for the same reason the \lam-terms $P,Q$ of Figure~\ref{fig:PQ} are. 
The formal reasoning is the following.

\begin{lem}\label{lem:BToIDoETAo}
 $\BTo\vdash\IDo = \ETAo$.
\end{lem}

\begin{proof}
Let $M\in\Lamo$, by Lemma~\ref{lemma:Omegak} there exists $k\in\nat$ such that $M\Omega^{\sim k} =_\BTth\Omega$. So we have:
$$
	\begin{array}{rll}
	\IDo M &=_\BTth& \lam x.[Mx,[M\Omega x,[\dots,[M\Omega^{\sim k-1}x,[\Omega,\dots]]]]]\\
	&=_\BTth& \lam x.[M(\bI x),[M\Omega(\bI x),[\dots,[M\Omega^{\sim k-1}(\bI x),[\Omega,\dots]]]]]\\	
	&=_{\beta\eta}& \lam x.[M(\eta_0x),[M\Omega(\eta_1x),[\dots,[M\Omega^{\sim k-1}(\eta_{k-1}x),[\Omega,\dots ]]]]]\\
	&=_\BTth& \ETAo M,
	\end{array}
$$
where the third equality follows from $\bI =_{\beta\eta} \eta_i$ for all $i\in\nat$.
Since $M$ is an arbitrary closed \lam-term, we can apply the $\omega$-rule and conclude that $\BTo\vdash \IDo =\ETAo$.
\end{proof}

As the variable $y$ occurs in head-position in the \lam-terms occurring in the stream $\ETAo yx$ (resp. $\IDo yx$), we can substitute for it a suitably modified projection that erases the $\Omega$'s and returns the $n$-th occurrence of $x$ in $\IDo$ (resp.\ $(\eta_n x)$ in $\ETAo$).

\begin{lem}\label{lem:IDoeqETAo}
 There exists a closed \lam-term $\Eq$ such that, for all $n\in\nat$:
$$
	\Eq\, \Church{n}\IDo =_\BTth \bI,\qquad\qquad	
	\Eq\, \Church{n}\ETAo =_\BTth \eta_n.
$$
\end{lem}
\begin{proof} Let $\Eq$ be a \lam-term satisfying the recursive equation
$$
\Eq\, n\, s = \ifz(n,\lam z.s\bI z\bK,\Eq\, (\pred\, n)\, (\lam
zw.s(\bK z)w\bF)).
$$
By induction on $n$, we show 
$$
	\Eq\, \Church{n}(\lam
yx.[y\Omega^{\sim i}(\eta_{i+k}x)]_{i\in\nat}) =_\BTth \eta_{n+k}
$$ for all $n,k\in\nat$.
Note that $\eta_i\in\ETAset$ entails $\lam z.\eta_iz =_\beta \eta_i$.

If $n =0 $ then 
\begin{align*}
&\hspace{1.1cm}\Eq\, \Church{0}(\lam yx.[y\Omega^{\sim i}(\eta_{i+k}x)]_{i\in\nat})\hspace{5cm}\\
&=_\beta\quad\lam z.(\lam yx.[y\Omega^{\sim i}(\eta_{i+k}x)]_{i\in\nat})\bI z\bK\\
&=_\beta\quad \lam z.[\Omega^{\sim i}(\eta_{i+k}z)]_{i\in\nat}\bK\\
&=_\beta\quad \lam z.\bK\Omega^{\sim 0} (\eta_{0+k}z) [\Omega^{\sim i+1} (\eta_{i+1+k}z)]_{i\in\nat}\\
&=_\beta\quad \lam z.\eta_kz \quad =_\beta\quad \eta_k.
\end{align*}

If $n=n'+1$, then IH yields $\Eq\,\Church{n'}(\lam y x. [y\Omega^{\sim
  i}(\eta_{i+k}x)]_{i \in \nat})=\eta_{n'+k}$ for all $k$.  Thus
\begin{align*}
&\hspace{1.1cm} \Eq\, \Church{n'+1}(\lam y x. [y \Omega^{\sim i}(\eta_{i+k}x)]_{i \in
  \nat}) \\
&=_\beta\quad
 \Eq\,\Church{n'}(\lam zw.\left(\lam y x. [y \Omega^{\sim i}(\eta_{i+k}x)]_{i \in
  \nat}\right) (\bK z) w \bF) \\
&=_\beta\quad \Eq\, \Church{n'}(\lam zw. \left[(\bK z) \Omega^{\sim
  0}(\eta_{0+k}w), [(\bK z) \Omega^{\sim i+1}(\eta_{i+1+k}w)]_{i \in
  \nat}\right]\bF)\\
&=_\beta\quad \Eq\, \Church{n'}(\lam zw.[z
\Omega^{\sim i}(\eta_{i+1+k}w)_{i \in \nat}])\\
&=_\alpha\quad
\Eq\, \Church{n'}(\lam y x. [y \Omega^{\sim i}(\eta_{i+(k+1)}x)]_{i \in \nat})
\quad =_{\text{IH}}\quad \eta_{k+1}.
\end{align*}
Analogous calculations show $\Eq\, \Church{n}\IDo =_\BTth \bI$.
\end{proof}
%
%

\begin{cor}\label{cor:IDeqETA}
 $\BTo \vdash \ID = \ETA$.
\end{cor}

\begin{proof} From Lemmas~\ref{lem:IDoeqETAo}, \ref{lem:eqimpBTeq} and \ref{lem:BToIDoETAo} we get:
\begin{equation*}
	\ID  =_\BTth [\Eq\, \Church{n}\IDo]_{n\in\nat} =_{\BTo} [\Eq\,
        \Church{n}\ETAo]_{n\in\nat} =_\BTth \ETA.
        \tag*{\qedhere}
\end{equation*}
\end{proof}

In Section~\ref{subsec:building_etamax} we have seen that, when $M\leeta  N$ holds, the \lam-term $\etamax\code{M}\code{N}$ computes the B\"ohm tree of $N$ from~$\ETA$ (Lemma~\ref{lem:Nwins}) and the B\"ohm tree of $M$ from $\ID$ (Lemma~\ref{lem:Mwins}), but now we have proved that $\ETA =_\BTo\ID$.
As a consequence, we get that $M$ and $N$ are equal in $\BTo$.

\begin{thm}\label{thm:main}
$\BTo = \Hpl$.
\end{thm}

\begin{proof}
$(\subseteq)$ This inclusion was shown in Theorem~\ref{thm:Hplomega}.

$(\supseteq)$ 
By Remark~\ref{rem:closed_terms}, it is enough to consider $M,N\in\Lamo$.
If $\Hpl\vdash M = N$, then by Proposition~\ref{prop:Hplchar} we have $M\leeta  P\geeta  N$ for $P= \etamax \code M \code N\ETA$.
Then we have:
$$
\begin{array}{rll}
	M\  =_\BTth& \etamax\code{M}\code{P}\ID&\textrm{by Lemma~\ref{lem:Mwins}}\\
	=_{\BTo}&\etamax\code{M}\code{P}\ETA&\textrm{by Corollary~\ref{cor:IDeqETA}}\\
	=_\BTth&P&\textrm{by Lemma~\ref{lem:Nwins}}\\	
	=_\BTth&\etamax\code{N}\code{P}\ETA&\textrm{by Lemma~\ref{lem:Nwins}}\\
	=_{\BTo}&\etamax\code{N}\code{P}\ID&\textrm{by Corollary~\ref{cor:IDeqETA}}\\
	=_\BTth& N&\textrm{by Lemma~\ref{lem:Mwins}}\\
	\end{array}
$$
We conclude that $\BTo\vdash M=N$.
\end{proof}

This theorem constitutes a refutation of Sall\'e's conjecture and settles one of the few open problems left in Barendregt's book~\cite{Bare}.


\section{A Characterization of $\BTe$}\label{sec:BTeta}
As discussed in Section~\ref{subsec:BTeta}, $\BTe$ equates strictly less than $\BTo$, and by Theorem~\ref{thm:main} than $\Hpl$.
However, to the best of our knowledge, in the literature there is no formal characterization of $\BTe$ in terms of some extensional equality between B\"ohm trees.
The only property shown in \cite[\S16.4]{Bare} concerning the interaction between $\BT-$ and  $\redto[\eta]$ is~Lemma~\ref{lemma:Bareta}.

In this section we show that the approach described in
Section~\ref{subsec:building_etasup} is general enough to prove that
$\BTe\vdash M = N$ holds exactly when $\BT{M}$ and $\BT{N}$ are equal
up to countably many $\eta$-expansions having \emph{uniformly bounded size} (Theorem~\ref{thm:main2}).

\subsection{Bounded $\eta$-Expansions of B\"ohm Trees}

We start by defining the \emph{size} of a finite $\eta$-expansion $Q$ of the identity.
Intuitively, the size of $Q$ is the maximum between its height and its width (namely, the maximal number of branching in its B\"ohm tree).

\begin{defi}
The \emph{size} of $Q\in\ETAset$ is defined inductively as follows:
$$
	\mes{Q} = \begin{cases}
	0&\textrm{if }Q\msto\bI,\\
	\max\{m,\max_{i \le m}\{\mes{\lam z_i.Q_i}\}\!+1\}&\textrm{if }Q\msto\lam yz_1\dots z_m.yQ_1\cdots Q_m.\\
	\end{cases}
$$
\end{defi}
For $Q,Q'\in\ETAset$, we have that $Q =_\beta Q'$ entails $\mes{Q} = \mes{Q'}$ since the size $\mes-$ is determined by their $\beta$-normal forms.
It is easy to verify that for each $n\in\nat$ the size of $\One^n$ is $\mes{\One^n} = n$.

\begin{defi}
For $p\in\nat$, let $\ETAset[p] = \{Q\in\ETAset\st \mes{Q} < p\}$ be the set of $\eta$-expansions of $\bI$ whose size is bounded by~$p$.
\end{defi}
We have $\ETAset[0]=\emptyset$, $\ETAset[1]=\{\bI\}$, $\ETAset[2]=\{\bI,\lam yz.yz\}$, and so on. 
Therefore $\ETAset[p]\subseteq\ETAset[p+1]$ and the sequence $\big(\ETAset[p]\big)_{p\in\nat}$ is increasing.

The following lemma will be useful for studying the behaviour of $\etamax$, when applied to the codes of \lam-terms whose B\"ohm trees differ because of bounded $\eta$-expansions.

\begin{lem}\label{lemma:Q_p}
Let $m,p\in\nat$.
If $m\le p$ and $\lam z_i.Q_i\in\ETAset[p]$ holds for all $i\le m$, then we have that $\lam yz_1\dots z_m.yQ_1\cdots Q_m\in\ETAset[p+1]$.
\end{lem}

\begin{proof}
Since $m \le p$ and $\mes{\lam z_i.Q_i} < p$ for each $i\le m$, it follows that $\mes{\lam yz_1\dots z_m.yQ_1\cdots Q_m}=\max\{m,\max\{\mes{\lam z_1.Q_1},\dots,\mes{\lam z_m.Q_m}\}+1\} \le p < p+1$.
\end{proof}

Now that we have formalized when $Q\in\ETAset$ has size bounded by a certain~$p$, we specify when two B\"ohm-like trees $U,V$ are such that $V$ is an $\eta$-expansion of $U$ bounded by $p$.


\begin{defi}
For all $p\in\nat$, we define the greatest relation $\leeta[p]$ between B\"ohm-like trees such that $U\leeta[p] V$ entails that:
\begin{itemize}
\item either $U = V= \bot$, 
\item or for some $m\le p$
$$
	U = \lam\seq x.y U_1\cdots U_k
	\textrm{ and }
	V = \lam\seq xz_1\dots z_m.y V_1\cdots V_kQ_1\cdots Q_m,
$$
where $z_\ell\notin\FV{yU_1\cdots U_kV_1\cdots V_k}$, $\lam z_i.Q_i\in\ETAset[p]$ for all $i\le m$ and $U_j\leeta[p] V_j$ for all $j\le k$.
\end{itemize}
\end{defi}

Remark that in the definition above we verify not only that the size of each $\lam z_i.Q_i$ is bounded by $p$, but also that their number~$m$ is.
Notice the asymmetry between the strict bound $\mes{\lam z_i.Q_i} < p$ and the bound $m\le p$, which arises naturally from Lemma~\ref{lemma:Q_p}.
\begin{nota}
For $M,N\in\Lam$ and $p\in\nat$, we write $M\leeta[p]N$ whenever $\BT{M}\leeta[p]\BT{N}$.
\end{nota}

The next lemma follows straightforwardly from the definition.

\begin{lem}\label{lemma:selfevident}
For $M,N\in\Lam$ and $p\in\nat$. If $M\leeta[p] N$ then:
\begin{enumerate}[label={(\roman*)}]
\item\label{lemma:selfevident1}
	$M\leeta  N$,
\item\label{lemma:selfevident2}
	$\BTth\vdash M = M'$ entails $M'\leeta[p] N$,
\item\label{lemma:selfevident3}
	for all $p' \ge p$, we have $M\leeta[p'] N$.
\end{enumerate}
\end{lem}

We need a couple of technical lemmas.
The first one exhibits the interaction between $\leeta[p]$ and the size~${\mes-}$, the intuition being that $\leeta[p]$ can increase the size of $Q\in\ETAset$ by at most~$p$.

\begin{lem}\label{lemma:sclero}\ 
Given $p,p'\in\nat$, we have that $\lam y.Q\in\ETAset[p]$ and $Q\leeta[p'] Q'$ imply $\lam y.Q'\in\ETAset[p+p']$.
\end{lem}

\begin{proof} We proceed by induction on the structure of $\nf_\beta(Q)$. 

From $\lam y.Q\in\ETAset[p]$ we get $\nf_\beta(Q)= \lam z_1\dots z_m.yQ_1\cdots Q_m $ for some $m < p$ and $Q_i$'s such that $\mes{\lam z_i.Q_i} < p$.
Since $Q\leeta[p'] Q'$ holds we have $\nf_\beta(Q') = \lam z_1\dots z_mw_1\dots w_{m'}.yQ'_1\cdots Q'_{m+m'}$ where $m'\le p'$, $Q_i\leeta[p'] Q'_i$ for all $i\le m$ and $\lam w_j.Q'_{j}\in\ETAset[p']$ for all $j > m$.
Therefore $m + m' < p+p'$, $\mes{\lam w_j.Q'_{j}} < p'$ for all $j > m$ and, by inductive hypothesis, $\mes{\lam z_i.Q'_i} < p+p'$ for all $i\le m$.
We conclude that $\mes{\lam y.Q'} < p+p'$, which entails $\lam y.Q'\in\ETAset[p+p']$.
\end{proof}

The next lemma is devoted to showing that the relation $\leeta[p]$ enjoys the following ``weighted'' transitivity property.

\begin{lem}\label{lemma:etasum} For all $M,N\in\Lam$, $p_1,p_2\in\nat$ we have that $M \leeta[p_1] P$ and $P\leeta[p_2] N$ entail $M\leeta[p_1+p_2] N$.
\end{lem}

\begin{proof} We proceed by coinduction on their B\"ohm trees.

If one among $M,N,P$ is unsolvable, then all their B\"ohm trees are $\bot$ and we are done.

Otherwise from $M \leeta[p_1] P$ we get 
$$
	M\msto[h]\lam\seq x.yM_1\cdots M_k, 
	\textrm{ and }
	P\msto[h]\lam\seq x z_1\dots z_{m_1}.yP_1\cdots P_kQ_1\cdots Q_{m_1}
$$ 
where $\seq z\notin\FV{\BT{y\seq M\seq P}}$, $m_1\le p_1$, $\lam z_i.Q_i\in\ETAset[p_1]$ for $i\le m_1$ and ${M_j\leeta[p_1] P_j}$ for $j\le k$.

From $P \leeta[p_2] N$ we obtain (for $\seq w\notin\FV{\BT{y\seq P\seq Q\seq N\seq Q'}}$):
$$
	N\msto[h]\lam\seq x z_1\dots z_{m_1} w_1\dots w_{m_2}.yN_1\cdots N_kQ'_1\cdots Q'_{m_1}Q''_1\cdots Q''_{m_2}
$$ 
with $m_2\le p_2$, $P_j\leeta[p_2] N_j$ for $j\le k$, $Q_i\leeta[p_2] Q'_i$ for $i\le m_1$ and $\lam w_\ell.Q''_\ell\in\ETAset[p_2]$ for $\ell\le m_2$.\linebreak
We notice that $m_1\le p_1$ and $m_2\le p_2$ imply $m_1+m_2 \le p_1+p_2$, and that $\ETAset[p_1]\cup\,\ETAset[p_2]\subseteq \ETAset[p_1+p_2]$.\linebreak
From $\lam z_i.Q_i\in\ETAset[p_1]$ and $Q_i \leeta[p_2] Q'_i$ we obtain by Lemma~\ref{lemma:sclero} that $\lam z_i.Q'_i\in\ETAset[p_1+p_2]$ for $i\le m_1$.
Since, for all $j\le k$, $M_j\leeta[p_1] P_j$ and $P_j\leeta[p_2] N_j$ we conclude that $M_j\leeta[p_1+p_2] N_j$ holds by applying the coinductive hypotheses.
\end{proof}

The following is an easy corollary of Lemma~\ref{lemma:Bareta}.

\begin{cor}\label{cor:fromBaresLemma}
If $M\msto[\eta] N$ then there exists a bound $p\in\nat$ such that $N\leeta[p] M$.
\end{cor}

\begin{proof} We perform an induction loading and prove that if the reduction $M\msto[\eta] N$ has length $p$, then $N\leeta[p] M$ holds.

If $p = 0$ then $\BT{M} = \BT{N}$ so we have $N \leeta[0] M$.

If $p >0$ then $M\msto[\eta] N'$ in $p-1$ steps and $N'\redto[\eta] N$.
From Lemma~\ref{lemma:Bareta}, $N\leeta[1] N'$ since at every position $\sigma$ of their B\"ohm trees the lengths of $x_1,\dots,x_n$ and $x_1,\dots,x_nz$ differ at most by 1, and $\lam z.z\in\ETAset[1]$.
By induction hypothesis we have $N'\leeta[p-1] M$, so we conclude by applying Lemma~\ref{lemma:etasum}.
\end{proof}

\subsection{The Behaviour of $\etamax$ on Bounded $\eta$-Expansions}

Recall that in Definition~\ref{def:streamETA} we have fixed an effective enumeration $\vec\eta = (\eta_0,\eta_1,\dots)$ of the set $\ETAset$, together with the corresponding stream $\ETA = [\eta_0,[\eta_1,[\eta_2,\dots]]]$.
Moreover, we consider fixed a map $\iota$ satisfying the properties of Lemma~\ref{lemma:iotacomp}.
We start proving the following technical lemma concerning bounded $\eta$-expansions of the identity.

\begin{lem}\label{lem:old_claim}
Let $Q,Q'\in\Lam$, $\seq z \supseteq \FV{QQ'}$. 
If $\lam z_\ell.Q,\lam z_\ell.Q'\in\ETAset[p]$ for some index $\ell$ then 
\[
	Q\leeta[p]\etamax\code{\lam\seq z.Q}\code{\lam\seq z.Q'}\ETA\,\seq z\geeta[p]Q'.
\]
\end{lem}

\begin{proof}
Since $Q,Q'$ are normalizing, we assume they are in $\beta$-normal form and proceed by structural induction. 
We have, say:
$$
	Q = \lam\seq y.z_\ell Q_1\cdots Q_m\quad
	Q' = \lam\seq y\seq w.z_\ell Q'_1\cdots Q'_{m+m'}
$$
where $m+m'\le p$, $\lam y_i.Q_i,\lam y_i.Q'_i\in\ETAset[p]$ for $i\le m$ and $\lam w_j.Q'_j\in\ETAset[p]$ for $j >m$.

We split into two subcases.
\begin{itemize}
\item
	If $m = 0$ then $Q = z_\ell$ and $\lam \seq z.z_\ell\leeta[p] \lam\seq z.Q'$, so by Lemma~\ref{lem:Nwins} we get that $\BTth\vdash \lam\seq z.Q' = \etamax\code{\lam\seq z.Q}\code{\lam\seq z.Q'}\ETA$ and the case follows by applying the variables $\seq z$ to both sides.
\item
Otherwise $m > 0$, so for $q =\iota(\num{\lam x\seq w.x Q'_{m+1}\cdots Q'_{m+m'}})$ we have:
$$
	\begin{array}{rcl}
	\etamax\code{\lam \seq z.Q}\code{\lam \seq z.Q'}\ETA\seq z 
	&=_\beta& 
	\lam \seq y.\pi_q\ETA(z_\ell(\Upsilon_1\seq z\seq y)\cdots(\Upsilon_m\seq z\seq y))\\
	&=_\beta&
	\lam \seq y.(\lam x\seq w.x Q'_{m+1}\cdots Q'_{m+m'})(z_\ell(\Upsilon_1\seq z\seq y)\cdots(\Upsilon_m\seq z\seq y))\\
	&=_\beta&
	\lam \seq y w.z_\ell(\Upsilon_1\seq z\seq y)\cdots(\Upsilon_m\seq z\seq y)Q'_{m+1}\cdots Q'_{m+m'}
	\end{array}
$$
where  $\Upsilon_i = \etamax\code{\lam\seq z\seq y.Q_i}\code{\lam\seq z\seq y.Q'_i}\ETA$ and the case follows by induction hypothesis, thus concluding the proof of the lemma.\qedhere
\end{itemize}
\end{proof}

As shown in Section~\ref{sec:Salleiswrong}, given two \lam-terms $M$ and $N$ whose B\"ohm trees differ because of countably many $\eta$-expansions, the \lam-term $\etamax\code{M}\code{N}$ of Definition~\ref{def:etamax} builds their $\eta$-supremum from the stream $\ETA$ (Proposition~\ref{prop:Hplchar}).
We now prove that when the size of such $\eta$-expansions is bounded by $p$, then the B\"ohm trees of $M,N$ also differ from $\BT{\etamax\code{M}\code{N}\ETA}$ because of $\eta$-expansions bounded by $p$.

\begin{lem}\label{lemma:supM1M2p}
Let $M,N,P\in\Lamo$. If $M\geeta[p]P\leeta[p] N$, then:
$$
	M\leeta[p] \etamax \code{M}\code{N}\ETA\geeta[p] N.
$$
\end{lem}
\begin{proof} We proceed by coinduction on their B\"ohm trees.

If $P$ is unsolvable, then also $M,N$ and $\etamax \code{M}\code{N}$ are unsolvable and we are done.

Otherwise $P\msto[h] \lam\seq x.x_jP_1\cdots P_k$ and from $P\leeta[p] M$ we obtain for some $m\le p$:
$$
	M\msto[h] \lam\seq x z_1\dots z_m.x_jM_1\cdots M_kQ_1\cdots Q_m
$$ 
where $P_i\leeta[p] M_i$ for $i\le k$ and $\lam z_\ell.Q_\ell\in\ETAset[p]$ for $\ell\le m$.
Similarly, from $P\leeta[p] N$ we get, say, 
$$
	N\msto[h] \lam\seq xz_1\dots z_mw_1\dots w_{m'}.x_jN_1\cdots N_kQ'_1\cdots Q'_{m+m'}
$$ where $m+m'\le p$, $P_i\leeta[p] N_i$ for $i\le k$, $\lam z_\ell.Q'_\ell\in\ETAset[p]$ for $\ell\le m$ and $\lam w_{\ell}.Q'_{m+\ell}\in\ETAset[p]$ for $\ell\le m'$, so we have $M\le_h N$.
(Notice that the case $N<_h M$ is symmetrical, one just needs to apply Lemma~\ref{lemma:etamax_comm}.)
For $q = \iota(\num{\NF[\beta](\lam y\seq w.yQ'_{m+1}\cdots Q'_{m+m'})})$, we have that:
$$
\begin{array}{rcl}
\etamax\code{M}\code{N}\ETA &=_\beta&
\lam\seq x\seq z.\pi_q\ETA (x_j(\Upsilon_1\,\seq x\seq z\,)\cdots(\Upsilon_k\,\seq x\seq z\,)(\Upsilon'_1\,\seq x\seq z\,)\cdots(\Upsilon'_m\,\seq x\seq z\,))\\
 &=_\beta&
\lam\seq x\seq z\seq w.x_j(\Upsilon_1\seq x\seq z\,)\cdots(\Upsilon_k\seq x\seq z\,)(\Upsilon'_1\seq x\seq z\,)\cdots(\Upsilon'_m\seq x\seq z\,)Q'_{m+1}\cdots Q'_{m+m'}
\end{array}
$$
where $\Upsilon_i = \etamax\code{\lam\seq x\seq z.M_i}\code{\lam\seq x\seq z.N_i}\ETA$ for $i\le k$, and~$\Upsilon'_\ell = \etamax\code{\lam\seq x\seq z.Q_\ell}\code{\lam\seq x\seq z.Q'_\ell}\ETA$ for $\ell \le m$.
By applying Lemma~\ref{lem:old_claim}, we get $Q_\ell\leeta[p]
\Upsilon'_\ell\seq x\seq z\geeta[p] Q'_\ell$.  By coinductive
hypotheses, we get $M_i\leeta[p]\Upsilon_i\seq x\seq z \geeta[p]N_i$.
Since $m\le m+m'\le p$ and $Q'_{m+\ell} \in\ETAset[p]$ we conclude that $M\leeta[p] \etamax \code{M}\code{N}\ETA\geeta[p] N$ holds.
\end{proof}

\subsection{A New Characterization of the Equality in $\BTe$}

We are now ready to provide a characterization of the equality $M =_\BTe N$ in terms of equality between $\BT{M}$ and $\BT{N}$ up to bounded $\eta$-expansions.
The key idea we exploit is the fact that, under these hypotheses, $\etamax\code{M}\code{N}$ only depends on a finite restriction of its input stream $S$ and that all finite restrictions of $\ETA$ and $\ID$ of the same length are $\beta\eta$-convertible with each other (since they are finite). 

\begin{defi} Let $S = [S_i]_{i\in\nat}$ be a stream of \lam-terms.
For $n\in\nat$, define the \emph{$n$-truncation of $S$} as the following sequence:
$$
	\restr{S\,}n\ =\ [S_0,[S_1,[\,\cdots,[S_n,\Omega]\cdots]]]
$$
\end{defi}
It is easy to check that the $i$-th projection $\pi_i$ defined for $S$ also works on $\restr{S\,}n$ for all $i\le n$.
Notice that the $\Omega$ at the end of the sequence $\restr{S}n$ does not have a profound meaning: we just need to have an $n+1$-component since $\pi_n = \lam y.y\bF^{\sim n}\bK$ needs something to erase in that position.

We have seen in Lemma~\ref{lemma:iotacomp} that $\iota(\#Q)$ corresponds to the \emph{smallest} index $i$ such that $Q$ occurs in $\ETA$.
The following property is a consequence of such a  minimality condition.

\begin{lem}\label{lemma:iotabounded}
Let $p\in\nat$. There exists an index $n$ such that for every closed \lam-term $Q\in\ETAset[p]$ we have $\iota(\#Q) \le n$.
\end{lem}

\begin{proof} Since the size of each $Q\in\ETAset[p]$ is bounded by $p$, the set $\nf_\beta(\ETAset[p]) = \{\nf_\beta(Q)\st Q\in\ETAset[p]\}$ is finite.
By Lemma~\ref{lemma:iotacomp}, $\iota(\#Q)$ gives the smallest index $i$ such that $Q =_\beta\eta_i$, therefore the set $\iota[\ETAset[p]] =  \{ \iota(\#Q) \st Q \in\ETAset[p]\cap\Lamo\}$ is finite. 
We conclude by considering $n = \max(\iota[\ETAset[p]])$.
\end{proof}

As a corollary we get that if $M\leeta[p] N$ holds, then $\etamax\code{M}\code{N}$ only uses a finite portion of its input stream.

\begin{cor}\label{cor:restricting_stream}
For all $M,N\in\Lamo$, if $M\leeta[p] N$ then there exists an index $n\in\nat$ such that for every stream $S = [S_i]_{i\in\nat}$ we have:
$$
	\BTth\vdash\etamax\code{M}\code{N} S = \etamax\code{M}\code{N} (\restr{S} n).
$$
\end{cor}
\begin{proof} By Lemma~\ref{lemma:iotabounded}, there exists $n$ such that for every $Q\in\ETAset[p+1]\cap\Lamo$ we have $\iota(\#Q) \le n$.
We prove the statement by coinduction for that particular $n$.

If $M$ or $N$ are unsolvable, than so is $\etamax\code{M}\code{N}$ and we are done.

Otherwise, we have $M\msto[h] \lam\seq x.x_jM_1\cdots M_k$ and $N\msto[h] \lam\seq x\seq z.x_jN_1\cdots N_kQ_1\cdots Q_m$ for $m\le p$ and $\lam z_\ell.Q_\ell\in\ETAset[p]$ for all $\ell\le m$. 
By Lemma~\ref{lemma:Q_p}, the \lam-term defined as $Q = \nf_\beta(\lam y\seq z.yQ_1\cdots Q_m)$ belongs to the set $\ETAset[p+1]$.
Hence, for some index $ q = \iota(\#Q) \le n$, we have on the one side:
$$
	\begin{array}{rl}
	\etamax\code{M}\code{N}S =_\beta&
	\lam \seq x.\pi_qS(x_j(\Upsilon_1\seq x)\cdots(\Upsilon_k\seq x))\\
	=_\beta&
	\lam \seq x.S_q(x_j(\Upsilon_1\seq x)\cdots(\Upsilon_k\seq x))\\
	\end{array}	
$$
where $\Upsilon_i = \etamax\code{M_i}\code{N_i}{S}$. On the other side, we have:
$$
	\begin{array}{rl}
	\etamax\code{M}\code{N}(\restr{S\,} n) =_\beta&
	\lam \seq x.\pi_q(\restr{S\,} n)(x_j(\Upsilon'_1\seq x)\cdots(\Upsilon'_k\seq x))\\
	=_\beta&
	\lam \seq x.S_q(x_j(\Upsilon'_1\seq x)\cdots(\Upsilon'_k\seq x))\\
	\end{array}
$$
where $\Upsilon'_i = \etamax\code{M_i}\code{N_i}(\restr{S\,}{n})$.
We can conclude since, by coinductive hypothesis, $\BT{\Upsilon_i} = \BT{\Upsilon'_i}$ holds for all $i\le k$.
\end{proof}
Since $M =_\BTe N$ holds exactly when there is an alternating sequence of shape $M=_\BTth M_0 =_\eta M_1 =_\BTth\cdots =_\eta M_k =_\BTth~N$, ``zig-zag'' sequences like the one in the following lemma naturally arise (for instance, in the proof of Theorem~\ref{thm:main2}).

\begin{lem}\label{lemma:zigozago}
Let $M_1,\dots, M_{k+1}\in\Lamo$ and let $p\in\nat$. 
If, for some $N_i\in\Lamo$, there is a zig-zag sequence such that:
\begin{center}
\begin{tikzpicture}
\node at (-4,0) {};
\node at (0,0) {$~M_1\qquad\qquad M_2 \quad \cdots\quad M_{k-1}\qquad\qquad M_k\qquad\qquad M_{k+1}$};
\node at (0pt,-25pt) {$N_1\qquad\qquad N_2 \qquad \cdots\quad N_{k-1}\qquad \quad\  N_k\ \ $};
\node[rotate=-40] at (-112pt,-13pt) {$\geeta[p]$};
\node[rotate=40] at (-82pt,-13pt) {$\leeta[p]$};
\node[rotate=-40] at (-52pt,-13pt) {$\geeta[p]$};
\node[rotate=40] at (-24pt,-13pt) {$\leeta[p]$};
\node[rotate=-40] at (13pt,-13pt) {$\geeta[p]$};
\node[rotate=40] at (43pt,-13pt) {$\leeta[p]$};
\node[rotate=-40] at (73pt,-13pt) {$\geeta[p]$};
\node[rotate=40] at (103pt,-13pt) {$\leeta[p]$};
\end{tikzpicture}
\end{center}
then there exists $P\in\Lamo$ such that $M_1\leeta[kp]P\geeta[kp] M_k$.
\end{lem}

\begin{proof} We proceed by induction on $k$.

Case $k =0$. Trivial, just take $P = M_1$.


Case $k>0$. By applying Lemma~\ref{lemma:supM1M2p} to each pair $M_i$ and $M_{i+1}$ we get, setting $N'_i = \etamax\code{M_i}\code{M_{i+1}}\ETA$, the sequence:
$$
	M_1\leeta[p] N'_1\geeta[p]M_2\cdots \leeta[p]N'_k\geeta[p]M_{k+1}.
$$
As the subsequence $N'_1\geeta[p]M_2\cdots \leeta[p]N'_k$ is shorter and satisfies the hypotheses of the lemma, we get from the inductive hypothesis a \lam-term $P\in\Lamo$ such that $N'_1\leeta[(k-1)p] P\geeta[(k-1)p]N'_k$.
We conclude by Lemma~\ref{lemma:etasum} since $(k-1)p + p = kp$.
\end{proof}

We are now able to provide the following characterization of $\BTe$,
which constitutes the second main result of the paper.

\begin{thm}\label{thm:main2}
For all $M,N\in\Lamo$, the following are equivalent:
\begin{enumerate}
\item $\BTe\vdash M = N$,
\item there exist $P\in\Lamo$ and $p\in\nat$ such that $M\leeta[p] P \geeta[p] N$.
\end{enumerate}
\end{thm}

\begin{proof}
$(1\Rightarrow 2)$ Since $\BTe$ is the join of two congruences, namely $=_\BTth$ and $=_\eta$, we have that $\BTe\vdash M = N$ holds if and only if there are $M_0,\dots,M_k\in\Lamo$ such that
$M=_\BTth M_0 =_\eta  M_1 =_\BTth \cdots =_\eta M_{k}  =_\BTth N$ (cf.~\cite[Thm.~4.6]{BurrisS81}).
Since $\eta$-reduction is Church-Rosser and the $M_j$'s are closed, for even indices $i$ (as odd indices correspond to $=_\BTth$ steps), we have $M_i =_\eta M_{i+1}$ if and only if $M_{i}\msto[\eta] N_i\invmsto[\eta]M_{i+1}$ for some $N_i\in\Lamo$. By Corollary~\ref{cor:fromBaresLemma} there are $p_{i},q_{i}$ such that $M_i\geeta[p_i]N_i\leeta[q_i] M_{i+1}$. 
By Lemma~\ref{lemma:selfevident}\ref{lemma:selfevident3}, setting $p' = \max_{i}\{p_i,q_i\}$ we get $M_i\geeta[p']N_i\leeta[p'] M_{i+1}$.
By applying Lemma~\ref{lemma:selfevident}\ref{lemma:selfevident2} we can get rid of the equality $=_\BTth$ and obtain the sequence:
$$
	M\geeta[p'] N_1\leeta[p']M_2\geeta[p']N_2\cdots N_{k-1}\leeta[p'] N.
$$
Therefore this implication follows from Lemma~\ref{lemma:zigozago}.

$(2\Rightarrow 1)$ We assume that $M\leeta[p] P \geeta[p] N$ holds.
From Lemma~\ref{lemma:selfevident}(i), we obtain $M\leeta P\geeta N$ as well. 
Therefore (for some $k,k'\in\nat$):
$$
\begin{array}{rll}
\quad\	M\  =_\BTth& \etamax\code{M}\code{P}\ID&\textrm{by Lemma~\ref{lem:Mwins}}\\
	=_{\BTth}&\etamax\code{M}\code{P}(\restr{(\ID)}{k})&\textrm{by Corollary~\ref{cor:restricting_stream}}\\
	=_{\beta\eta}&\etamax\code{M}\code{P}(\restr{(\ETA)}{k})&\textrm{by $\beta\eta$-conversion}\\	
	=_{\BTth}&\etamax\code{M}\code{P}\ETA&\textrm{by Corollary~\ref{cor:restricting_stream}}\\
	=_\BTth&P&\textrm{by Lemma~\ref{lem:Nwins}}\\	
	=_\BTth&\etamax\code{N}\code{P}\ETA&\textrm{by Lemma~\ref{lem:Nwins}}\\
	=_{\BTth}&\etamax\code{N}\code{P}(\restr{(\ETA)}{k'})&\textrm{by Corollary~\ref{cor:restricting_stream}}\\	
	=_{\beta\eta}&\etamax\code{N}\code{P}(\restr{(\ID)}{k'})&\textrm{by $\beta\eta$-conversion}\\		
	=_{\BTth}&\etamax\code{N}\code{P}\ID&\textrm{by Corollary~\ref{cor:restricting_stream}}\\
	=_\BTth& N&\textrm{by Lemma~\ref{lem:Mwins}}\qquad\\
	\end{array}
$$
We conclude that $\BTe\vdash M = N$.
\end{proof}

This result confirms informal intuition about $\BTe$ discussed in \cite[\S16.4]{Bare}.


\section{Conclusion} \label{sec:Con}
Refutation of Sall\'e's Conjecture provides the final picture
of relationships between the classical lambda theories.
The next theorem should substitute Theorem~17.4.16 in \cite{Bare}.

\begin{thm} The following diagram indicates all possible inclusion relations of the \lam-theories involved (if $\cT_1$ is above $\cT_2$, then $\cT_1\subsetneq \cT_2$):
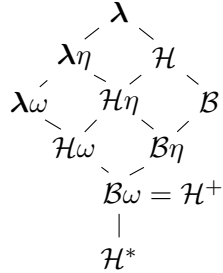
\begin{center}
  \begin{tikzpicture}
  	\node (root) at (0,0) {};
	\node (blam) at (root) {$\blam$};
	\node (lameta) at ($(blam)+(-17pt,-17pt)$) {$\blam\eta$};			
	\node (H) at ($(blam)+(17pt,-17pt)$) {$\cH$};			
	\node (Heta) at ($(blam)+(0,-34pt)$) {$\cH\eta$};	
	\node (Homega) at ($(Heta)+(-17pt,-17pt)$) {$\cH\omega$};		
	\node (BTeta) at ($(Heta)+(18pt,-18pt)$) {$\BTth\eta$};			
	\node (lamomega) at ($(blam)+(-34pt,-34pt)$) {$\blam\omega$};			
	\node (BT) at ($(blam)+(34pt,-34pt)$) {$\BTth$};		
	\node (BTomega) at ($(Heta)+(0,-34pt)$) {$\qquad\quad\cB\omega=\Hpl$};		
	\node (Hst) at ($(BTomega)+(0,-25pt)$) {$\Hst$};	
	\draw (blam) -- ($(lameta.north)+(4pt,-2pt)$);
	\draw (blam) -- ($(H.north)-(4pt,1pt)$);	
	\draw (lameta) -- ($(lamomega.north)+(4pt,-2pt)$);
	\draw ($(lameta.south east)+(-3pt,2pt)$) -- ($(Heta.north)-(5pt,1pt)$);	
	\draw (H) -- ($(Heta.north)+(4pt,-2pt)$);
	\draw (H) -- ($(BT.north)-(4pt,1pt)$);		
	\draw ($(lamomega.south)+(6pt,1pt)$) -- ($(Homega.north)-(5pt,1pt)$);
	\draw (BT) -- ($(BTeta.north)-(-4pt,2pt)$);	
	\draw ($(Heta.south west)+(3pt,2pt)$) -- ($(Homega.north)+(4pt,-2pt)$);
	\draw ($(Heta.south east)+(-3pt,2pt)$) -- ($(BTeta.north)-(5pt,1pt)$);		
	\draw (Homega) -- ($(BTomega.north)-(4pt,1pt)$);
	\draw ($(BTeta.south west)+(2pt,2pt)$) -- ($(BTomega.north)-(-4pt,2pt)$);	
	\draw (BTomega) -- (Hst);
  \end{tikzpicture}
\end{center}
\end{thm}

Fortuitously, the technique of constructing the $\eta$-supremum of two terms
effectively from their codes also yields a characterization of equality in
the theory $\BTe$.
Together these results illuminate a ``spectrum'' of degrees of
extensionality in the theory of B\"ohm trees:\\

\renewcommand{\arraystretch}{1.4}
{\centering \begin{tabular}{@{}l l l@{}} \toprule
Theory & Syntactic characterization &Sample equality\\ \midrule
$\BTth$ & $\BT{M}=\BT{N}$ & $\IDX = \Phi \code{\IDX}$\\
$\BTe$ & $\BT{M} \leeta[p] U \geeta[p] \BT{N}$ & $\IDX = \ONE$\\ 
$\BTth\omega$ & $\BT{M} \leeta U \geeta \BT{N}$ & $\IDo = \ETAo$\\
$\Hpl$& $\BT{M} \leeta U \geeta \BT{N}$ &  $\IDX = \ETAge$ \\
$\Hst$& $\BT{M} \leetainf U \geetainf \BT{N}$ & $\IDX = \JAY$ \\ \bottomrule
\end{tabular} \par}
\[ \BTth \quad \subsetneq\quad \BTe\quad \subsetneq\quad
\BTth\omega\quad =\quad \Hpl\quad \subsetneq\quad \Hst \]


\section*{Acknowledgment}
The authors would like to thank Henk Barendregt, Flavien Breuvart, Mariangiola Dezani-Ciancaglini, Jean-Jacques L\'evy, Simona Ronchi della Rocca and Domenico Ruoppolo for interesting discussions on $\Hpl$.

\bibliographystyle{plain}
\bibliography{include/biblio}

\newcommand{\online}[1]{Available at \url{#1}}
\begin{thebibliography}{10}

\bibitem{BarendregtTh}
Henk~P. Barendregt.
\newblock {\em Some extensional term models for combinatory logics and
  $\lambda$-calculi}.
\newblock {Ph.D.} thesis, Utrecht Universiteit, the Netherlands, 1971.

\bibitem{Barendregt77}
Henk~P. Barendregt.
\newblock The type free lambda calculus.
\newblock In J.~Barwise, editor, {\em Handbook of Mathematical Logic},
  volume~90 of {\em Studies in Logic and the Foundations of Mathematics}, pages
  1091--1132. North-Holland, Amsterdam, 1977.

\bibitem{Bare}
Henk~P. Barendregt.
\newblock {\em The lambda-calculus, its syntax and semantics}.
\newblock Number 103 in Studies in Logic and the Foundations of Mathematics.
  North-Holland, second edition, 1984.

\bibitem{BarendregtBKV78}
Henk~P. Barendregt, Jan~A. Bergstra, Jan~Willem Klop, and Henri Volken.
\newblock Degrees of sensible lambda theories.
\newblock {\em Journal of Symbolic Logic}, 43(1):45--55, 1978.

\bibitem{bohm68}
Corrado B\"ohm.
\newblock Alcune propriet\`a delle forme $\beta$-$\eta$-normali nel
  $\lambda$-{$K$}-calcolo.
\newblock {\em INAC}, 696:1--19, 1968.

\bibitem{Breuvart14}
Flavien Breuvart.
\newblock On the characterization of models of $\mathcal{H}^{*}$.
\newblock In T.~Henzinger and D.~Miller, editors, {\em {CSL-LICS}'14}, pages
  24:1--24:10. ACM, 2014.

\bibitem{BreuvartMPR16}
Flavien Breuvart, Giulio Manzonetto, Andrew Polonsky, and Domenico Ruoppolo.
\newblock New results on {M}orris's observational theory.
\newblock In Delia Kesner and Brigitte Pientka, editors, {\em Formal Structures
  for Computation and Deduction}, volume~52 of {\em LIPIcs}, pages 15:1--15:18.
  Schloss Dagstuhl, 2016.

\bibitem{BreuvartMR18}
Flavien Breuvart, Giulio Manzonetto, and Domenico Ruoppolo.
\newblock Relational graph models at work.
\newblock {\em Logical Methods in Computer Science}, 14(3), 2018.

\bibitem{BrownP16}
Matt Brown and Jens Palsberg.
\newblock Breaking through the normalization barrier: a self-interpreter for
  {F}-omega.
\newblock In Rastislav Bod{\'{\i}}k and Rupak Majumdar, editors, {\em
  Proceedings of the 43rd Annual {ACM} {SIGPLAN-SIGACT} Symposium on Principles
  of Programming Languages, {POPL} 2016}, pages 5--17. {ACM}, 2016.

\bibitem{BurrisS81}
Stanley~N. Burris and Hanamantagouda~P. Sankappanavar.
\newblock {\em A course in universal algebra}.
\newblock Springer-Verlag, Berlin, 1981.

\bibitem{CoppoDR78}
Mario Coppo, Mariangiola Dezani{-}Ciancaglini, and Simona Ronchi~Della Rocca.
\newblock ({S}emi)-separability of finite sets of terms in {S}cott's
  $\mathcal{D}_\infty$-models of the lambda-calculus.
\newblock In Giorgio Ausiello and Corrado B{\"{o}}hm, editors, {\em Automata,
  Languages and Programming, Fifth Colloquium, Udine, Italy}, volume~62 of {\em
  Lecture Notes in Computer Science}, pages 142--164. Springer, 1978.

\bibitem{CoppoDS79}
Mario Coppo, Mariangiola Dezani{-}Ciancaglini, and Patrick Sall{\'{e}}.
\newblock Functional characterization of some semantic equalities inside
  lambda-calculus.
\newblock In Hermann~A. Maurer, editor, {\em Automata, Languages and
  Programming}, volume~71 of {\em Lecture Notes in Computer Science}, pages
  133--146. Springer, 1979.

\bibitem{CoppoDZ87}
Mario Coppo, Mariangiola Dezani{-}Ciancaglini, and Maddalena Zacchi.
\newblock Type theories, normal forms and $\mathcal{D}_\infty$-lambda-models.
\newblock {\em Information and Computation}, 72(2):85--116, 1987.

\bibitem{DezaniG01}
M.~Dezani{-}Ciancaglini and E.~Giovannetti.
\newblock From {B}\"ohm's theorem to observational equivalences: an informal
  account.
\newblock {\em Electr. Notes Theor. Comput. Sci.}, 50(2):83--116, 2001.

\bibitem{GianantonioFH99}
Pietro~Di Gianantonio, Gianluca Franco, and Furio Honsell.
\newblock Game semantics for untyped $\lambda\beta\eta$-calculus.
\newblock In {\em Typed Lambda Calculi and Applications}, volume 1581 of {\em
  Lecture Notes in Computer Science}, pages 114--128. Springer, 1999.

\bibitem{Given-WilsonJ11}
Thomas Given{-}Wilson and Barry Jay.
\newblock A combinatory account of internal structure.
\newblock {\em J. Symb. Log.}, 76(3):807--826, 2011.

\bibitem{GouyTh}
Xavier Gouy.
\newblock {\em {\'E}tude des th\'eories \'equationnelles et des propri\'et\'es
  alg\'ebriques des mod\`eles stables du $\lambda$-calcul}.
\newblock Th\`ese de doctorat, Universit\'e de Paris~7, 1995.

\bibitem{Hyland75}
Martin Hyland.
\newblock A survey of some useful partial order relations on terms of the
  $\lambda$-calculus.
\newblock In {\em Lambda-Calculus and Computer Science Theory}, volume~37 of
  {\em Lecture Notes in Computer Science}, pages 83--95. Springer, 1975.

\bibitem{Hyland76}
Martin Hyland.
\newblock A syntactic characterization of the equality in some models for the
  $\lambda$-calculus.
\newblock {\em Journal London Mathematical Society (2)}, 12(3):361--370,
  1975/76.

\bibitem{IntrigilaMP17}
Benedetto Intrigila, Giulio Manzonetto, and Andrew Polonsky.
\newblock Refutation of {S}all{\'{e}}'s longstanding conjecture.
\newblock In Dale Miller, editor, {\em 2nd International Conference on Formal
  Structures for Computation and Deduction, {FSCD} 2017, September 3-9, 2017,
  Oxford, {UK}}, volume~84 of {\em LIPIcs}, pages 20:1--20:18. Schloss Dagstuhl
  - Leibniz-Zentrum fuer Informatik, 2017.

\bibitem{IntrigilaN03}
Benedetto Intrigila and Monica Nesi.
\newblock On structural properties of $\eta$-expansions of identity.
\newblock {\em Inf. Proc. Lett.}, 87(6):327--333, 2003.

\bibitem{IntrigilaS04}
Benedetto Intrigila and Richard Statman.
\newblock The omega rule is {$\Pi_2^0$}-hard in the $\lambda\beta$-calculus.
\newblock In {\em Symposium on Logic in Computer Science {(LICS} 2004)}, pages
  202--210. {IEEE} Computer Society, 2004.

\bibitem{IntrigilaS09}
Benedetto Intrigila and Richard Statman.
\newblock The omega rule is {$\Pi^1_1$}-complete in the
  $\lambda\beta$-calculus.
\newblock {\em Logical Methods in Computer Science}, 5(2), 2009.

\bibitem{JacobsR97}
Bart Jacobs and Jan Rutten.
\newblock A tutorial on (co)algebras and (co)induction.
\newblock {\em EATCS Bulletin}, 62:62--222, 1997.

\bibitem{KozenS17}
Dexter Kozen and Alexandra Silva.
\newblock Practical coinduction.
\newblock {\em Mathematical Structures in Computer Science}, 27(7):1132--1152,
  2017.

\bibitem{Lassen99}
S{\o}ren~B. Lassen.
\newblock Bisimulation in untyped lambda calculus: B{\"{o}}hm trees and
  bisimulation up to context.
\newblock {\em Electr. Notes Theor. Comput. Sci.}, 20:346--374, 1999.

\bibitem{Levy78}
Jean-Jacques L{\'{e}}vy.
\newblock Approximations et arbres de {B}\"ohm dans le lambda-calcul.
\newblock In Bernard Robinet, editor, {\em Lambda Calcul et S\'emantique
  formelle des langages de programmation, Actes de la 6\`eme \'Ecole de
  printemps d'Informatique th\'eorique, La Ch\^atre}, LITP-ENSTA, pages
  239--257, 1978.
\newblock (In French).

\bibitem{SalibraL04}
Stefania Lusin and Antonino Salibra.
\newblock The lattice of $\lambda$-theories.
\newblock {\em Journal of Logic and Computation}, 14(3):373--394, 2004.

\bibitem{Manzonetto09}
Giulio Manzonetto.
\newblock A general class of models of $\mathcal{H}^*$.
\newblock In Rastislav Kr{\'{a}}lovic and Damian Niwinski, editors, {\em
  Mathematical Foundations of Computer Science 2009}, volume 5734 of {\em
  LNCS}, pages 574--586. Springer, 2009.

\bibitem{ManzonettoR14}
Giulio Manzonetto and Domenico Ruoppolo.
\newblock Relational graph models, {T}aylor expansion and extensionality.
\newblock {\em ENTCS}, 308:245--272, 2014.

\bibitem{Mogensen92}
Torben~{\AE}. Mogensen.
\newblock Efficient self-interpretations in lambda calculus.
\newblock {\em J. Funct. Program.}, 2(3):345--363, 1992.

\bibitem{Morristh}
James~H. Morris.
\newblock {\em Lambda calculus models of programming languages}.
\newblock PhD thesis, Massachusetts Institute of Technology, 1968.

\bibitem{Paolini08}
Luca Paolini.
\newblock Parametric $\lambda$-theories.
\newblock {\em Theoretical Computer Science}, 398(1):51 -- 62, 2008.

\bibitem{Plotkin74}
Gordon~D. Plotkin.
\newblock The lambda-calculus is $\omega$-incomplete.
\newblock {\em Journal of Symbolic Logic}, 39(2):313--317, 1974.

\bibitem{Polonsky11}
Andrew Polonsky.
\newblock {Axiomatizing the Quote}.
\newblock In Marc Bezem, editor, {\em Computer Science Logic (CSL'11) - 25th
  International Workshop/20th Annual Conference of the EACSL}, volume~12 of
  {\em Leibniz International Proceedings in Informatics (LIPIcs)}, pages
  458--469, Dagstuhl, Germany, 2011. Schloss Dagstuhl--Leibniz-Zentrum fuer
  Informatik.

\bibitem{RonchiP04}
Simona Ronchi Della~Rocca and Luca Paolini.
\newblock {\em The Parametric $\lambda$-Calculus: a Metamodel for Computation}.
\newblock EATCS Series. Springer, Berlin, 2004.

\bibitem{Salle1978}
Patrick Sall\'e.
\newblock Une extension de la th\'eorie des types en $\lambda$-calcul.
\newblock In Giorgio Ausiello and Corrado B{\"o}hm, editors, {\em Automata,
  Languages and Programming: Fifth Colloquium, Udine, Italy, July 17--21,
  1978}, pages 398--410. Springer Berlin Heidelberg, 1978.

\bibitem{Scott72}
Dana~S. Scott.
\newblock Continuous lattices.
\newblock In Lawvere, editor, {\em Toposes, Algebraic Geometry and Logic},
  volume 274 of {\em Lecture Notes in Mathematics}, pages 97--136. Springer,
  1972.

\bibitem{Severi2002}
Paula Severi and Fer-Jan de~Vries.
\newblock An extensional {B\"o}hm model.
\newblock In Sophie Tison, editor, {\em Rewriting Techniques and Applications:
  13th International Conference, RTA 2002}, pages 159--173. Springer Berlin
  Heidelberg, 2002.

\bibitem{SeveridV17}
Paula Severi and Fer-Jan de~Vries.
\newblock The infinitary lambda calculus of the infinite $\eta$-{B\"{o}}hm
  trees.
\newblock {\em Mathematical Structures in Computer Science}, 27(5):681--733,
  2017.

\bibitem{Wadsworth76}
Christopher~P. Wadsworth.
\newblock The relation between computational and denotational properties for
  {S}cott's $\mathcal{D}_{\infty}$-models of the lambda-calculus.
\newblock {\em {SIAM} Journal of Computing}, 5(3):488--521, 1976.

\end{thebibliography}

\end{document}